\documentclass{article}
\usepackage[utf8]{inputenc}
\usepackage{graphicx}
\usepackage{authblk}
\usepackage{amssymb}
\usepackage{amsthm}
\usepackage{url}
\usepackage{amsmath}
\usepackage{textcomp}
\usepackage{xcolor}
\usepackage{tikz} 
\usepackage{tkz-graph}
\usepackage{geometry}
\usepackage{float}
\usepackage{caption}
\usepackage{subcaption}

\newtheorem{prop}{Proposition}
\newtheorem{lem}{Lemma}
\newtheorem{theo}{Theorem}
\newtheorem{remark}{Remark}
\newtheorem{definition}{Definition}

\title{Generating Boolean functions on totalistic automata networks}
\author[1,2]{Eric Goles}
\author[1]{Andrew Adamatzky}
\author[2]{Pedro Montealegre}
\author[3,4]{Martín Ríos-Wilson}

\affil[1]{Unconventional Computing Laboratory, University of the West of England, Bristol, UK.}
\affil[2]{Facultad de Ingenier\'{i}a y Ciencias, Universidad Adolfo Ib\'{a}\~{n}ez, Santiago, Chile.}
\affil[3]{Departamento de Ingeniería Matemática, FCFM, Universidad de Chile, Santiago, Chile.}
\affil[4]{ Aix Marseille Univ, Université de Toulon, CNRS, LIS, Marseille, France.}
\date{}

\begin{document}

\maketitle

\begin{abstract}
\noindent
We consider the problem of studying the simulation capabilities of the dynamics of arbitrary networks of finite states machines. In these models, each node of the network takes two states 0 (passive) and 1 (active). The states of the nodes are updated in parallel following a local totalistic rule, i.e., depending only on the sum of active states. Four families of totalistic rules are considered: linear or matrix defined rules (a node takes state 1 if each of its neighbours is in state 1), threshold rules (a node takes state 1 if the sum of its neighbours exceed a threshold), isolated rules (a node takes state 1 if the sum  of its neighbours equals to some single number) and interval rule (a node takes state 1 if the sum of its neighbours belong to some discrete interval). We focus in studying the simulation capabilities of the dynamics of each of the latter classes. In particular,  we show that  totalistic automata networks governed by matrix defined rules can only implement constant functions and other matrix defined functions. In addition, we show that t by threshold rules can generate any monotone Boolean functions. Finally, we show that networks driven by isolated and the interval rules exhibit a very rich spectrum of boolean functions as they can, in fact, implement any arbitrary Boolean functions. We complement this results by studying experimentally the set of different Boolean functions generated by totalistic rules on random graphs.  \\
\vspace{3mm}
\end{abstract}

 Corresponding author: Eric Goles, Facultad de Ingeniería y Ciencias, Universidad Adolfo Ibáñez, eric.chacc@uai.cl

\section{Introduction}

Unconventional computing aims to uncover principles of information processing in chemical, physical and living computing substrates~\cite{adamatzky2016advances}. A predominant majority of living systems are comprised of two key networks: vascular system (metabolites transfer and processing) and nervous or other signalling systems (information transfer and processing). Both types of networks are conducive to propagation of electrical~\cite{hodgkin1952propagation,fromm2007electrical}, mechanical/sound~\cite{heimburg2005soliton,shrivastava2014evidence,fichtl2016protons} and optical~\cite{contreras2013non,appali2012comparison} signals. When these electrical, mechanical or optical solitons interact with each other while travelling and colliding on the networks they change velocity vectors, states or existence. Thus by encoding Boolean values in presence and absence of solitons one can implement logical circuits on the networks. 

An idea to implement a computation by using collisions of signals travelling along one-dimensional non-linear geometries can be traced back to the mid 1960s when Atrubin developed a chain of finite-state machines executing multiplication~\cite{atrubin1965one}, Fisher designed prime numbers generators in cellular automata~\cite{fischer1965generation} and Waksman proposed the eight-state solution for a firing squad synchronisation problem~\cite{waksman1966optimum}. 
In 1986, Park, Steiglitz, and Thurston~\cite{park1986soliton} designed a parity filter in cellular automata with soliton-like dynamics of localisations. Their design led to a construction of a one-dimensional particle machine, which performs the computation by colliding particles in one-dimensional cellular automata, i.e. the computing is embedded in a bulk media~\cite{squier1994programmable}. In 1990s Goles and colleagues demonstrated that sand pile and chip firing game are universal computers. That is by representing logical truth by presence of a sand grain or a chip and logical false by absence of the grain/chip one route information as avalanches and implement logical gates via interaction of avalanches in an appropriate geometrical structure~\cite{goles1996sand,gajardo2006crossing,goles1997universality}.

Most close to biophysical reality models of computing on biological networks have been implemented with molecular structures of verotoxin protein~\cite{adamatzky2017computing} and actin monomer~\cite{
adamatzky2017logical}, actin bundles networks derived from experimental laboratory data~\cite{adamatzky2019computing}, plant leaf vascular system~\cite{adamatzky2019plant} and microscopic images of three-dimensional fungal colonies~\cite{adamatzky2020boolean}. These approaches employed the following method. Two loci (atoms, molecules, parts of the network) are considered to be input and all other loci outputs. Sequences (01), (10), (11), represented by solitons or impulses are sent to the inputs and solitons/impulses are recorded on the outputs. Each of the outputs implements a two-input-one-output logical gate. In the computational experiments with the molecular, polymer, vascular or mycelial networks \cite{adamatzky2017computing,adamatzky2019computing,adamatzky2019plant,adamatzky2020boolean} we did not analyse where exactly in the networks the computation takes place. In this context, and in order to fill the gap in our knowledge,  we study the generation of Boolean functions on totalistic automaton networks, i.e., each site changes state according to specific values of the sum of active sites in its neighbourhood. From a mathematical point of view, totalistic rules are well-known models in the context of the study of cellular automata and automata networks as dynamical systems \cite{wolfram1984universality,marr2009outer}. In addition, the approach of studying the computational complexity of some specific decision problems that are some how related to the dynamical behaviour of totalisitc automata network has been also broadly studied, particularly in regards to how difficult is to predict the dynamical behaviour of some specific entity in the network \cite{goles2014computational,goles2018complexity,goles2020complexity}. Although,  as from a theoretical viewpoint, computational complexity has been proposed as an approach for somehow measure the complexity of the dynamics of an specific automata network model, the simulation capabilities of the network, in the sense of the experiments performed in the previous context, has remained unexplored. In this paper, we characterise the complexity of totalistic automata networks  according to its capabilities in order to simulate Boolean gates by the automaton's dynamics. Particularly, we explore this approach from a mathematical point of view, by systematizing and formalizing the spectrum of an automata network as a measure of the different boolean networks that the system is able to implement. In this context we proved that linear or matrix rules generate few different Boolean networks:only constant or matrix defined ones. Threshold totalistic can generate any monotone Boolean Function and the isolated or interval ones, every Boolean function. Further, we study by computational experiments the generation of Boolean functions capabilities in random totalistic networks.

\section{Preliminaries}
\label{Preliminaries}

An automata network  is a tuple $ \mathcal{A}= (G,Q,\mathcal{F})$ where $G=(V,E)$ is an undirected finite graph such that $|V|=n,$ $Q=\{0,1\}$ is the finite set of states  called $\textit{alphabet}$ and $\mathcal{F} = \{f_v: v \in V\}$ is a collection of functions called \textit{local functions} such that  each local function $f_v:N_v \to Q$ takes it arguments as a neighbourhood of $v$ in $G$ given by $N_v = \{u\in V: uv \in E\}.$ We define  $F:Q^n \to Q^n$ as the \textit{global transition function} of the automata network defined by $F(x)_v = f_v(x|_{N_v})$ for all $x \in Q^n$ and for all $v \in V$ where $x|_{N_v}$ are the coordinates of $x$ that are representing the neighbours of $v$. Formally, $x|_{N_v} \in Q^{N_v}$ and for all $u \in N_v$ we have $(x|_{N_v})_u = x_u.$  We consider only the case where every local function $f_v$ is \textit{totalistic}, i.e., a result of the function depends on the sum of the active values (states $1$'s ) as its arguments. Suppose that the maximun degree  $\Delta(G)$ on $G$ is such that $\Delta(G) = \Delta$, so the sum may take values in the set
$\{1, ..., \Delta\} $, i.e,  given a configuration $x\in Q^{n}$ we have  $f_v(x|_{N_v})= 1$  if and only if   $\sum \limits_{u\in N_v} x_u \in \mathcal{I}_v = \{a_1, ..,a_s\}$, where   $\{a_1, ..,a_s\}\subseteq \{1, ..., \Delta\}.$ We will call the set $\mathcal{I}_v$ the activation set of $f_v$.
We refer to each local rule by the digit associated to  $\mathcal{I}_v $,  so if  $\mathcal{I}_v = \{a_1,\hdots, a_s\}$ then, the  totalistic rule number will be $a_1a_2 \cdots a_s$. For instance, rule $25$ means that the associated vertices become $1$ if and only if the sum if either $2$ or $5$. Note that, as each of these rules depends also on the set of values that each node in some neighborhood will take, we can consider that totalistic rules are defined over the set $S_{\Delta} = \{1,\hdots, \Delta\}$. This is required because in further sections we will work with a fixed set of totalistic rules and use them to define automata networks over different graphs. In this regard,  as we will always work over some class of graphs $\mathcal{G}$ in which every graph have at most degree $\Delta$, we note that each totalistic rule having an activation set $I \subseteq \{1,\hdots,\Delta\}$  will be well defined over any graph in $\mathcal{G}$.  For example, if $\Delta = 10$ rule $25$ will be well-define over any graph in $\mathcal{G}$. More precisely,  for any totalistic function $f$ such that $I_f \subseteq S_{\Delta}$ we can assume that $f: S_{\Delta} \to \{0,1\}$ and that $f(s) = 1$ if and only if $s \in I_f$.

\subsection{The problem}

Let $\mathcal{A} = (G, \mathcal{F})$  be a totalistic automata network with global transition function $F$.  A configuration $\overline{x} \in Q^n$ is a fixed point of $\mathcal{A}$ if $F(\overline{x}) = \overline{x}.$ Note that, by definition,  $\vec{0}$ is always a fixed point of $\mathcal{A}$. Consider the set $\text{Fix}(\mathcal{A})$ of fixed points of $\mathcal{A}$. Since $\vec{0}$ is a fixed point,  $\text{Fix}(\mathcal{A}) \not = \emptyset$. Now we consider two arbitrary disjoint sets of vertices $I = \{i_1, \hdots, i_l\}\subseteq V$ and $O  = \{o_1,\hdots,o_s\} \subseteq V$ for $s,l \geq 1$. We call these sets the  \textit{input}  set of  $\mathcal{A}$ and the  \textit{output} set of  $\mathcal{A}$ respectively. As we are interested in the simulation of logic gates, during the rest of the paper, we will focus on the case of at most two outputs, that is to say $O \subseteq \{o_1,o_2\},$ with special emphasis in the case $O = \{o\}$. 

 Let $t \geq 0$, we call a tuple $(I,o,t ) \in 2^{V} \times V \times \mathbb{N}$ an $I/O$ setting for $\mathcal{A}$. The variable $t$ represents some specific time step in which we are interested to analyze the output of the automata network in $o$. We call $t$ an \textit{observation} time.


Consider now a fixed point $\overline{x} \in \text{Fix} (\mathcal{A})$ with $I/O$ setting $(I,o,t) \in 2^{V} \times V \times \mathbb{N}.$  Let $z = (z_1, ..., z_l) \in \{0,1\}^{l}$ be an assignation of values for the input set $I = \{i_1, ...i_l\}$.   We say that a configuration $y \in Q^n$ is a perturbation of $\overline{x}$ by $z$ if $y_u = \overline{x}_u$ for all $u \in V \setminus I$ and $y_u= z_u$ for $u \in I$. Now, given some observation time $t\geq 0$ we are interested in studying all possible states of the output $o$ after $t$ time steps given different assigments of variables to the inputs. More precisely, given a perturbation $y$ of $\overline{x}$ for the last $I/O$ setting  we define a realisation of $\overline{x}$ as a Boolean function $g^{(I,o,t)}_{\overline{x}}: \{0,1\}^l \to \{0,1\}$ such that $g^{(I,o,t)}_{\overline{x}}(y) = F^t(y)_o.$  Let $\binom{V}{l}$ be a  collection of subsets of $V$ with size $l$. Considering all possible combinations of inputs of certain size $l$ and possible outputs for a fixed time, we define the spectrum of Boolean functions with $l$ inputs generated by $\overline{x}$ at time $t$ as the set of realisations of $\overline{x}$ given by  $\mathbb{F}^{t,l}_{\overline{x}} = \{g^{(I,o,t)}_{\overline{x}}: (I,o) \in 2^V \times V\}$ and the spectrum of functions with $l$ inputs of $\mathcal{A}$ at time $t$ as  the set $\mathbb{F}^{t,l}_{\mathcal{A}} = \bigcup \limits_{\overline{x} \in \text{Fix}(\mathcal{A})} \mathbb{F}_{\overline{x}}$. Previous set may consider the "same" Boolean function but only by changing the choosing input-output indexes of variables. To avoid that, in  previous set we only consider  different Boolean functions, i.e., that is its difer in al least one set of Boolean inputs. On the other hand, the set of realisable Boolean functions depends on the observation time that we are considering and in the interaction graph of the network. Thus, it is strictly related to classical dynamical properties such as transient length and the attractor landscape of the automata network. 
 
\subsection{Totalistic classes of functions}

We recall that $S_{\Delta}$ is the set of natural numbers $\{1,\hdots,\Delta\}$ for some bound $\Delta \geq 2$. One of the advantages of working with totalistic rules is that we can fix a collection of these rules and change the underlying interaction graph in order to define different automata networks. For example, if we fix $n \in \mathbb{N}$ and we consider the set of local functions $\{f_k: S_{\Delta} \to \{0,1\}, k=1,\hdots,n\}$ in which every rule is the rule $1$, i.e., $f_k(s) = 1$ if and only if $s = 1$ for every  $k \in V$, we can, for every graph $G = (V,E)$ with $n$ nodes, define an automata network $\mathcal{A}_G = (G,\{\tilde{f}_k:  N(k) \to \{0,1\}\})$ where $\tilde{f}(x|_{N(k)}) = f(\sum_{v \in N(k)} x_v)$. In further sections, we will write simply $f$ while referring both to $\tilde{f}$ and $f$ to simplify the notation. Note now that if $G$ and $G'$ are two different graphs with $n$ nodes then their corresponding automata networks $\mathcal{A}_G$ and $\mathcal{A}_{G'}$ will have the same kind of global rule (each vertex has the same totalistic function). The degree of the graph $G$ plays a fundamental role in this definition. For example, suppose that some rule $f_k$ needs $l$ active neighbors in order to activate the node, i.e. $f_k(x|_S) = 1$ if and only if $\sum \limits_{v \in S} x_v  = l$ for some, $l \in \mathbb{N}$, for every $S\subseteq V$ and  $x \in \{0,1\}^n$. In addition, suppose that the degree of node $k$ is less than $l$. In that case, $f_k$ will be fixed in the initial configuration. Roughly this will not have an important effect in our results (nor theoretical or numerical) as we will work with connected graphs with bounded maximum degree and, in addition, it will allow us to keep the latter set-up simple. In addition, this way to generated different automata networks from the same class of local functions is more practical in order to study the dynamical behaviour of different totalistic rules for a fixed large collections of randomly generated graphs.

Roughly, in this context we will say that the complexity of $\mathcal{A}$ is given by the number of different Boolean functions that could be generated over any fixed point we consider.  Formally, we define the simulation complexity for a $\mathcal{A}$ by  $\rho(\mathcal{A},t,l) = |\mathbb{F}^{t,l}_{\mathcal{A}}|$. The large is variery of  Boolean functions  generated by the latter procedure, the more complex the automaton will be considered. In this context we will say the complexity of $\mathcal{A}$ is given by the number of different Boolean functions that could be generated over any fixed point we consider. In addition, we can define the complexity of a class of totalistic functions $\mathcal{F}$ related to some class of graph $\mathcal{G}$ as $\rho(\mathcal{F},\mathcal{G},t,l) =|\mathbb{F}^{G,t,l}_{\mathcal{F}}|$

\subsection{The spectrum as a measure of complexity}
\label{spectrum}

Generally speaking, the spectrum of a set or a class of totalistic rules is a measure of how many different types of a Boolean gates it can simulate. Nevertheless, it is well known that every Boolean function $f: \{0,1\}^r \to \{0,1\}^s$ can be represented by a directed graph $C$ in which every node is a Boolean gate. An asynchronous evaluation of every Boolean gate in $C$ performs the evaluation of the function $f$.  In this regard, we are interested in the study sets of totalistic rules  which not only are capable of simulating different Boolean gates, but to organize them in way that, they can simulate the evaluation of a Boolean circuit.

Roughly, our main idea is to show that some class $\mathcal{F}$ is able to simulate a complete set of logic gates, for example, $\textsc{AND}, \textsc{NOT}, \textsc{OR}$. Let $\mathcal{A}_{\textsc{AND}},\mathcal{A}_{\textsc{NOT}}$ and $\mathcal{A}_{\textsc{OR}}$ be the automata networks that simulates each of this gates. Then, we will try to combine its different underlying graphs in order to simulate an arbitrary circuit $C$, by considering for each gate $g$ in c  one of the latter graphs and then, try to connect them somehow. Note that this process is not straightforward as every automata network simulates a logic gate through its dynamics and so, it is not trivial how we should glue them in order to generate coherent global dynamical behaviour.  In simple words, what we want to achieve is, exhibit a large automata network that has a set of small subgraphs simulating logic gates. This big network will simulate the evaluation of $C$ through its dynamics, in the sense that, by identifying a group of nodes as ``input nodes'' we can read the same output we would have read after the evaluation of $C$ by reading the state of another group of nodes labelled as ``output nodes'', after some time steps.  More precisely, we introduce the following definition:


\begin{definition}
	Let $\Delta \in \mathbb{N}$ and let $\mathcal{G}$ be a collection of graphs with maximum degree at most $\Delta$. Let $\mathcal{F}$ be a set of totalistic rules and $f:\{0,1\}^r \to \{0,1\}^s$ an arbitrary Boolean function. We say that $\mathcal{F}$ simulates $f$ in  $\mathcal{G}$  if there exists such $n \in \mathbb{N}$ such that $n = r^{\mathcal{O}(1)}$, a graph $G_n = (V,E) \in \mathcal{G}$ with $|V| = n$ with global rule $F_n$ and  $t = n^{\mathcal{O}(1)}$ such that $f(y) = (F_n^t(x|_I))|_O$ for every $y \in \{0,1\}^r$, for some $x \in \{0,1\}^n$ and for some sets $I,O \subseteq V$ such that $|I|=r$ and $|O|=s$.
\end{definition}

\begin{remark}
	Note that, in the latter definition, the sets $I$ and $O$ are one-to-one related to the input and outputs of the graph. This means the simulation is very strong in the sense that inputs and outputs are represented by one node in the network.
\end{remark}


On the other hand, we note that the latter definition is intrinsically related to the computation complexity of some decision problems that have been studied in order to measure the complexity of the dynamics of automata networks. In particular, it is closely related to \textit{prediction problem}. Given an automata network $\mathcal{A} = (G,\mathcal{F})$, this problem is roughly defined by a given configuration $x \in Q^n$ and a node $v \in V$ for which we would want to know if the state of $v$ will change at some point in the orbit of $x$. More precisely, we ask if there exist $t$ such that $F^t(x)_v \not = x_v.$ Depending on $\mathcal{A}$, it can be shown that the complexity of this decision problem is closely related to the capability of $\mathcal{A}$ of simulating the evaluation of an arbitrary Boolean circuit.  This is because, depending on the rules defining $\mathcal{A}$ this problem can be verified or solved in a polynomial time, and thus the complexity bounds are deduced through a reduction to canonical problems such as \textsc{Circuit value problem} or \textsc{SAT}.


\section{Results}

In this section we present different results on the complexity of different totalistic rules in the sense of its spectrum. In particular, we focus in exhibiting for different classes of totalistic rules,  small automata networks, that we call gadgets, that can simulate logic gates. Then, we show how we can combine them in order to simulate arbitrary Boolean functions. In this regard, we present a classification based in the structure of the activation sets $\mathcal{I}_v$ of the different totalistic rules in $\mathcal{F}$. We start by studying the simple case where $\mathcal{I}_v = \{1\}$. This relates to disjunctive and conjunctive networks, and, in a more general way, rules which dynamics are defined by a some sort of matrix product. Then, we study the classic case in which $\mathcal{I}_v$ is given by some interval $[\theta_v,\Delta] \subseteq \{0,\hdots, \Delta\}.$ This class includes the well-known threshold networks which, as we show in this section,  have the capability of simulating any monotone Boolean network.  In particular, we explore the case in which isolated activation values are considered. Roughly we explore the case in which if $a \in \mathcal{I}_v$ then $a-2,a-1,a+1 \not \in \mathcal{I}_v$ and we find that this class is also capable of simulating arbitrary Boolean networks. Finally we study the intermediate case in which $\mathcal{I}_v = \{\alpha,\hdots, \beta\}$ with $\alpha \leq\beta \leq \Delta$. Particularly, we are interested in studying automata networks $\mathcal{A} = (G,\mathcal{F})$ in which for any $v \in V(G)$ its local function is such that $\mathcal{I}_v = [\alpha,\beta]$ with $\beta < \delta_v$, i.e. there exist a threshold that deactivates the local function for each node. We illustrate the dynamics of each subclass of threshold networks in Figure \ref{fig:threshold}.
\begin{figure}[!tbp]	
	\centering
	\begin{subfigure}[t]{3in}
		\centering
		\includegraphics[scale=0.15]{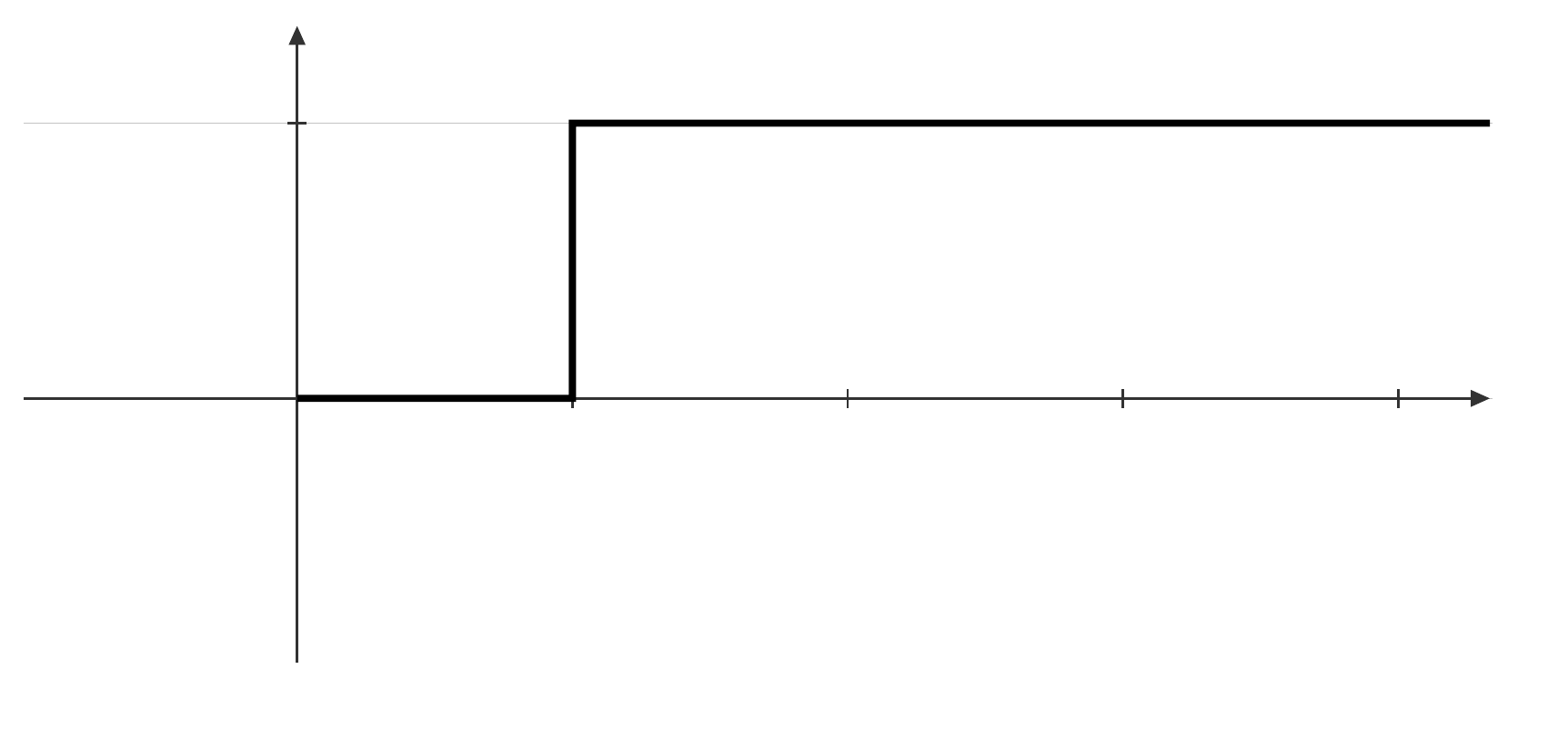}
		\caption{Threshold rules}\label{fig:1a}		
	\end{subfigure}
	\quad
	\begin{subfigure}[t]{3in}
		\centering
		\includegraphics[scale=0.15]{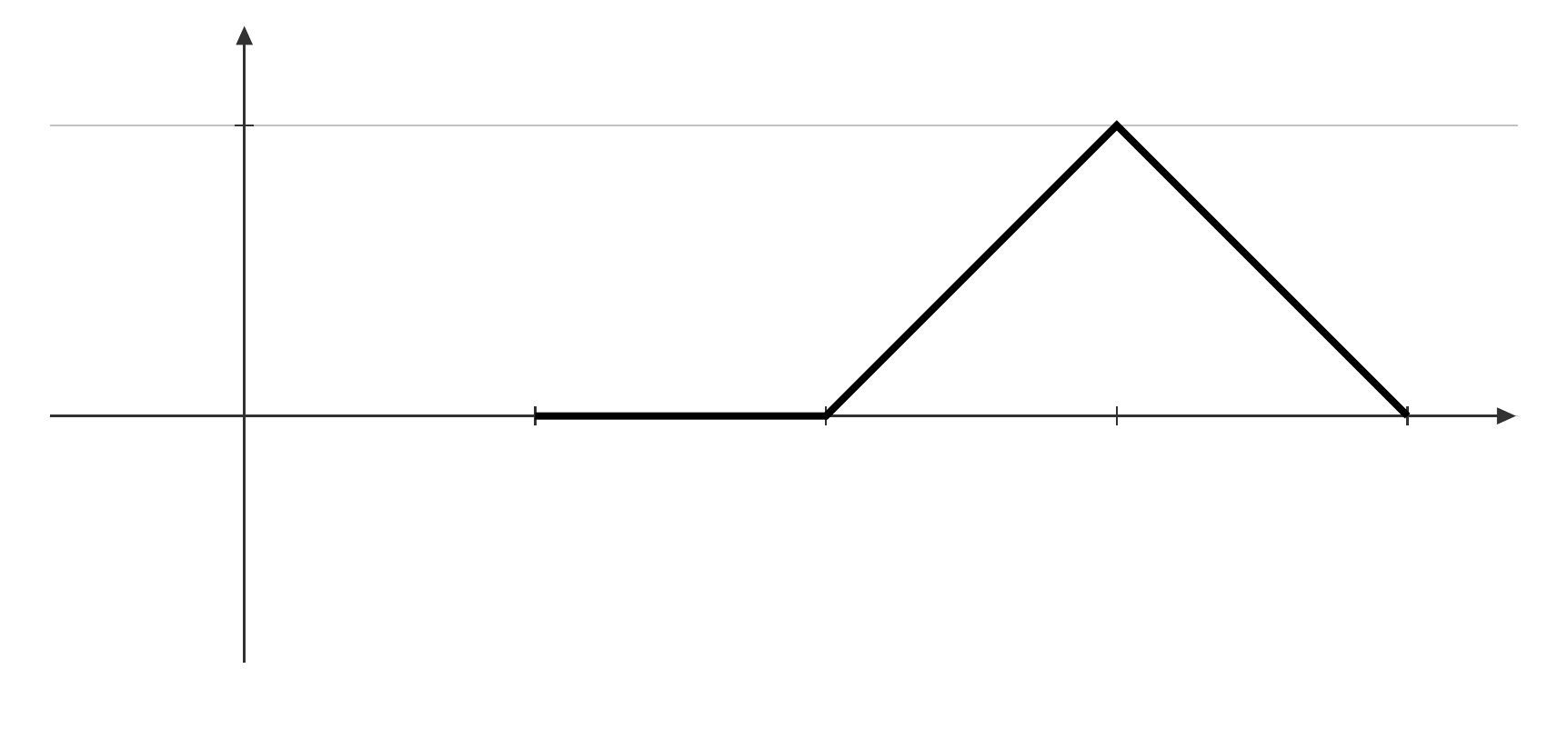}
		\caption{Isolated rules}\label{fig:1b}
	\end{subfigure}
	\begin{subfigure}[t]{3in}
	\centering
	\includegraphics[scale=0.15]{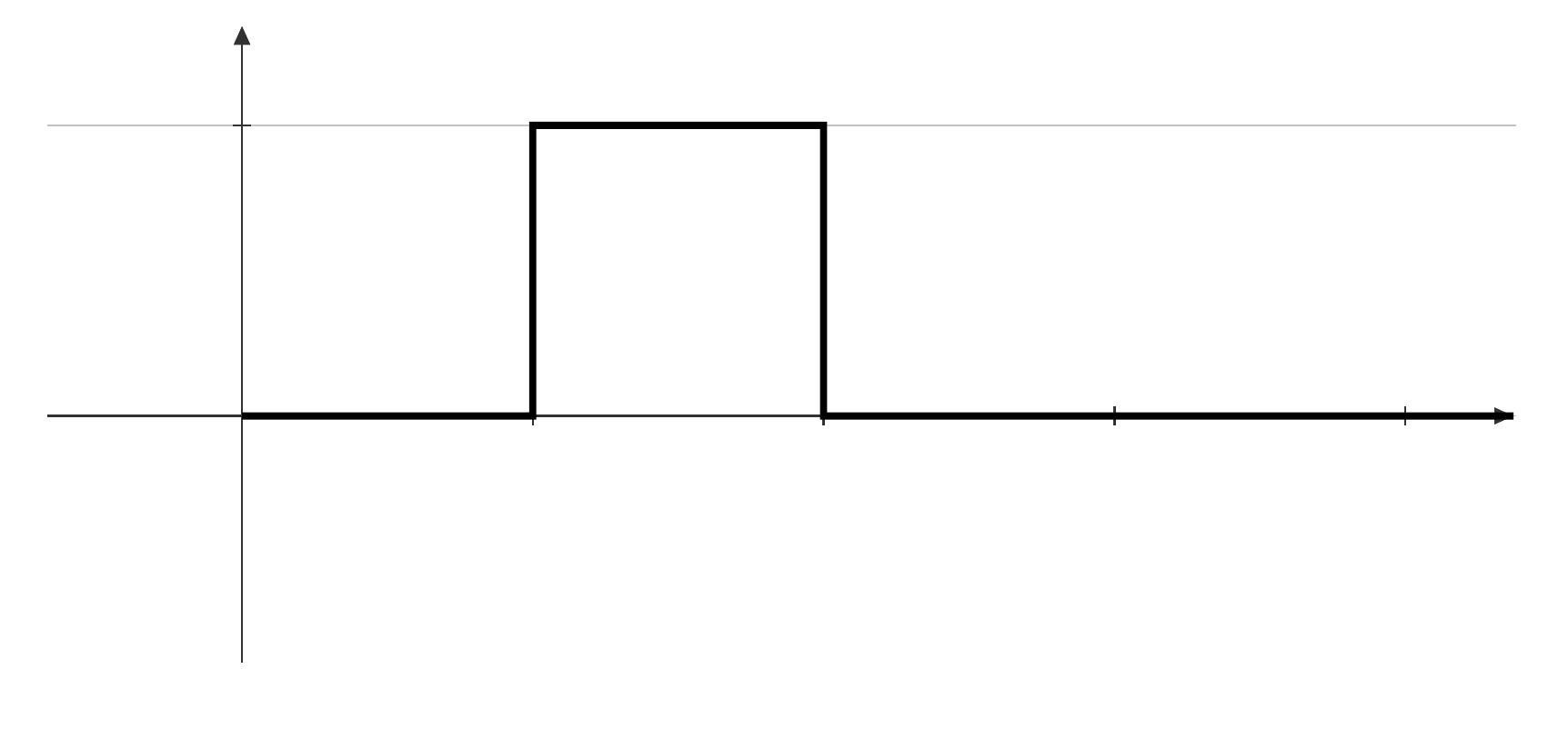}
	\caption{Interval rules}\label{fig:3b}
\end{subfigure}
	\caption{Different classes of threshold totalistic rules}\label{fig:threshold}
\end{figure}
We resume our main results in the following table:
\begin{table}[H]
	\resizebox{\textwidth}{!}{%
		\begin{tabular}{|l|l|}
			\hline
			\textbf{Class of Totalistic Rules} &\textbf{Simulation capabilities}                               \\ \hline
			Matrix-defined rules      & Constant functions and other matrix-defined functions. \\ \hline
			Threshold rules           & Arbitrary monotone Boolean functions.                  \\ \hline
			Isolated rules            & Arbitrary Boolean functions.                           \\ \hline
			Interval rules            & Arbitrary Boolean functions.                           \\ \hline
		\end{tabular}%
	}
	\caption{Classification of totalistic rules according to their simulation capabilities.}
	\label{tab:summary}
\end{table}
\subsection{Matrix-defined rules}
We start by studying canonical cases of totalistic rules such as disjunctive (conjunctive) networks. Let $Q =  \{0,1\}$ and let $\mathcal{G}$ be a family of graphs. We say that some totalistic rule $f:S_{\Delta}\to Q$ is disjunctive if it takes the value $1$ if and only if there exist at least one $1$ in its assignment, i.e.,  $\mathcal{I}_f = \{1\}$. We say that a set or a class of totalistic rules $\mathcal{F}$ is disjunctive if every $f \in \mathcal{F}$ is disjunctive. Analogously, we can define a conjunctive totalistic function $f$ over a graph $G$ in some node $v \in V(G)$ by defining the transition to $1$ only in the case in which every neighbour of $v$ is in state $1$, i.e  $\mathcal{I}_v= \{\delta_v\}$. Note that in this case we cannot define the rule independently of the interaction graph. Nevertheless, both rules are completely analogous as it suffices to change the role of $1$ and $0$ in order to change from disjunctive to conjunctive and vice versa. As a consequence of these, and in order to simplify following reasoning we focus on disjunctive rules but of course all of the next results are valid also for conjunctive rules.

\begin{lem}
	Let $\mathcal{G}$ be an arbitrary family of graphs and take some graph $G \in \mathcal{G}$. Let $\mathcal{D} = (G,\mathcal{F})$ be an automata network where $\mathcal{F}$ is disjunctive then the spectrum of $\mathcal{D}$ contains only constant gates (everything goes to $1$ or $0$) and disjunctive gates (OR gates). For any $t,l\geq 1$ we have $\mathbb{F}^{(t,l)}_{\mathcal{D}}  \subseteq \{0,1\} \cup \{\vee_J\}_{J \subseteq I} $ where $\vee :\{0,1\}^l \to \{0,1\}$ is such that $\vee (z_1,\hdots,z_l) = \bigvee \limits^l_{k=1} z_l$  for $J \subseteq I.$
\end{lem}
\begin{proof}
	Let $G \in \mathcal{G}$, $t,l\geq 1$ and $\mathcal{A} = (G,\mathcal{D}(G))$. Let $F$ be a global transition function of $G$. Fix an input $I \subset V.$ Note that $\text{Fix}(\mathcal{A}) = \{\vec{0},\vec{1}\}.$ Also note that there exists a matrix $A \in M_n(\{0,1\})$ such that $F^t(x) = A^t \vee x = \bigvee \limits_{i \in N_v(G^t) \cap I} x_i $ for all $t \geq 0$. In particular, $A$ is the adjacency matrix of $G$. Note that for every $i \in I$ there exists a path between $i$ and $o$ of length $t$ if and only if $(A^t)_{io} = 1$ and thus, the $t$-th power of $A$ define the power graph $G^t$. By definition we have that 
	\begin{equation}
	F^t(x)_o = (A^tx)_o = \bigvee \limits_{i \in N_v(G^t)} x_i = \left( \bigvee \limits_{i \in N_v(G^t) \cap I} x_i \right) \vee \left( \bigvee \limits_{i \in N_v(G^t) \cap V \setminus I} x_i \right).
	\label{eq:or}
	\end{equation} 
	Now, note that, if we start perturbing $\vec{1}$ then we have that $\mathbb{F}^{(G,t,l)}_{\vec{1}} \subseteq \{1,\vee\}$ as the only case in which we can do something different than $1$ is when $l = \delta_o$, $I = N_o$ and $t = 1$.
	
	On the other hand, as $G$ is connected, consider $P_1, \hdots, P_l$ as all the minimum length paths connecting each node in $I$ to $o$. Let $d_1,\hdots, d_l$ be the lengths of each of the path and let $d = \min \limits_{i \in \{1,\hdots,l\}} d_i$ and $D = \max \limits_{i \in \{1,\hdots,l\}} d_i.$ If we perturb $\vec{0}$, we have that, for $t\leq d$ where we have $g^{I,o,t} \equiv 0.$ For $t \geq d$ we can have that not all the nodes in $I$ are connected in $G^t$ with $o$ and thus by (\ref{eq:or}) we have that $g^{I,o,t} \equiv \vee_{J}$ for some $J \subseteq I$. Finally if $t \geq D$ then, we can have influence of  external nodes in $N_v(G^t) \cap V\setminus$.  The influence is given by an OR function. Therefore we have two possible cases: a) one external have state $1$ at time step $t$ and then  $ g^{I,o,t} \equiv 1$ or all stay in state $0$ and then there is no influence. In every case we conclude that $\mathbb{F}^{(t,l)}_{\vec{0}} \subseteq \{0,1,\vee\} \cup \{\vee_J\}_{J \subseteq I}$ and thus the lemma holds.

%

%
%
\label{lemma:OR}
\end{proof}

\begin{remark}
	Note that if $G= (V,E)$ is such that every totalistic function takes the value $1$ when the sum of the states of all neighbours of certain vertex is odd, then, we have the known XOR rule. More precisely, we have the XOR rule if for every $v \in V$ we have that $\mathcal{I}_v = \{a \in \{0,\hdots,\delta_v\}: a \text{ is odd} \}$. Also the global rule of that automata network in that case can be seen as a matrix product, i.e. $F^t(x) = A^t x$ where the product is the usual product in $\mathbb{F}_2.$ Thus, the previous result holds for XOR rules. 
\end{remark}
%
%
%
\subsection{Threshold networks}
In this section we introduce a class of  totalistic functions called threshold functions. Roughly, in this family we have that a function takes the value $1$ if the sum of the states of the neighbours of the corresponding vertex is in some interval $[\theta_v,\delta_v]$ where $\delta_v$ is the degree of the vertex $v$ that we are considering and $\theta_v$ is some positive threshold. We present this notion in the following definition:
\begin{definition}
	A totalistic function $f: S_{\Delta} \to \{0,1\}$ is a threshold if there exists some positive integer $\theta$ such that $ [\theta,\Delta] \subseteq \mathcal{I}_f $, where $I_f$ is the activation set of $f$.
\end{definition}
  We denote by $\mathcal{T}$ the class of all totalistic functions that are threshold, i.e. $f \in \mathcal{T}$ if and only if $f$ is threshold. We will show that there exist a class of graphs $\mathcal{G}$ for which $\mathcal{T}$ simulates any monotone Boolean function. 
\begin{lem}
	There are two automata networks $\mathcal{A}_2 = (G_1=(V_1,E_1),\mathcal{F}_1)$ and  $\mathcal{A}_2 = (G_2=(V_2,E_2),\mathcal{F}_2)$ with global rules $F_1$ and $F_2 $ respectively, such that $\mathcal{F}_i \in \mathcal{T}$ and that :
	\begin{enumerate}
		\item $\wedge(x,y) = F^2_1(z_1)_{o} = F^2_1(z_1)_{o'}.$
		\item  $\vee(x,y) = F^2_2(z_2)_{o} = F^2_2(z_2)_{o'},$
	\end{enumerate}
  for some $z_i \in \{0,1\}^{|V_i|},$ $i=1,2.$
  \label{lem:andorthreshold}
\end{lem}
\begin{proof}
	Consider the graph $G_1$ and $G_2$ given in Figure \ref{fig:andthrehsold} and Figure \ref{fig:orthrehsold}. Observe that $\overline{z}_1 = (0,0,0,0,0)$ and $\overline{z}_2 = (0,0,0,0,0)$ are fixed points for $F_1$ and $F_2$ respectively. Then, we can define $z_1$ and $z_2$ as a perturbation of these fixed points as it is shown in Figures  \ref{fig:andthrehsold} and \ref{fig:orthrehsold}. As we stated in the last section, it suffices to define $\theta_v \in \{1, \delta_v\}$ in order to define an AND or an OR function. Observe that this is exactly the threshold defined for each vertex in Figure \ref{fig:andthrehsold} and Figure \ref{fig:orthrehsold}. The result follows from the calculations in latter figures. 
\end{proof}

\begin{figure}[!tbp]
\centering
	\includegraphics[scale=0.5]{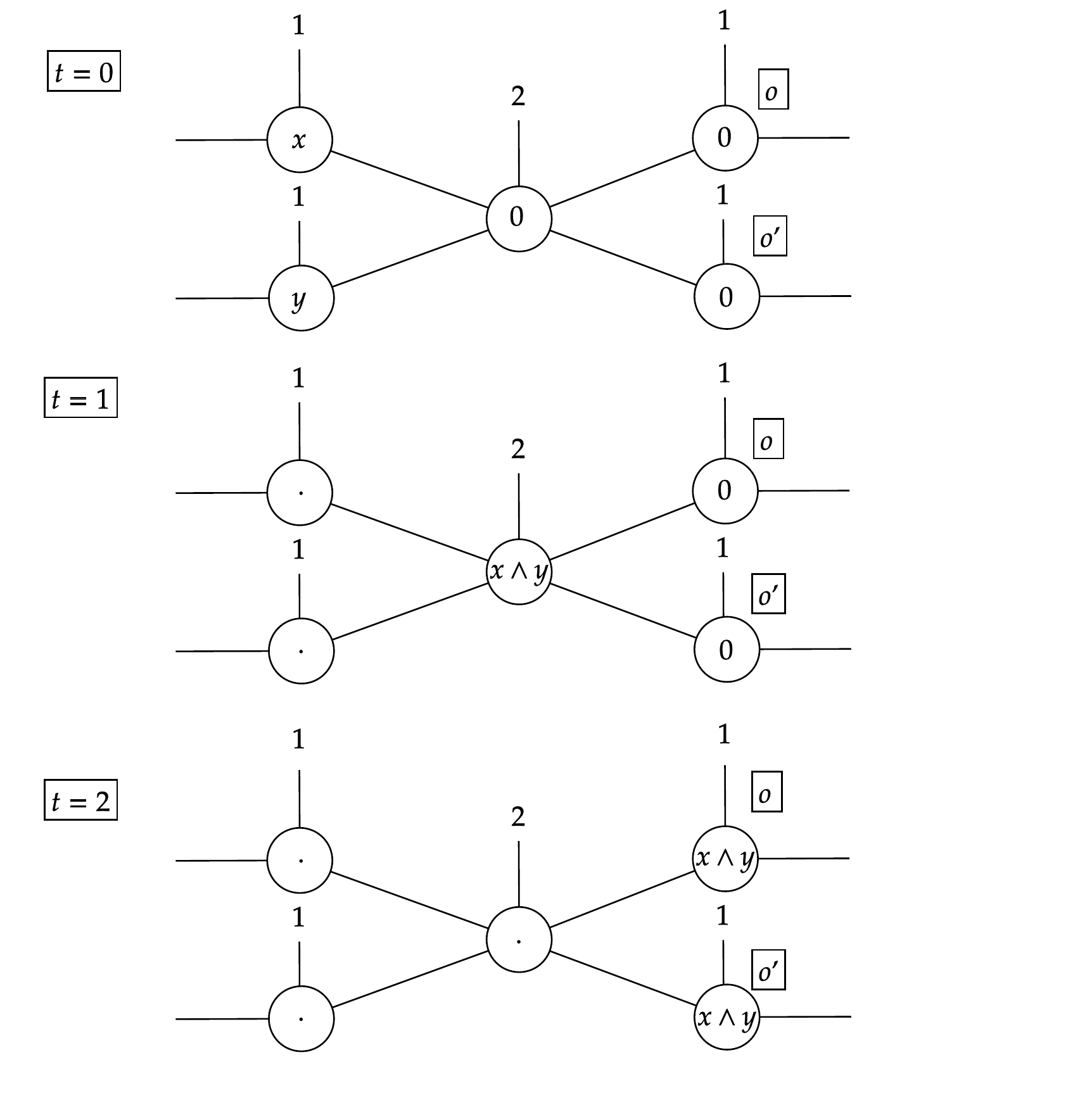}
		\caption{AND gadget for the class of threshold totalistic functions.}
		\label{fig:andthrehsold}
\end{figure}

\begin{figure}[!tbp]
\centering
	\includegraphics[scale=0.5]{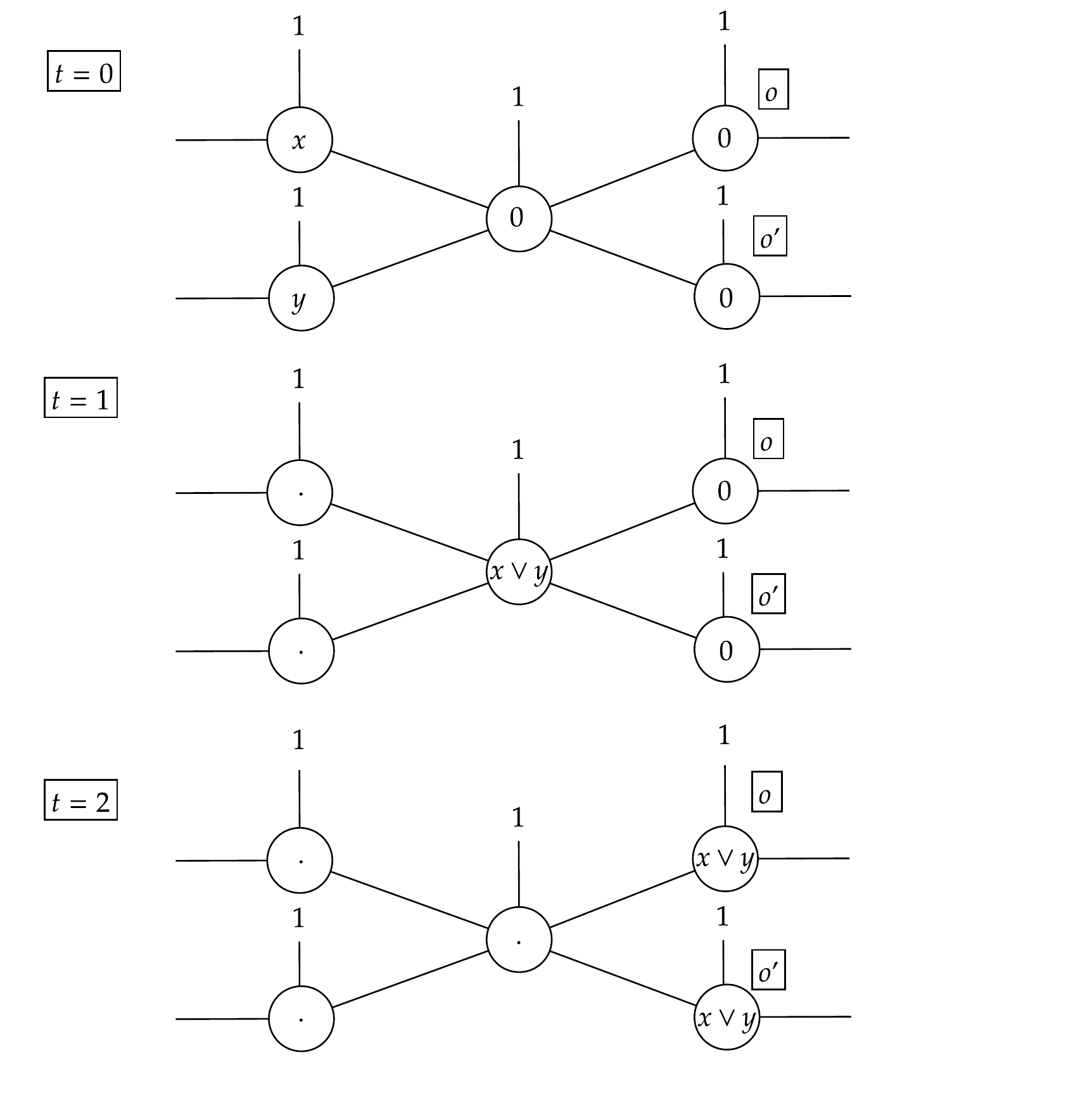}
	\caption{OR gadget for the class of threshold totalistic functions.}
	\label{fig:orthrehsold}
\end{figure}
 \begin{theo}
 	Let $r,s \in \mathbb{N}$ and  $f:\{0,1\}^r \to \{0,1\}^s$ be a monotone Boolean function. There exist a collection of graphs $\mathcal{G}$ such that $\mathcal{T}$ simulates $f$ in $\mathcal{G}$.
 	\label{teo:thershold}
 \end{theo}
\begin{proof}
Fix $r,s \in \mathbb{N}$ and $f: \{0,1\}^r \to \{0,1\}^s$ an arbitrary function. It is well known that $f$ can be represented by a Boolean circuit $C_f:\{0,1\}^{r} \to \{0,1\}^{s}$. More precisely, for every variable assignment $y \in \{0,1\}^r$ the evaluation of the circuit computes $f(z)$. In addition, it suffices to consider bounded fanin and fanout circuits (more specifically we can always assume fanin and fanout $2$ for all gates, with the exception of input and output gates) and vertex set can always be considered as partitioned in layers (see  \cite[Section 6.2]{greenlaw1995limits}). Each layer is defined by the length of longest path connecting a gate to an input gate. We are going to show that there exist some $t \in \mathbb{N}$, a graph $G=(V,E)$ and a set of threshold rules $\mathcal{F} = \{f_v:\{0,1\}^{N(v)} \to \{0,1\}\}$ defining an automata network $\mathcal{A}_f$ such that its associated global rule $F$ is such that $f(y) = F(x|_I)|_O$ for some sets nodes $I,O \subseteq V$ and some $x$ depending on $y$. Let $D_f$ the digraph defining circuit $C_f$. Without loss of generality, we can assume that $C_f$ is monotone, i.e., any gate computes only an AND or an OR gate. In other words, any node $v \in V(D_f)$ is labelled by a symbol $l(d) \in \{\wedge, \vee\}$ which represents the corresponding gate in the circuit. We define $G$ in the following way: for each $v \in V(D_f)$ that is not an input we assign one of the gadgets $\varphi(v)$ in Figure \ref{fig:andthrehsold} or  Figure \ref{fig:orthrehsold} according to $l(v)$. For input gates we consider input nodes of gadgets representing gates in the first layer. Note that in order to represent output gates it is sufficient to consider output nodes in some gadget given by Figure \ref{fig:andthrehsold} or  Figure \ref{fig:orthrehsold}. In addition, we can assume that $\delta^{+}_v = \delta^{-}_v = 2$. Note also that $\varphi(v)$ has two possible outputs $o$ and $o'$ in Figure \ref{fig:andthrehsold}. We define edges  in $G$ locally by the connections in each gadget $\varphi(v)$ for each $v \in V(D_f)$ and also we identify the output of gadget $\varphi(v)$ with one of the inputs of gadget $\varphi(v')$ if $v' \in N^{-}(v)$. Note that $|V(G)| \leq \sum \limits_{v \in V(D_f)} |\varphi(v)| = 5 |V(D_f)| = r^{\mathcal{O}(1)}.$ From previous lemma we now that $\varphi(v)$ computes $\wedge(x,y)$ or $\vee(x,y)$ where $x,y$ are its inputs. We know also that it is done in a uniform time $t=2$ and there is also two possible choices for the outputs $o$ and $o'$ which receive the signal carrying the result of the computation  at the same time. We define now an automata network $(G,\mathcal{F})$ where $\mathcal{F}$ contains all the rules defined for each gadget $\varphi(v)$ for each $v \in V(D_f)$. We define sets $I$ and $O$ as the nodes in $G$ corresponding to input gates in $D_f$ and the output nodes of gadgets representing output gates in $D_f$.  Now, we locally set every gadget $\varphi(v)$ to its fixed point configuration and we call it $x$. We assign $z =(z_1,z_2,\hdots, z_r)$ to each of the inputs of corresponding input gadgets. We claim that at time $t = 2 \text{deph}(D_f) = r^{\mathcal{O}(n)}$ the global function satisfies $f(z)|_I = F^t(x)|_{O}.$ In fact, it is not difficult to see that inductively, in $t_1 = 2$ all the gadgets in the first layer compute the assignment $(z_1,z_2,\hdots, z_r)$ and each of the outputs that are associated inputs in the first layer have now this information as a perturbation of their fixed point configuration $x$. Now assume that in some time $t = 2k$, gadgets in the $k$-th layer are computing the information received from layer $k-1.$ Again, because of Lemma \ref{lem:andorthreshold} we know that each gadget $\varphi(v)$ produces consistently an AND or an OR computation of its inputs in uniform time $t = 2$ and thus, $k+1$-th layer computes information of $k$-th layer in time $2k+2 = 2(k+1)$. Then, the claim holds. As a consequence of the claim we have that $(G,\mathcal{F})$ simulates $f$. The result holds.
\end{proof}

\subsubsection{Majority rules}
 
 An important example is the case in which each local rule will change to $1$ when the majority of the nodes in the neighbourhood of its associated node $v$ is in state $1$. More precisely, when $\theta_v = \frac{\delta_v}{2}.$ When this happens, we say that local rule $f_v$  is a majority rule. Of course this depend on the graph. Now we will show that we can simulate any monotone Boolean network by using only majority rules. Analogously to the previous result, we show first that we can find AND and OR functions as a part of some automata networks defined by majority rules.
  \begin{lem}
 	There exist two automata networks $\mathcal{A}_2 = (G_1=(V_1,E_1),\mathcal{F}_1)$ and  $\mathcal{A}_2 = (G_2=(V_2,E_2),\mathcal{F}_2)$ with global rules $F_1$ and $F_2 $ respectively, such that $\mathcal{F}_i$ are majority rules and

 \begin{enumerate}
 		\item $\wedge(x,y) = F^2_1(z_1)_{o} = F^2_1(z_1)_{o'}.$
 		\item  $\vee(x,y) = F^2_2(z_2)_{o} = F^2_2(z_2)_{o'},$
 	\end{enumerate}	for some $z_i \in \{0,1\}^{|V_i|},$ $i=1,2.$
 	\label{lem:andormaj}
 \end{lem}
\begin{proof}
	Consider the graph $G_1$ and $G_2$ given in Figure \ref{fig:andmaj} and Figure \ref{fig:ormaj}. We define $\overline{z}_1 = (0,0,0,0,0)$ and $\overline{z}_2 = (0,0,0,0,0)$. The result follows from the calculations in latter figures. 
\end{proof}

\begin{figure}[!tbp]
\centering
	\includegraphics[scale=0.5]{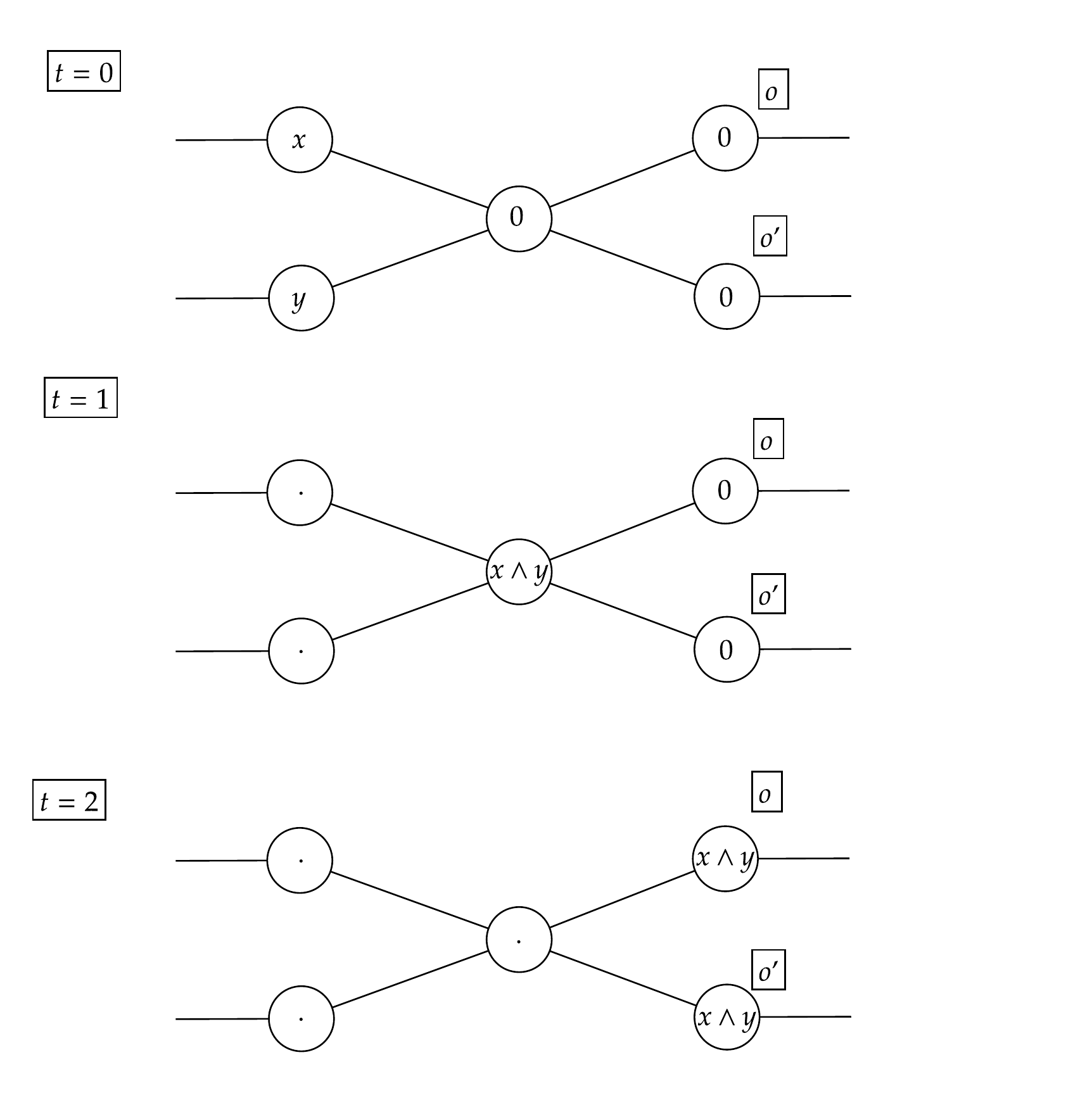}
	\caption{AND gadget for the class of majority totalistic functions.}
	\label{fig:andmaj}
\end{figure}
\begin{figure}
\centering
	\includegraphics[scale=0.5]{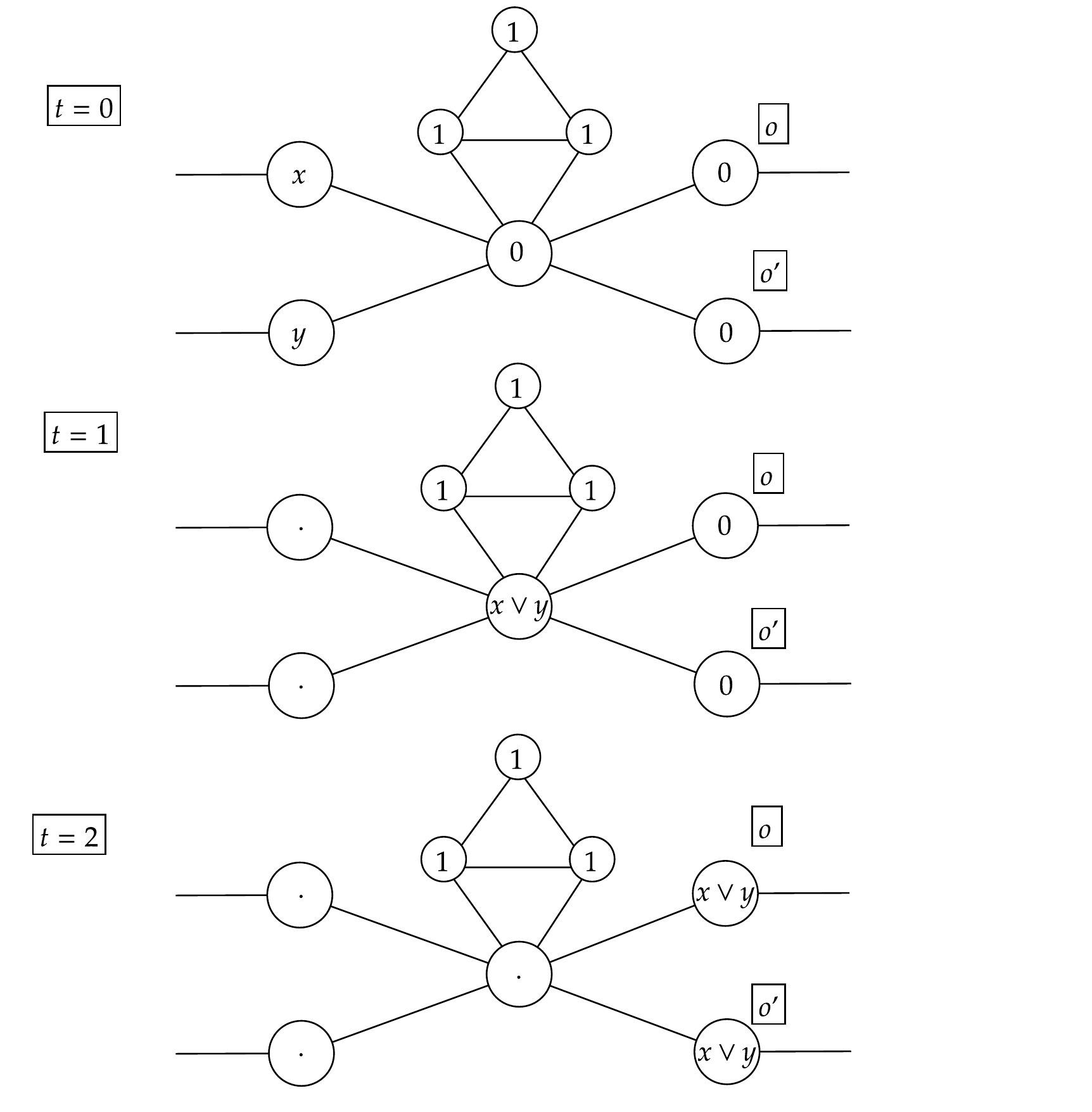}
	\caption{OR gadget for the class of majority totalistic functions.}
	\label{fig:ormaj}
\end{figure}
\begin{theo}

Let $r,s \in \mathbb{N}$ and  $f:\{0,1\}^r \to \{0,1\}^s$ be a monotone Boolean function. There exist an automata network $\mathcal{A}_f = (G,\mathcal{F})$ with global rule $F$ such that $\mathcal{F}$ are majority rules such that there exist $t= r^{\mathcal{O}(1)}$ satisfying $f(y) = F(x|_I)|_O$ for all $y \in \{0,1\}^r$, some sets $I,0 \subseteq V$ and some $x$ depending on $y$.
\end{theo}
\begin{proof}
	The proof of this result is completely analogous to Theorem \ref{teo:thershold}. 
\end{proof}
\begin{remark}
 It is also possible to simulate the evaluation of an arbitrary Boolean circuit by a monotone circuit. The construction duplicates the gates and uses De Morgan's laws in order to simulate the evaluation of NOT gates. Roughly, for each gate $v$ we work with a two duplicates $v_+$ and $v_{-}$  such that $v_+$ is true if and only if $v$ is true and $v_-$ is true if and only if $v$ is false. By duplicating in this way every gate in the original circuit we use $v_{-}$ any time we need to evaluate a NOT gate.  However, we mention this only as a remark because we consider it goes quite far away from our definition of simulation (though, one could adapt things to make it work).
\end{remark}
%
%
%

%
\subsection{Isolated totalistic rules}

 Now we introduce another class of totalistic rules that we  call \textit{isolated}. In general, the class of isolated totalistic rules are rules that are activated by a precise level of activation in the neighbourhood of  a given node. In fact, these rules will be activated if and only if the amount of active neighbours is exactly some value $\alpha$ and will be $0$ for any other value in sufficiently large enough interval containing $\alpha$. We detail this as following.
\begin{definition}
		A totalistic rule $f:S_{\Delta} \to \{0,1\}$ is isolated if there is a positive integer $\alpha\geq3$ such that $[\alpha-2,\alpha+1] = \{\alpha-2,\alpha-1, \hdots, \alpha+1\} \cap \mathcal{I}_f = \{\alpha\}.$
\end{definition}
For example the rule $3$ is isolated because configurations that have an amount of $1$s in the interval $[1,4]$ will only produce $1$ as image if they have exactly $3$ ones. Note that, for example, any other totalistic rule of the form $3a$ with $a\geq 5$ will be isolated with $\alpha = 3$. In the next section, we will call the value $\alpha$ an isolated value for some fixed rule. For example $3$ is an isolated value for rule $35$ and so for rule $3$. Note also that one fixed isolated rule can have multiple isolated values. For example, rule $36$ has $3$ as isolated value and also $6$. 

\begin{remark}
	Note that in the latter definition, taking $\alpha -2$ as a non active value for the rules is necessary in order to avoid considering matrix-defined functions that we have already studied. In fact, if we assume $\alpha=3$ and we allow $1$ to be active, then rule $135$ is the XOR rule for a neighbourhood with $6$ nodes.
\end{remark}
%
\begin{lem}
  For each $\alpha \geq 3$ there is an automata network $\mathcal{A}_\alpha= (G_\alpha,\mathcal{F}_\alpha)$ such that every  $f \in \mathcal{F}_\alpha$ is a  totalistic function with isolated value $\alpha$ and such that its global rule $F_{\alpha}$ satisfies that:  $\exists i_1, i_2 ,o_1,o_2\in V(G): F_{\alpha}^3(x)_{o_j}  = \textbf{NAND}(x|_{i_1},x|_{i_2}), j=1,2$ for any $x$ that is a perturbation in $i_0$ and $i_1$  of some $z \in \text{Fix}(\mathcal{A}_\alpha)$, i.e. $x_v = z_v$ for all $v \not \in \{i_1,i_2\}.$ In particular, NAND gate is in the spectrum of $\mathcal{A}_{\alpha}$ for $t=3$ and $l = 2$.
  \label{lemma:NAND}
\end{lem}

\begin{proof}
	Let $ \alpha  \geq 3$. We show explicitly the structure and  dynamics of $\mathcal{A}_\alpha$ in Figure \ref{fig:NANDiso}. Note that the graph structure strongly depends on the fact that complete graphs $K_{\alpha +1}$ are stable connected components for state $1$ in the sense that nodes inside this clique will be always in state $1.$ From the Figure \ref{fig:NANDiso} it is evident that fixed point $z$ is given by the state in which $z_v = 0$ for any $v$ which is not part of one the two cliques in the graph and that computation of NAND is a consequence of a perturbation of $z$.
\end{proof}

\begin{figure}[!tbp]
	\includegraphics[scale=0.45]{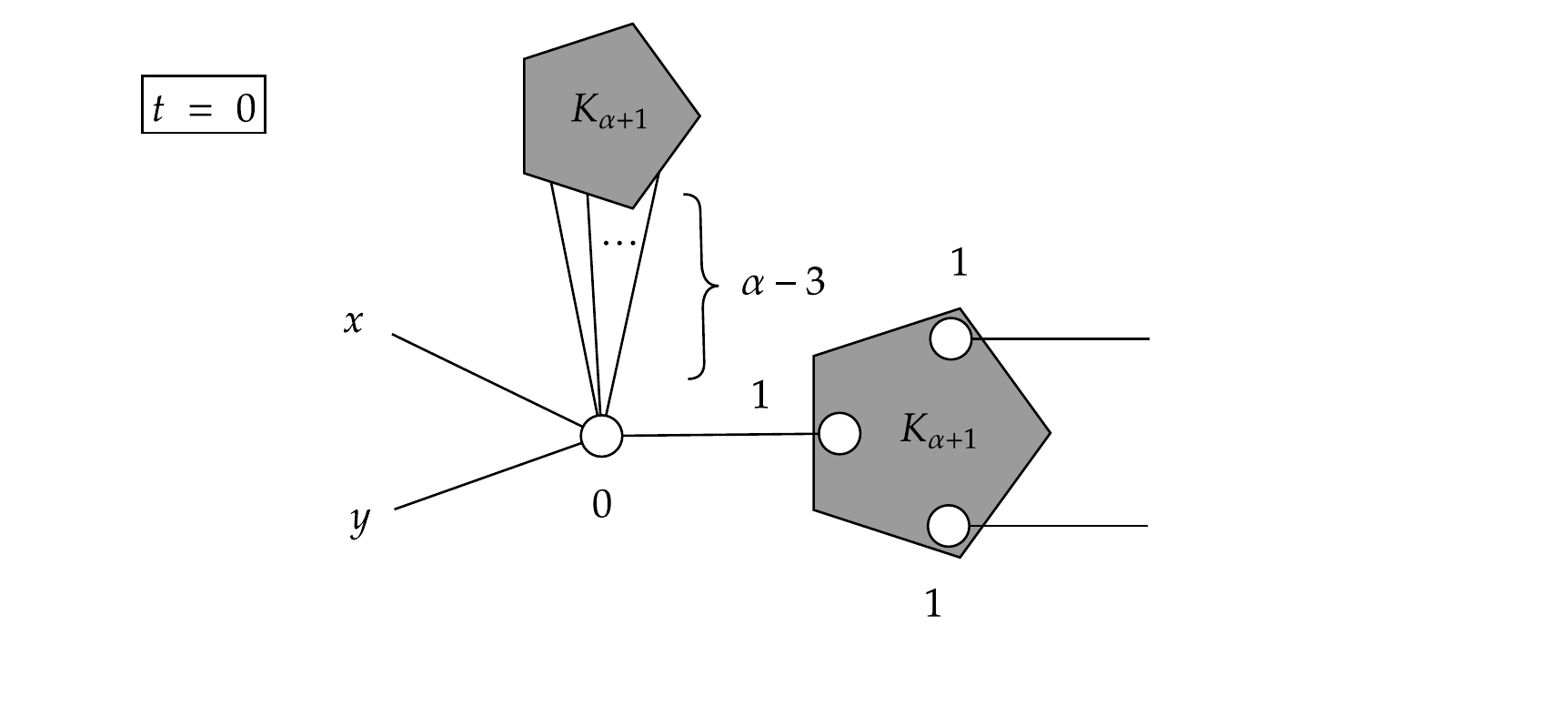}
	\includegraphics[scale=0.45]{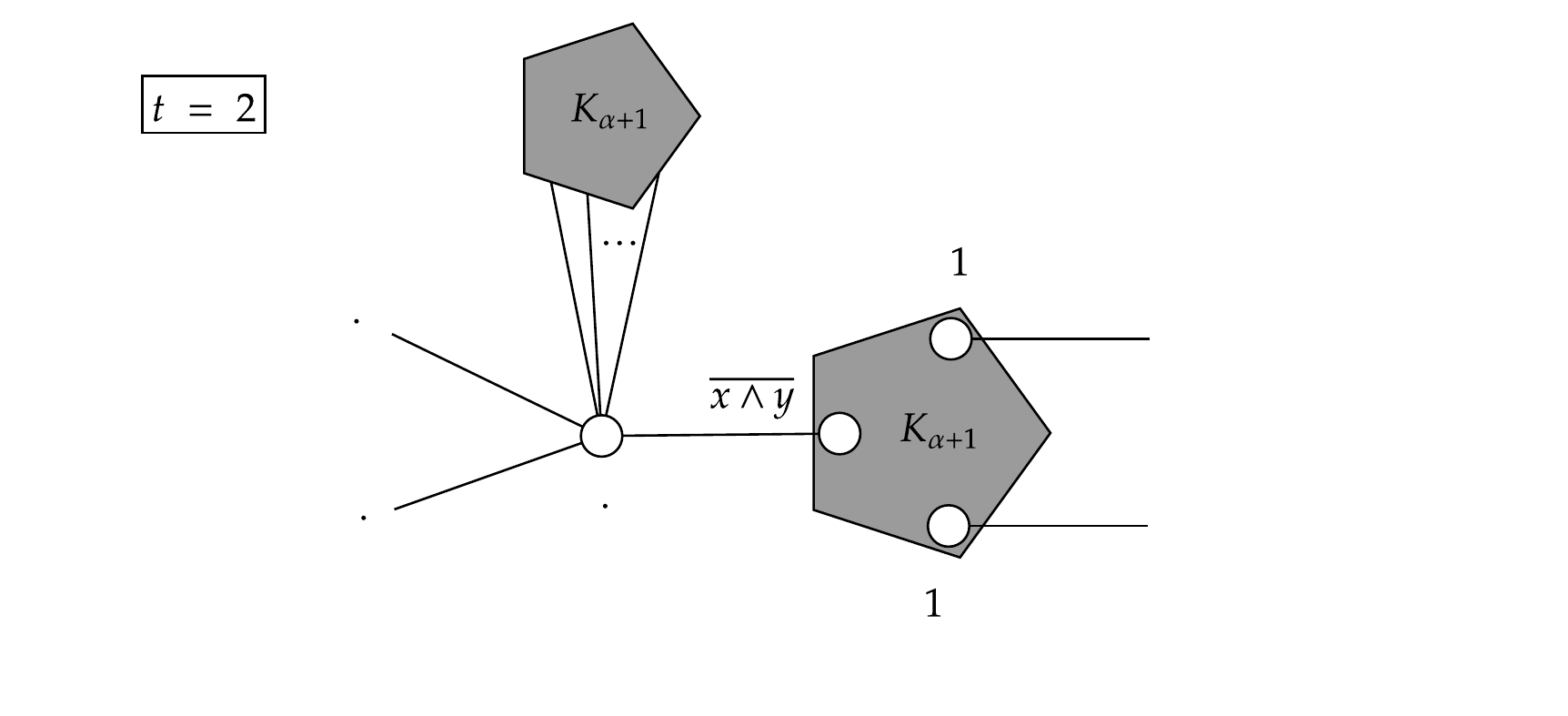}
	\includegraphics[scale=0.45]{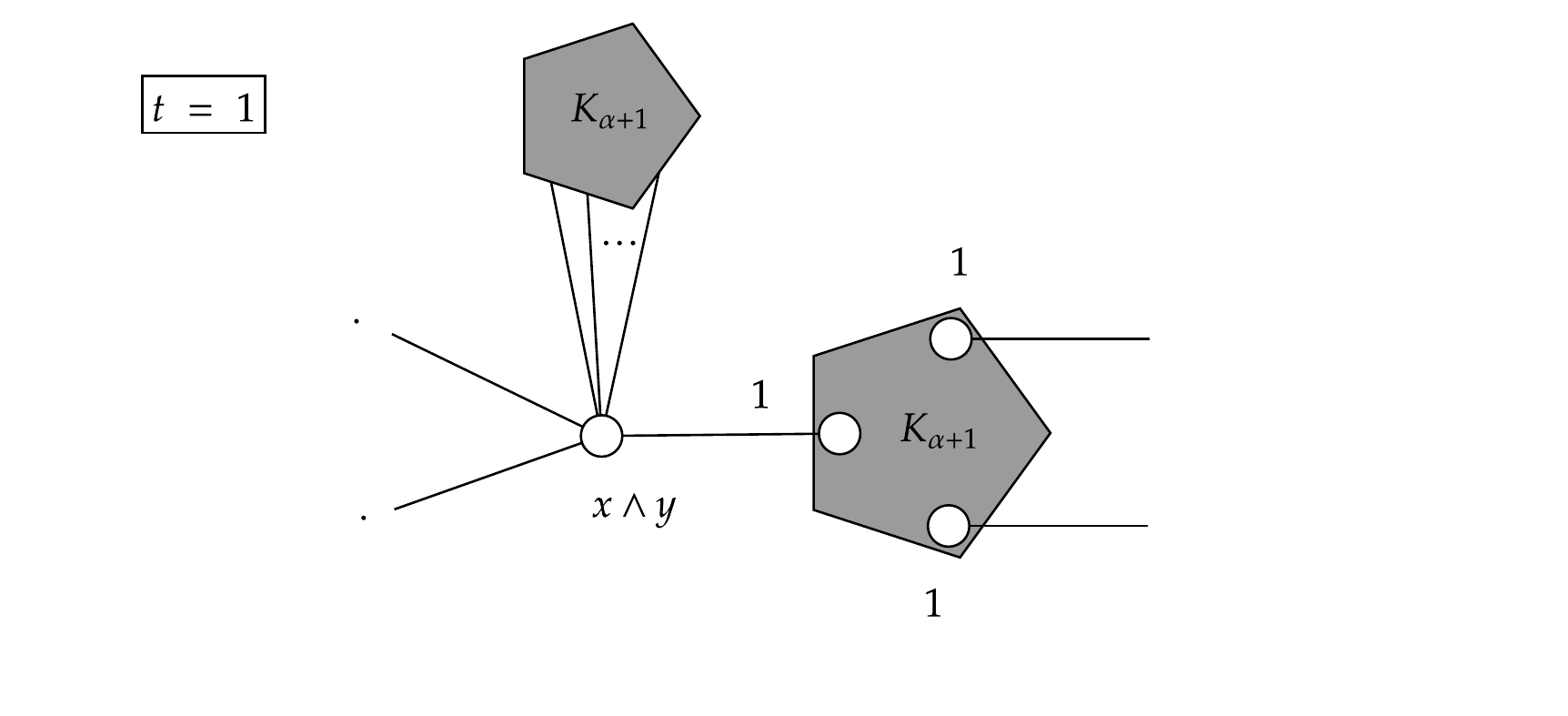}
	\includegraphics[scale=0.45]{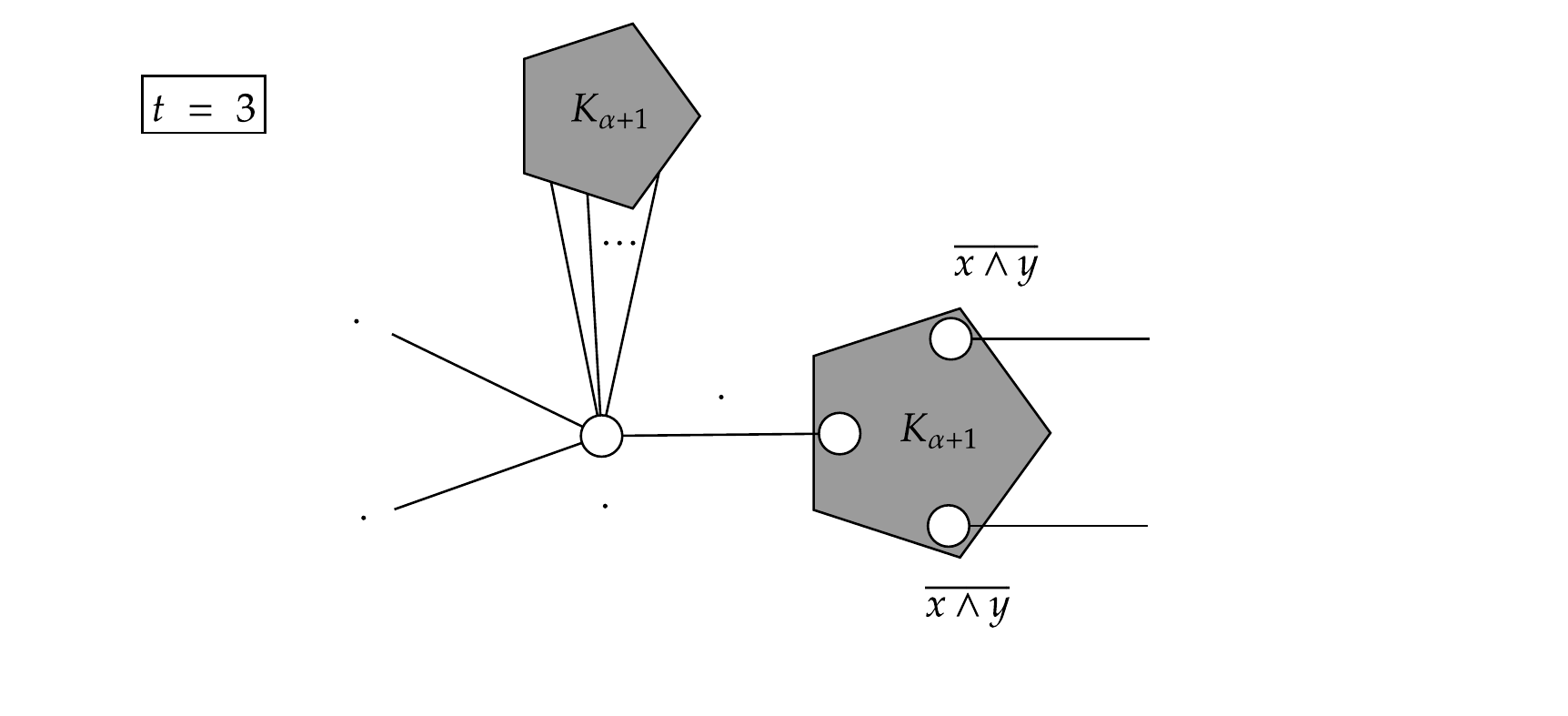}
	\caption{NAND gadget for $\alpha$-uniform isolated totalistic rules with $\alpha \geq 3.$ Grey $K_{\alpha+1}$ components are fixed in state $1$.}
	\label{fig:NANDiso}
\end{figure}
During the rest of this section we will call the automata network in Figure \ref{fig:NANDiso}  a NAND gadget. We will now show that the class of $\alpha$-uniform totalistic isolated rules is complete for some class of graphs $\mathcal{G}$ that we will show. To do that, we need to define two more gadgets: a clock gadget and a clocked NAND gadget.
\begin{figure}[!tbp]
\centering
	\includegraphics[scale=0.5]{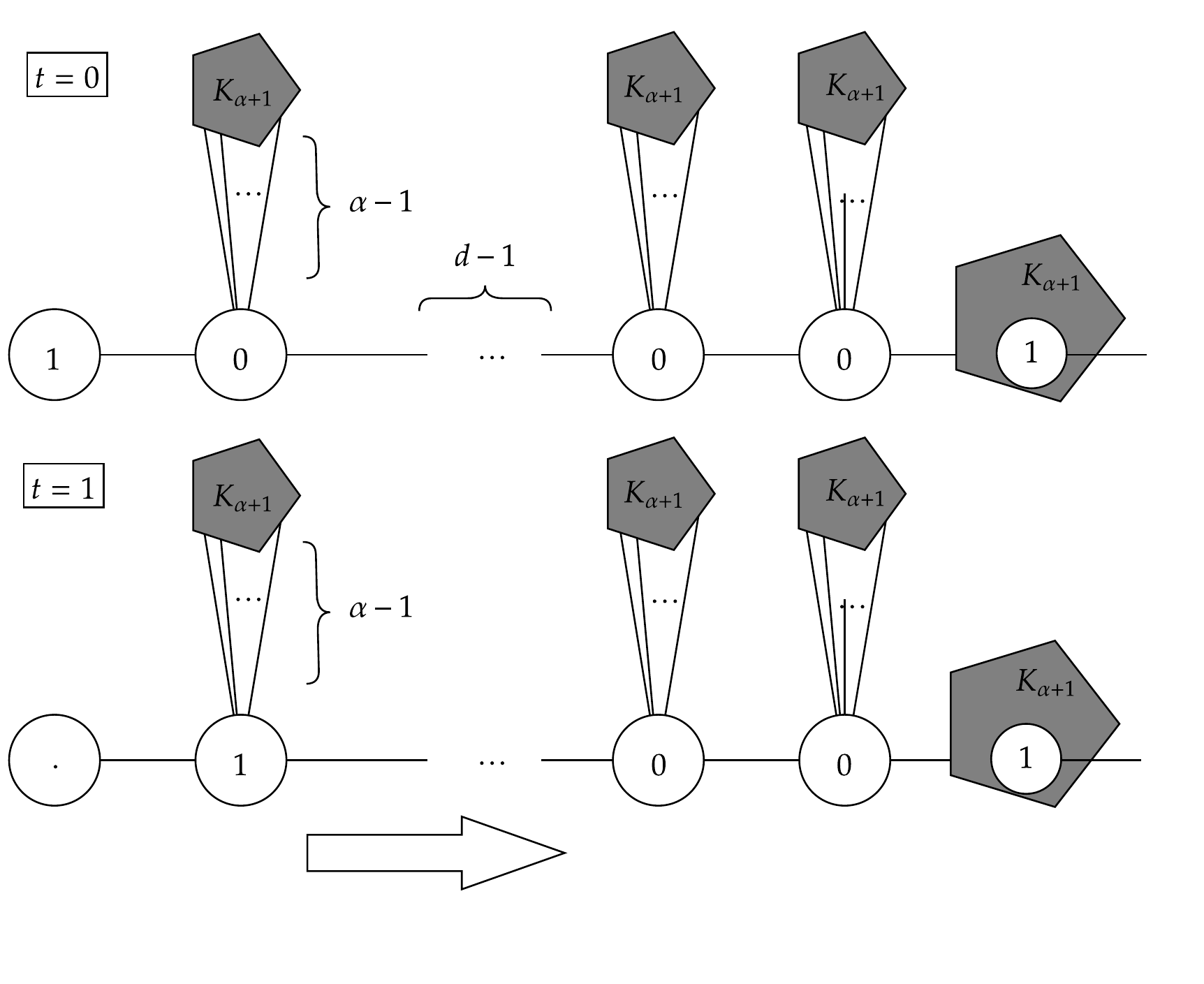}
	\caption{Clock gadget for $\alpha$-uniform isolated totalistic rules with $\alpha \geq 3.$ Grey $K_{\alpha+1}$ components are fixed in state $1$.}
	\label{fig:clockiso}
\end{figure}
\begin{lem}
	  For each $\alpha \geq 3$ and $d \geq 1$ there is an automata network $\mathcal{A}_{\alpha,d}= (G_{\alpha,d},\mathcal{F}_{\alpha,d})$ such that every $f \in \mathcal{F}_{\alpha,d}$ is isolated with isolated value $\alpha$ and such that its global rule $F_C$ satisfies that there exists $o \in V(G)$: $F_C^s(x)_{o}  = 1$ for $0\leq s \leq d-1$ and $F_C^{d}(x)_{o} = 0$ for some  $x \in \{0,1\}^n$.
\end{lem}
\begin{proof}
	See Figure \ref{fig:clockiso} for the structure of the gadget and the definition of $x$. Note that the gadget works as a wire defined by its path structure that carries the $1$ signal which perturbs one node in one of the complete graphs $K_{\alpha+1}$ at the end of the path after $d$ time steps. This is because each node in the path has exactly $\alpha-1$ neighbours in state $1$ and thus, incoming signal allow them to change its state.
\end{proof}
Finally we introduce the clocked NAND gadget in the following lemma.
\begin{lem}
	For each $\alpha \geq 3$ and $d \geq 1$ there is an automata network $\mathcal{A}_{\alpha,d}= (G_{\alpha,d},\mathcal{F}_{\alpha,d})$ such that $\mathcal{F}_{\alpha,d} \in \mathcal{S}_\alpha$ and such that its global rule $F_{CN}$ satisfies that there exist $i_1,i_2,o_1,o_2 \in V(G)$ such that $F_{CN}^s(w)_{o_j}  = 1, j=1,2$ for $0\leq s \leq d-1$ and $F_{CN}^{d+3}(x)_{o_j} = \textbf{NAND}(F^d(w)|_{i_1},F^d(w)|_{i_2}), j=1,2$ for some $w \in \{0,1\}^n$
\end{lem}
\begin{proof}
	In  Figure \ref{fig:clockedNANDiso} we show the structure of a clocked NAND gadget. We define $i_1 $and $i_2$  as the nodes that are labelled by $x$ and $y$ in Figure \ref{fig:clockedNANDiso}. We define $w' = z \cup r$ the concatenation of the configuration $z$ in Lemma \ref{lemma:NAND} and $r$ in the previous lemma. Finally we define $w_{i_1} = w_{i_2} = 1$ and $w_v = w'_v$ for all $v \not \in \{i_1,i_2\}$. Note that central node (which is connected to nodes labelled as $x$ and $y$) is in state $0$ and has exactly $\alpha + 1$ active networks: $\alpha - 3$ active networks from the clique in the upper part of the gadget, one from  the clique in the right, one from the clock gadget and two from the inputs. As $\alpha$ is an isolated value for the local rules in the gadget then, we have $F_{CN}^s(w)_{o_j}  = 1, j=1,2$ for $0\leq s \leq d-1.$ Now note that in $t=d$ the neighbour of the central node located in the clock gadget will change to $0$ and then we will recover the same scenario shown in Figure \ref{fig:clockiso} with nodes labelled $x$ and $y$ assuming the values $F(x)|_{i_1}$ and $F(x)|_{i_2}$ and thus as a consequence of latter lemma, the result holds.
\end{proof}
\begin{figure}[!tbp]
\centering
	\includegraphics[scale=0.5]{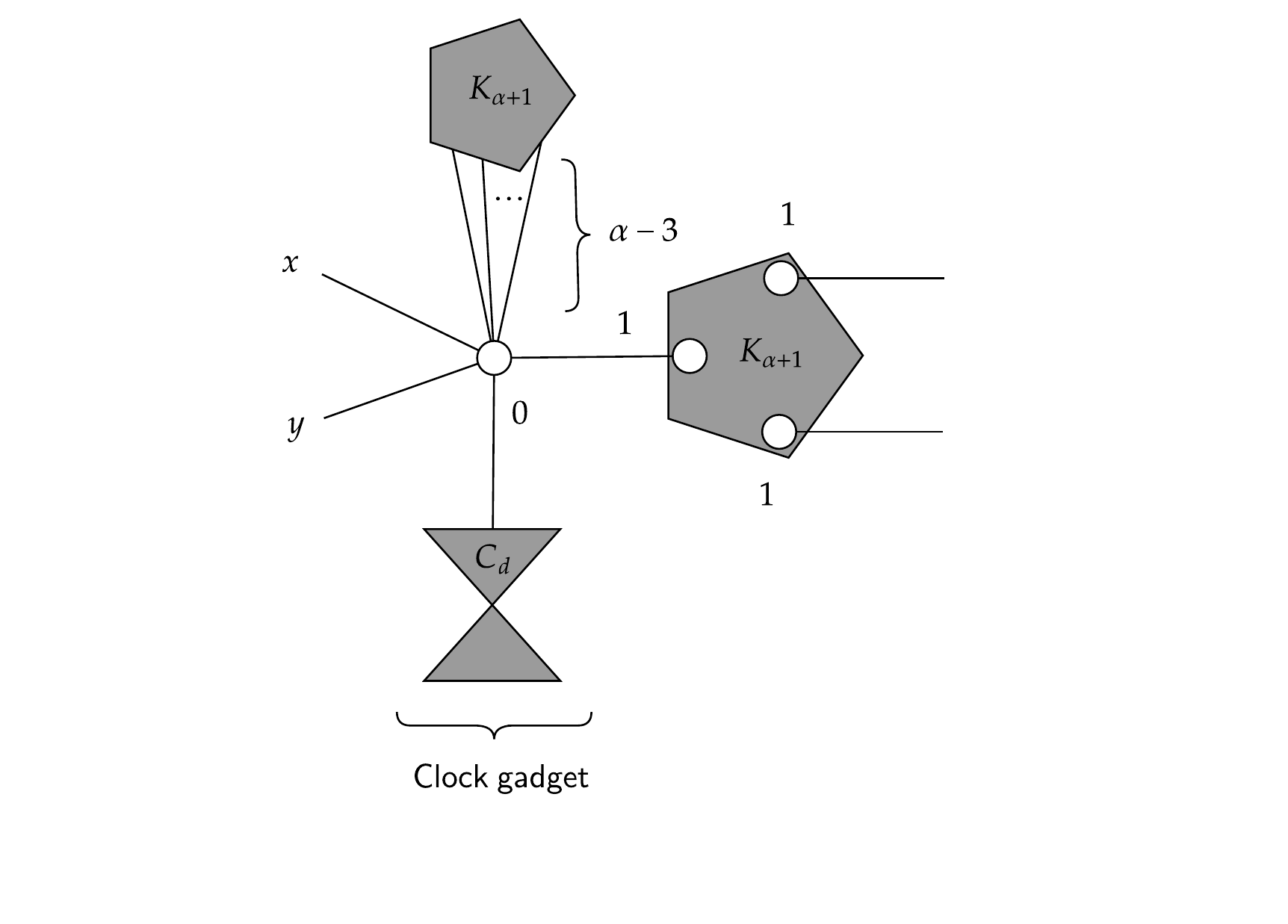}
	\caption{Clocked NAND gadget for $\alpha$-uniform isolated totalistic rules with $\alpha \geq 3$ and delay $d.$ Clock picture represents gadget in Figure \ref{fig:clockiso}.}
	\label{fig:clockedNANDiso}
\end{figure}

Note that as inputs in clocked NAND gadget in Figure \ref{fig:clockedNANDiso} will be fixed in $1$ and thus, the second part of the latter lemma may seem trivial. However, as now we would want to connect different clocked NAND gadgets in order to simulate a circuit, nodes $i_1$ and $i_2$ will be identified with some outputs of some other gadget, allowing us to make simulate the calculations of a NAND gate at the wright time.  We will detail this below. Now we are prepared to show the main result of this section.

\begin{theo}
	Let $r,s \in \mathbb{N}$ and $f: \{0,1\}^r \to \{0,1\}^s$ a Boolean function. For each $\alpha \geq 3$ there is a set of isolated functions $\mathcal{F}$ with isolated value $\alpha$ and a bounded degree class of graphs $\mathcal{G}$ such that $\mathcal{F}$ simulates $f$ in $\mathcal{G}$.
	\label{teo:isolated}
\end{theo}
\begin{proof}
	Let $C_f$ be a circuit representing $f$. Without loss of generalisation we can assume that any gate is a NAND gate. We will use again the same reasoning we used for Theorem \ref{teo:thershold}. The only difference here is that we have to be very careful in order to initialise clocked NAND gadgets in each layer of the circuit as it is shown in Figure \ref{fig:circuitiso}. In order to do that, we set locally each clocked NAND get to the local configuration $w$ given in the latter lemma with exception of the nodes representing input gates (those nodes will be identify with inputs of gadgets in first layer which do not have any clock ).  Nodes in the second layer will wait time $d=3$ as a consequence of their clocks. When the calculations of the first layer arrives to the second one, their clocks will be off (the neighbour located in the clock will be inactive) and thus they will be able to compute a NAND from the values computed in the first layer. Then, third layer has to wait $6$ in order to coherently simulate the evaluation of the circuit so we choose $d=3$ as it is shown in Figure \ref{fig:circuitiso}. In general the $k$-th layer will have a clock with delay $d = 3k$. Note that, as a consequence of last lemma, any gadget with exception of the first layer will stay fixed in configuration $w$ and will calculate NAND when its clock gets inactive. Then, after $t = 3 \text{depth}(C_f)$ and thus, the result holds.
\end{proof}
\begin{figure}[!tbp]
    \centering
	\includegraphics[scale=0.56]{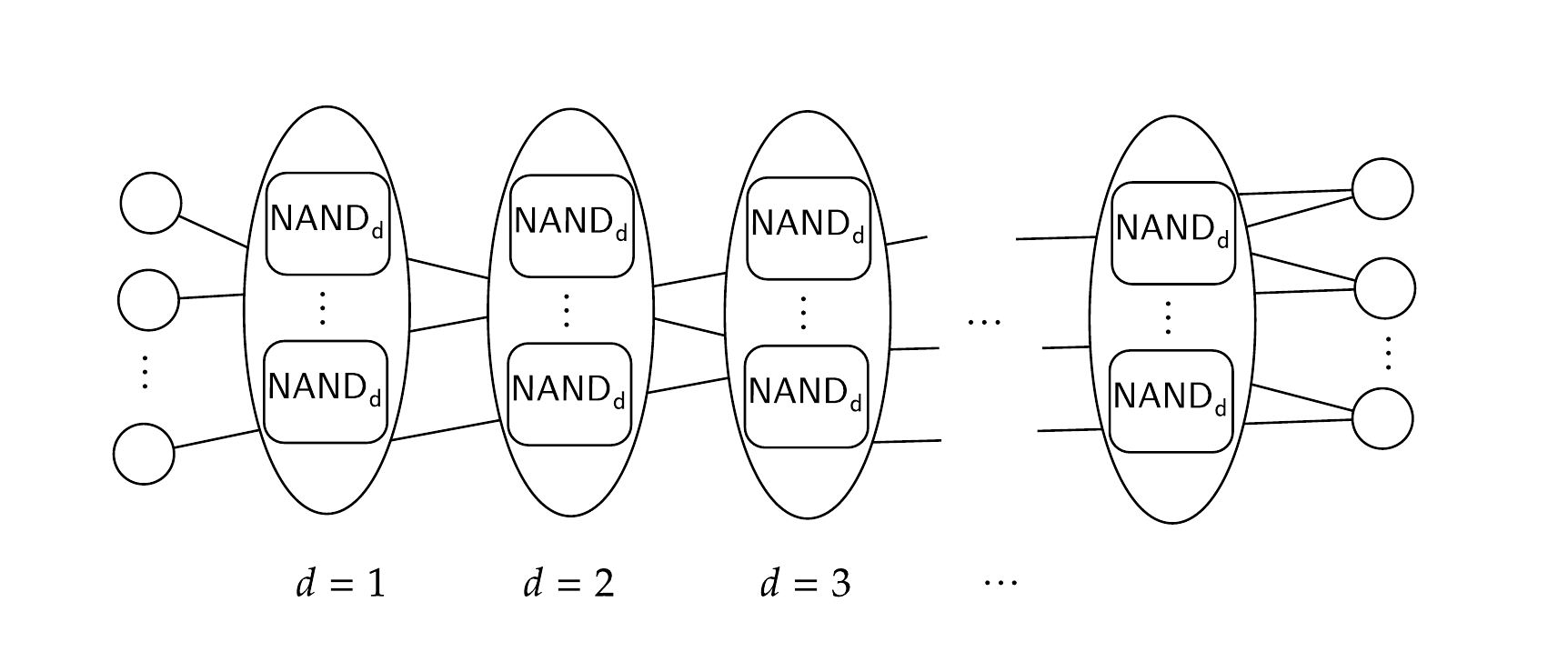}
	\caption{Scheme of the underlying graph of an automata network simulating a circuit defined over a class of $\alpha$-uniform isolated totalistic rules with $\alpha \geq 3.$ $\text{NAND}_d$ represents clocked NAND gadget in Figure \ref{fig:clockedNANDiso} }
	\label{fig:circuitiso}
\end{figure}

\subsubsection{Special cases: Rule $1$ and Rule $2$.}
In this section, we study two special cases of totalistic rules that are not strictly included in the last formalism. These are the rule $1$ and rule $2$ that change its values to active (change to state $1$) if and only if there is $1$ (respectively $2$) active neighbours in some given time step. The case of rule $1$ is particularly interesting considering that the study of the dynamics (particularly related to decision problems such as prediction) is still open in the case in which the network is a two-dimensional cellular automata (the underlying interaction graph of the network is a grid). 
\paragraph{Rule $1$}
We start that showing that the XOR gates is in the spectrum of some automata network in which every rule is given by the rule $1$
\begin{prop}
	There is a Rule $1$ automata network $\mathcal{A} = (G =(V,E),\mathcal{F})$ such that for some $i_1, i_2, o_1, o_2 \in V$  we have that $F^2(z)_{o_j} = \textbf{XOR}(z|_{i_1},z|_{i_2}) =  z|_{i_1} \oplus z|_{i_2} $ for $j=1,2$ and some perturbation $z$ of fixed point $\vec{0}$ in $i_1$ and $i_2$. In particular, XOR gate is in the spectrum of $\mathcal{A}.$
\end{prop}
\begin{proof}
	We show the gadget that defines  $\mathcal{A} = (G =(V,E),\mathcal{F})$ in Figure \ref{fig:XOR1}.  We identify the nodes at the left labelled by $x$ and $y$ as inputs $i_1$ and $i_2$ and the ones in the wright in state $0$ as $o_1$ and $o_2$. Note that because of the local rule we will read in the output a $1$ if exclusively one of the inputs is in state $1$. Thus, gadget in Figure \ref{fig:XOR1} computes a XOR gate. The result holds.
\end{proof}
\begin{figure}[!tbp]
\centering
	\includegraphics[scale=0.28]{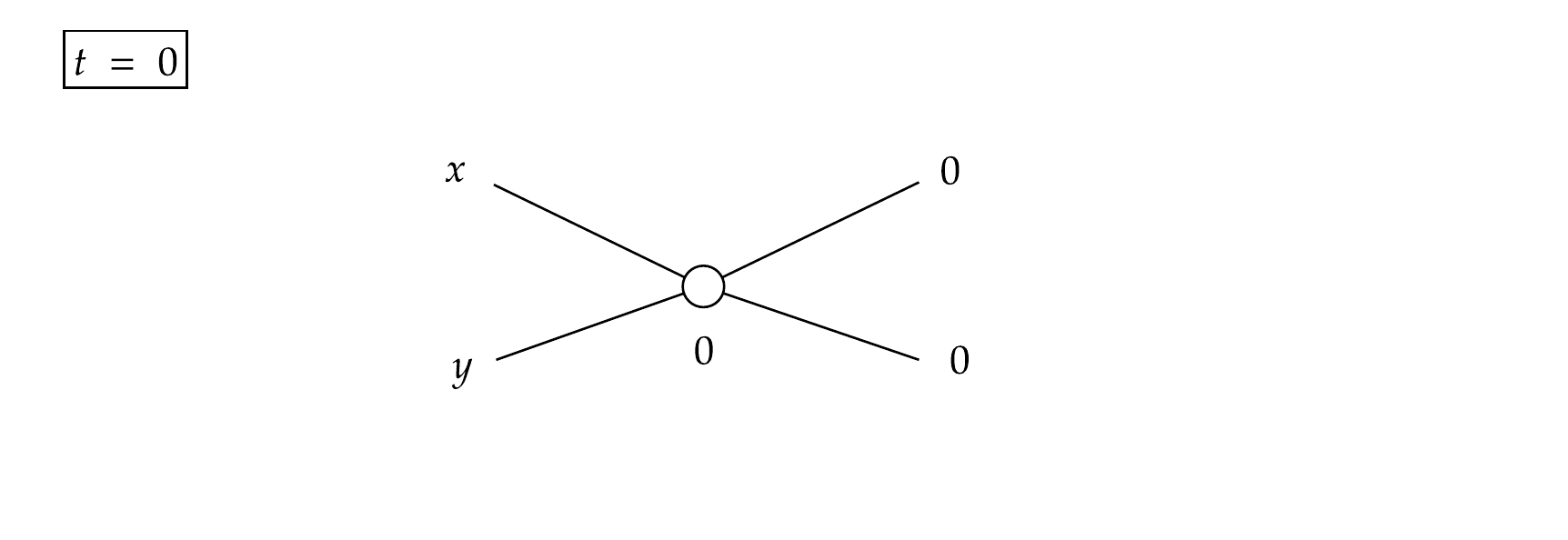}
	\includegraphics[scale=0.28]{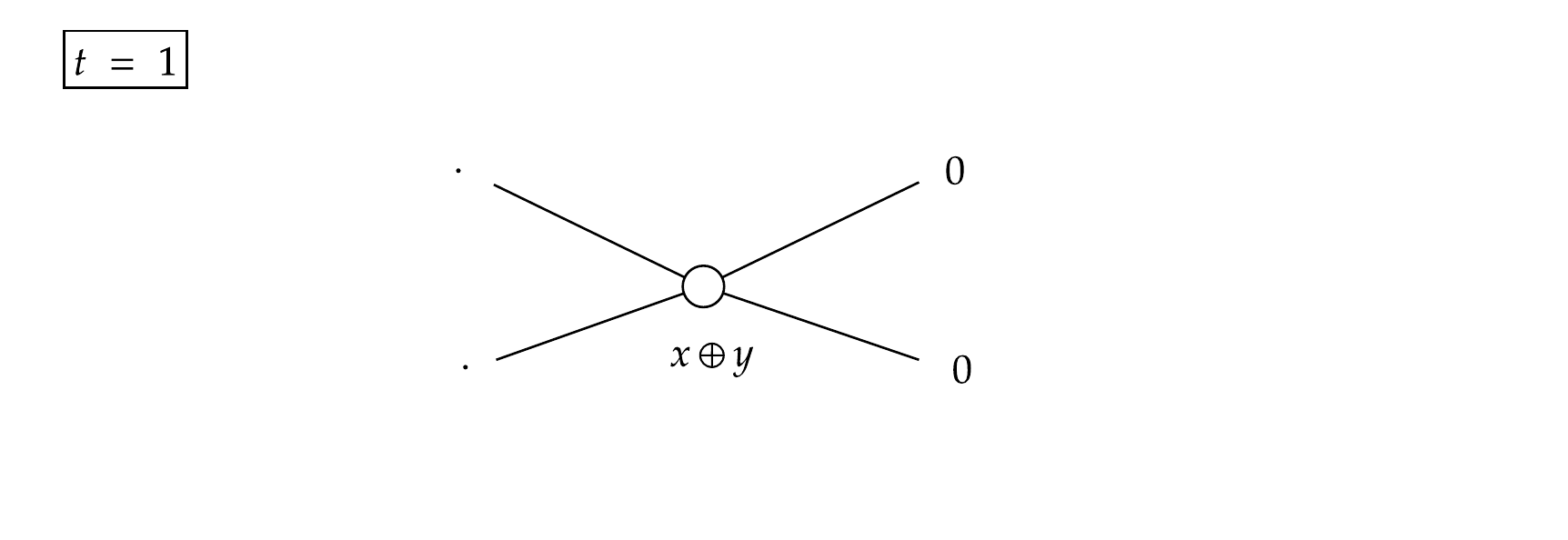}
	\includegraphics[scale=0.28]{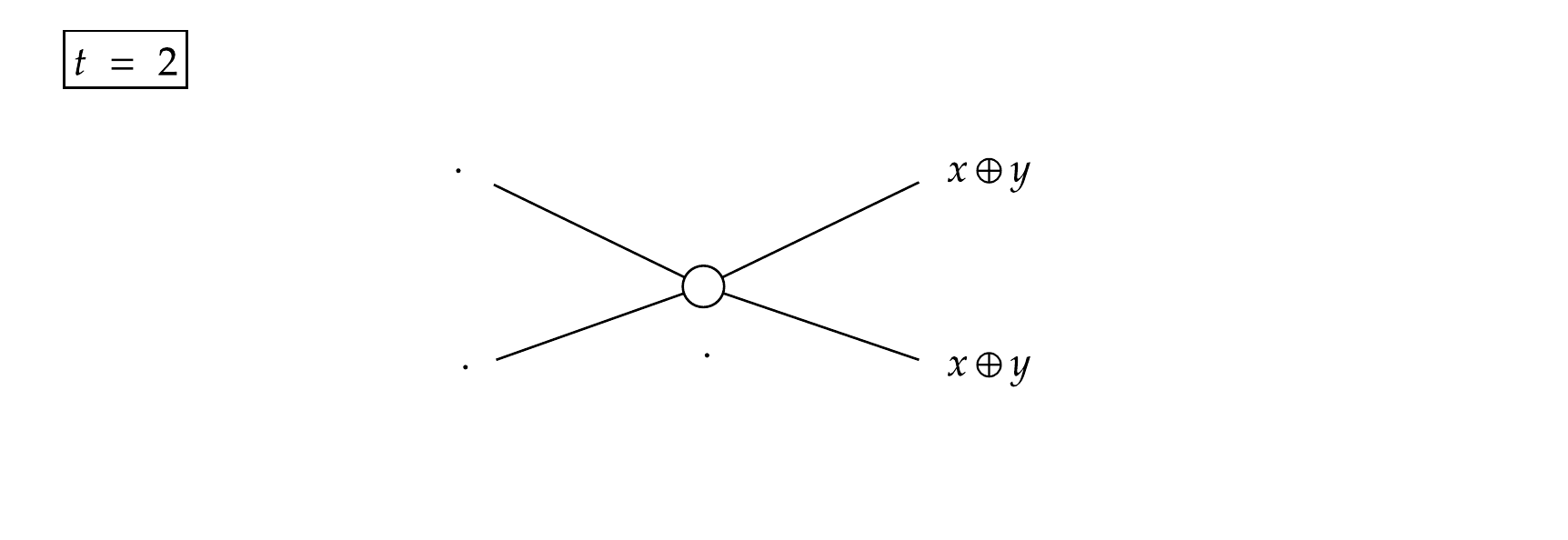}
	\caption{XOR gadget for Rule $1$. $x$ and $y$ nodes represent the inputs and the nodes in the left represent the output. Computation takes place in central node and output is read after $2$ time-steps.}
	\label{fig:XOR1}
\end{figure}
Now, we slightly modify this gadget in order to compute a NOR gate. This will allow us to simulate an arbitrary Boolean circuit.
\begin{prop}
	There is a Rule $1$ automata network $\mathcal{A} = (G =(V,E),\mathcal{F})$ such that for some $i_1, i_2, o_1, o_2 \in V$ and a configuration $y \in \{0,1\}^{|V|}$,  we have that $F^2(z)_{o_j} = \textbf{NOR}(z|_{i_1},z|_{i_2})$ for $j=1,2$ and some perturbation $z$ of $y$ in $i_1$ and $i_2$.
\end{prop}
\begin{proof}
	Figure \ref{fig:NOR1} shows the gadget that defines  $\mathcal{A} = (G =(V,E),\mathcal{F})$.  We identify the nodes on the left labelled by $x$ and $y$ as inputs $i_1$ and $i_2$ and the ones on the right in state $0$ as $o_1$ and $o_2$. Note that the configuration  $y$ sets all the nodes in state $0$ with exception of the node in the upper part, connected to the node in the central part of the graph, which is in state $1$. Note that this node blocks every signal coming from the inputs (every assignation for $x$ and $y$), forcing the outputs to change to $1$ in $2$ times steps only in the eventuality in which $x=y=0$. Thus, gadget in Figure \ref{fig:NOR1} computes a NOR gate. The result holds.
\end{proof}
\begin{figure}[!tbp]
\centering
\includegraphics[scale=0.28]{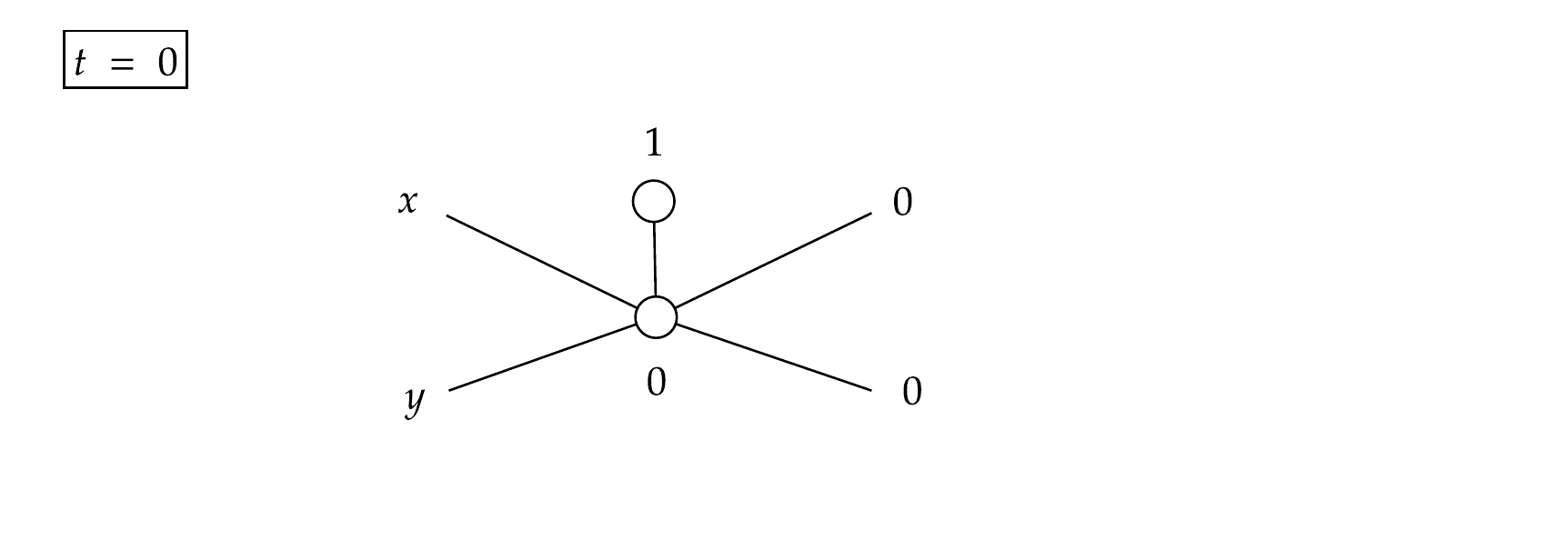}
\includegraphics[scale=0.28]{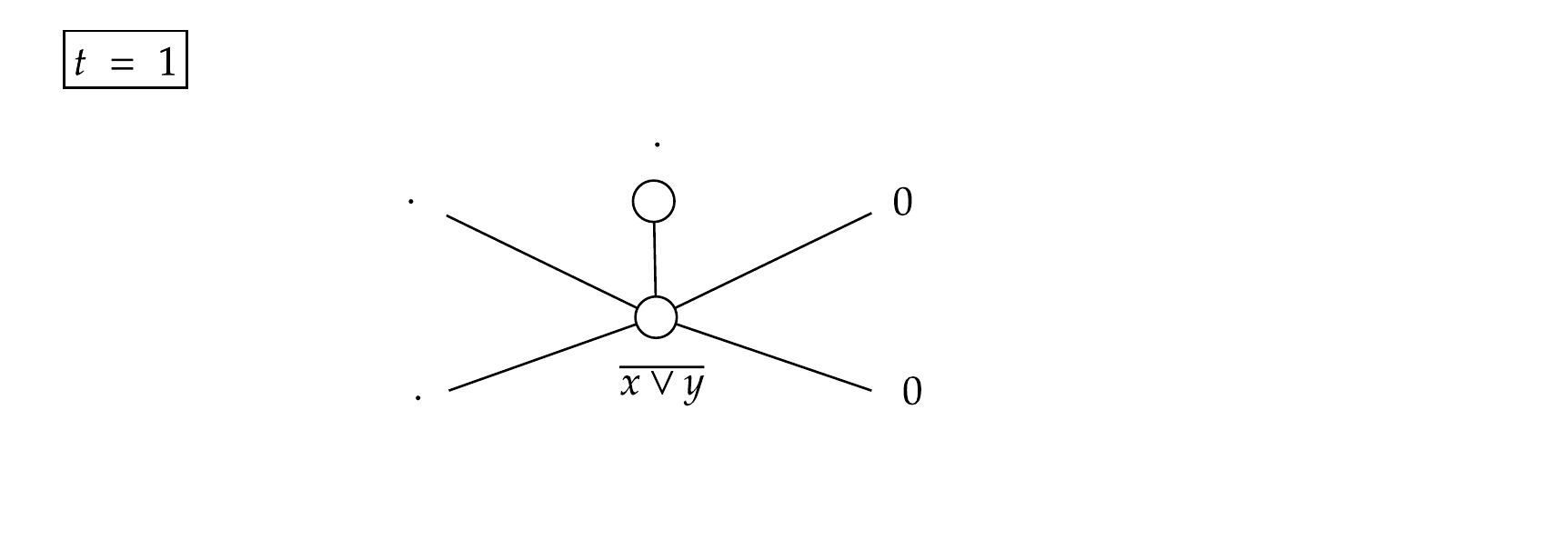}
\includegraphics[scale=0.28]{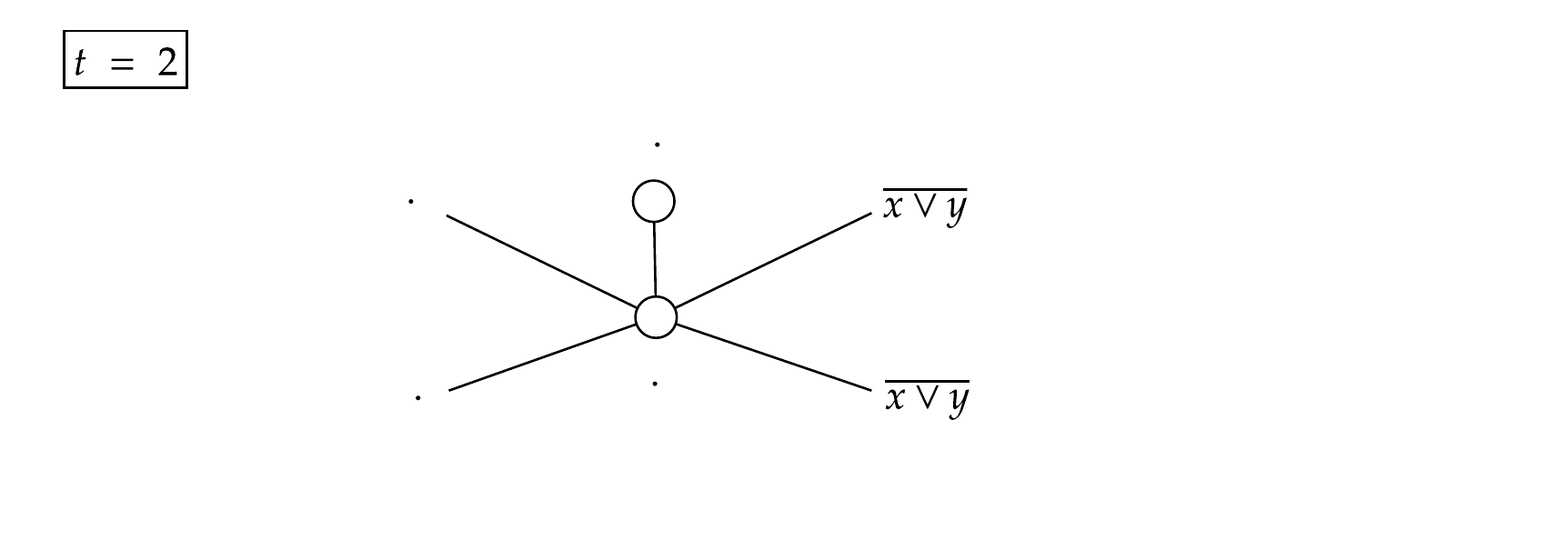}
	\caption{NOR gadget for Rule $1$. $x$ and $y$ nodes represent the inputs and the nodes in the left represent the output. Computation takes place in central node and output is read after $2$ time-steps.}
	\label{fig:NOR1}
\end{figure}
Now, we need to synchronise different NOR gadgets in order to connect them. This will allow us to simulate arbitrary Boolean circuits. In order to fulfil this task, we need first to introduce the $d$-wire gadget. This latter gadget is completely analogous to the clock gadget developed for isolated rules in Figure \ref{fig:clockiso}.
\begin{prop}
	There is a Rule $1$ automata network $\mathcal{A} = (G =(V,E),\mathcal{F})$ such that for some nodes $i$ and $o$ we have that $F^d(y)|_{o} = y|_i$ for any perturbation in $i$ of the steady-state $\vec{0}$.
\end{prop}
\begin{proof}
	$d$-wire gadget is shown in Figure \ref{fig:wire1}. The result holds directly from the definition of rule $1$.
\end{proof}
\begin{figure}[!tbp]
\centering
	\includegraphics[scale=0.5]{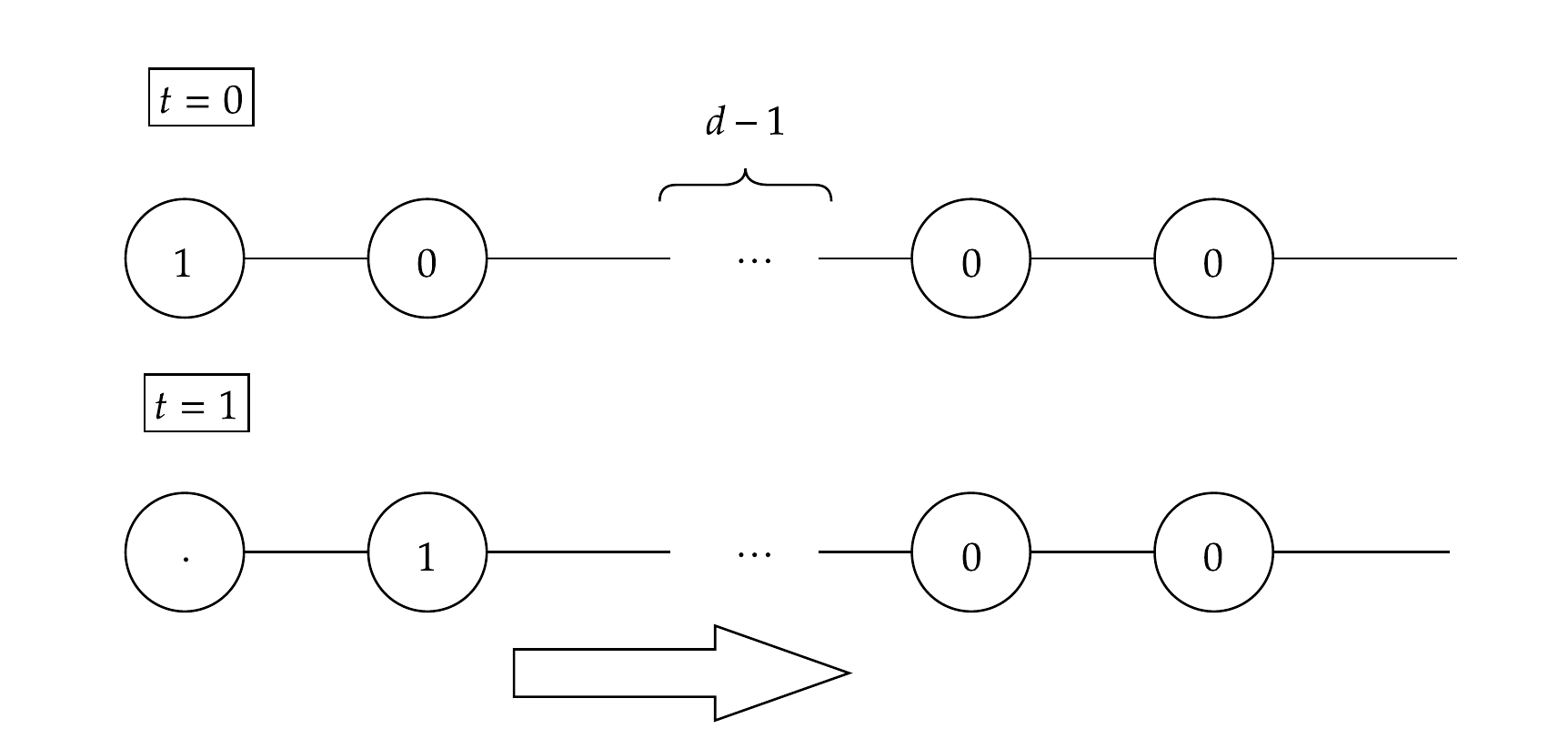}
	\caption{$d$-wire gadget for Rule $1$.}
	\label{fig:wire1}
\end{figure}
Now, we show a general gadget which computes a NOR gate in some given time $t = d$ taking as output an arbitrary state defined for the dynamics of an arbitrary Rule $1$ automata network, provided that,  the dynamics of these automata networks must not locally perturb the computation of the NOR gadget, i.e. nodes connected to this latter gadget must remain in state $0$ until time of computation, more precisely, for $t =0,\hdots, d-1$ time steps.
\begin{prop}
	Let $\mathcal{A} = (G,\mathcal{F})$ and $\mathcal{A'} = (G',\mathcal{F}')$ be two rule $1$ automata networks such that for some $i \in V(G)$ and $i' \in V(G')$ the respective global rules $F$ and $F'$ are such that $F^s(x)_i = F'^s(x')_{i'} = 0$ for $s = 0,\hdots, d-1$. Then, there exist a rule $1$ automata network $\mathcal{A}_d = (G_d,\mathcal{F}_d)$ such that its global rule $F_d$ is satisfies that $(F_d^{d+2}(z))_{o_1} = \textbf{NOR}(F^{d}(z)|_i,F'^{d}(z)|_{i'}).$
\end{prop}
\begin{figure}[!tbp]
\centering
	\includegraphics[scale=0.5]{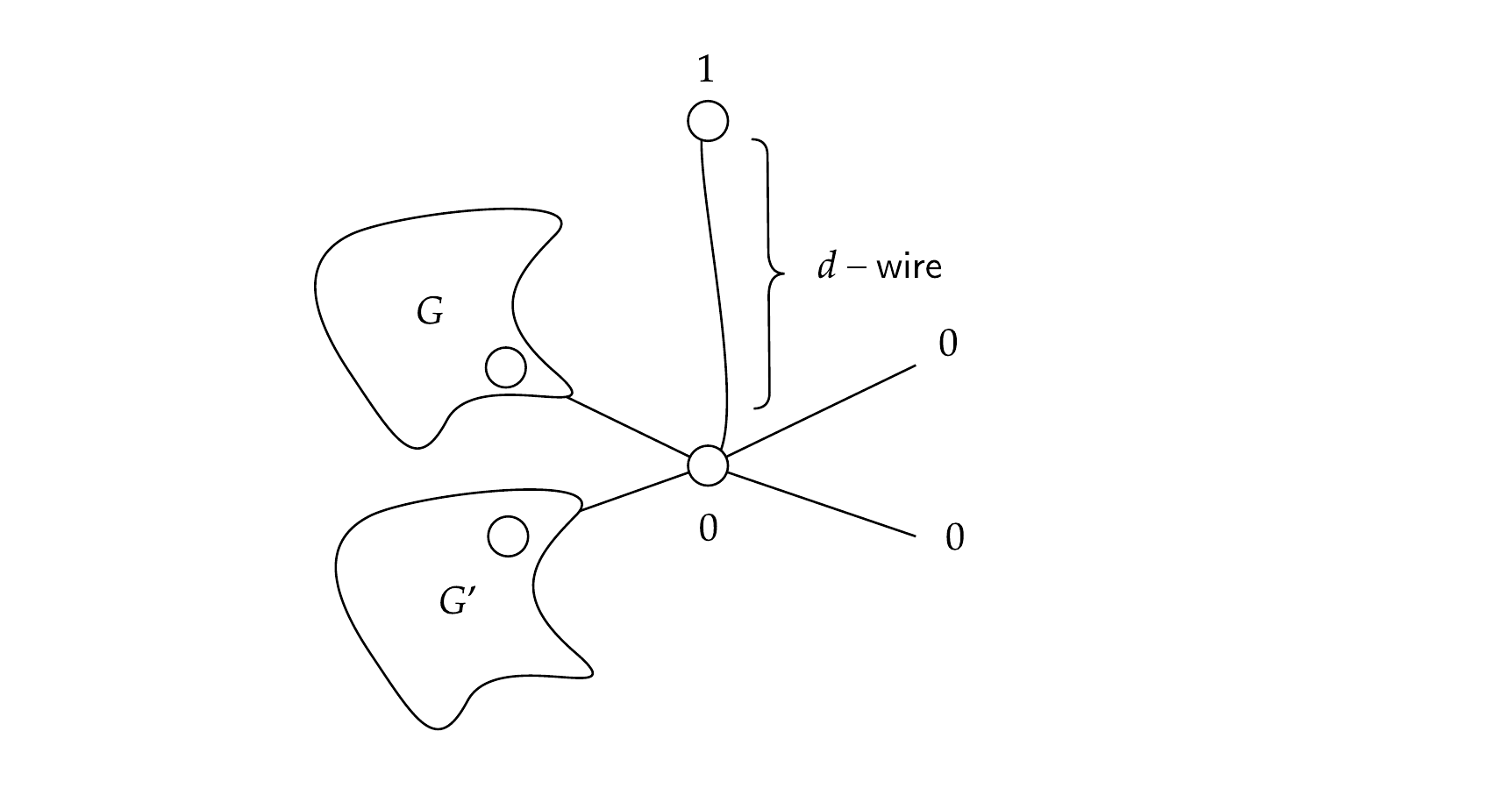}
	\caption{$d$-delayed NOR gadget for rule $1$.}
	\label{fig:NORgen1}
	\end{figure}
\begin{proof}
	We show the gadget $\mathcal{A}_d$ in Figure \ref{fig:NORgen1}. The result is the direct consequence of the last proposition and the fact that dynamics of $\mathcal{A}$ and $\mathcal{A'}$ do not perturb the dynamics of NOR gadget. The system remains in state $0$ with exception of the signal that travels through the $d$-wire gadget and arrives to the terminal node connected to central node of the gadget in time $t=d$. Computation takes place as shown in Figure \ref{fig:NOR1}. The result holds.
\end{proof}
We call the gadget in Figure  \ref{fig:NORgen1} a delayed NOR gadget. Finally, as a direct consequence of an analogous construction of the one we exhibited for isolated rules, we can show that we can simulate an arbitrary Boolean by correctly synchronizing different delayed NOR gates.
\begin{theo}
	Let $r,s \in \mathbb{N}$ and $f: \{0,1\}^r \to \{0,1\}^s$ a Boolean function. There exist a a bounded degree class of graphs $\mathcal{G}$ such that rule $1$ simulates $f$ in $\mathcal{G}$.
\end{theo}
\begin{proof}
	Proof is analogous to the proof of Theorem \ref{teo:isolated}.
\end{proof}
\paragraph{Rule 2}
For rule $2$ we show in Figure \ref{fig:NAND2} that it is capable of computing a NAND gate and we note that it is possible to use an analogous reasoning of the one of Theorem \ref{teo:isolated} in order to deduce that this rule can also simulate arbitrary Boolean networks. Particularly, we will need a clocked version of this gate that can be built by attaching three clock gadgets (see Figure \ref{fig:clockiso}) to the node that computes an AND gate from inputs (labelled by $x$ and $y$) in Figure \ref{fig:NAND2}. 
\begin{figure}[!tbp]
	\includegraphics[scale=0.45]{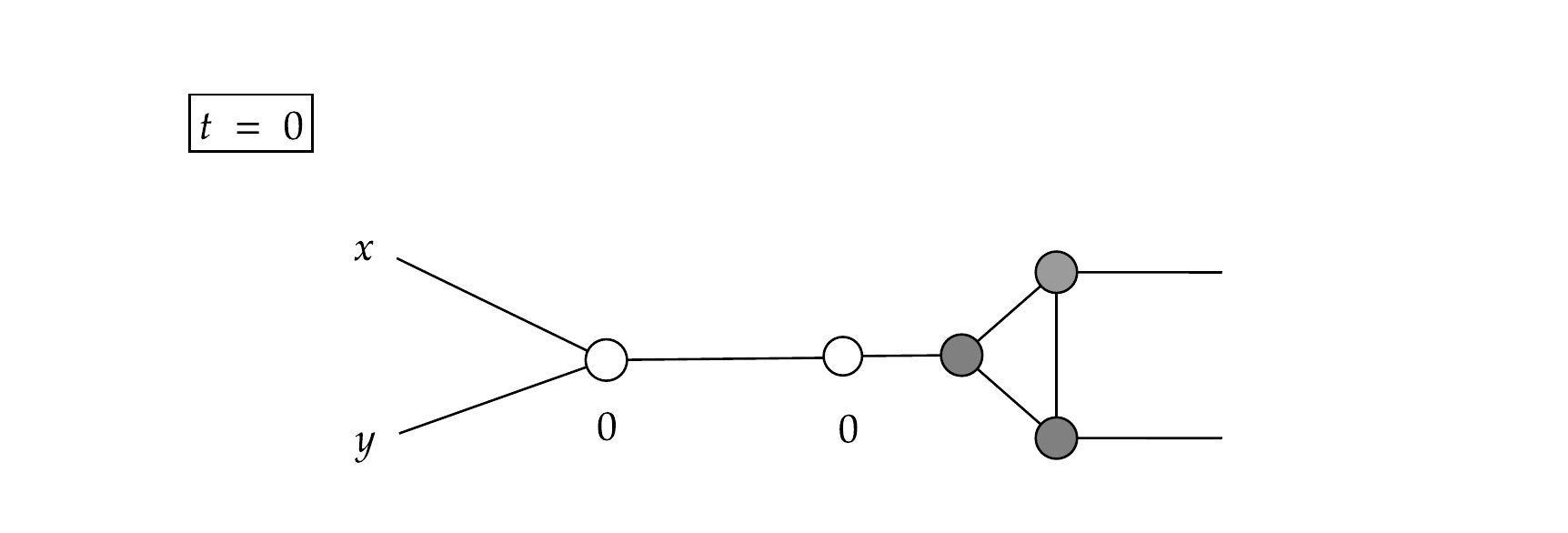}
	\includegraphics[scale=0.45]{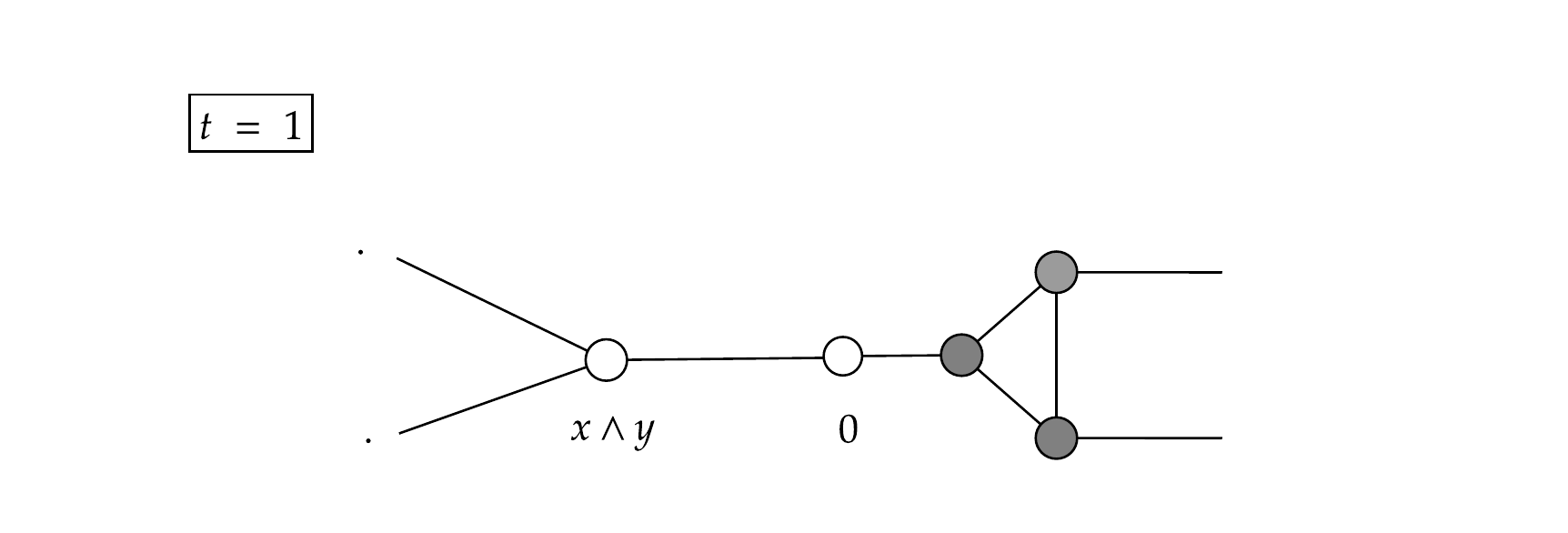}
	\includegraphics[scale=0.45]{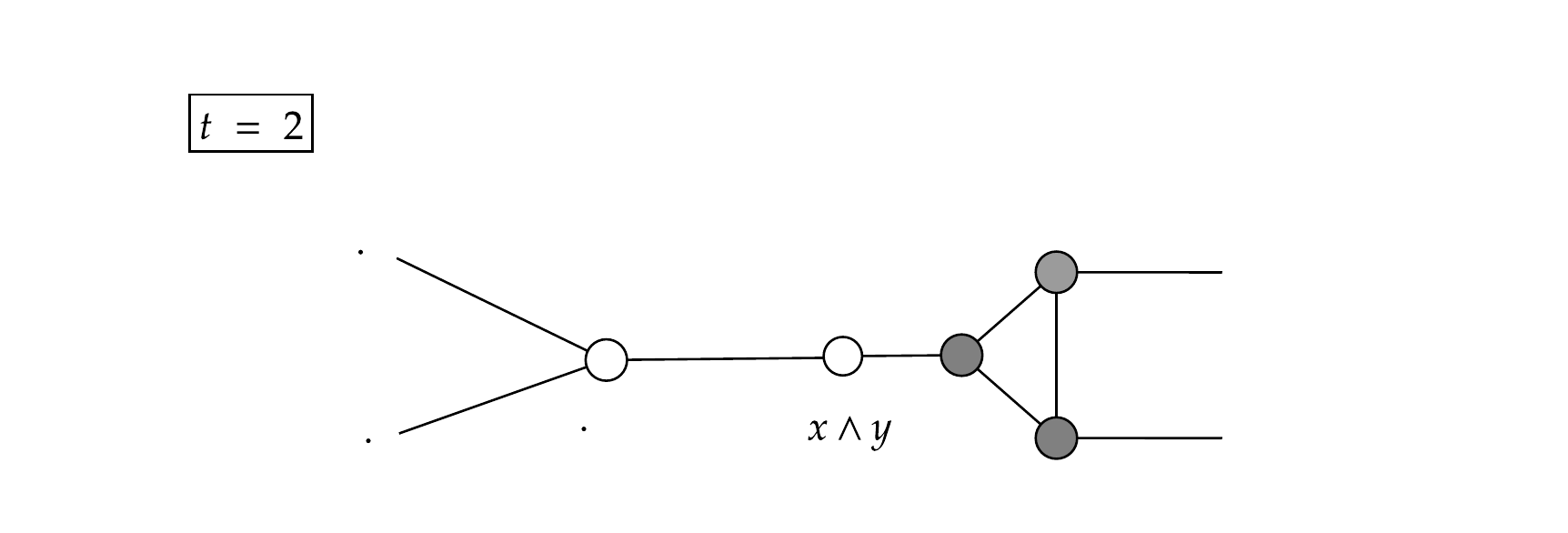}
	\includegraphics[scale=0.45]{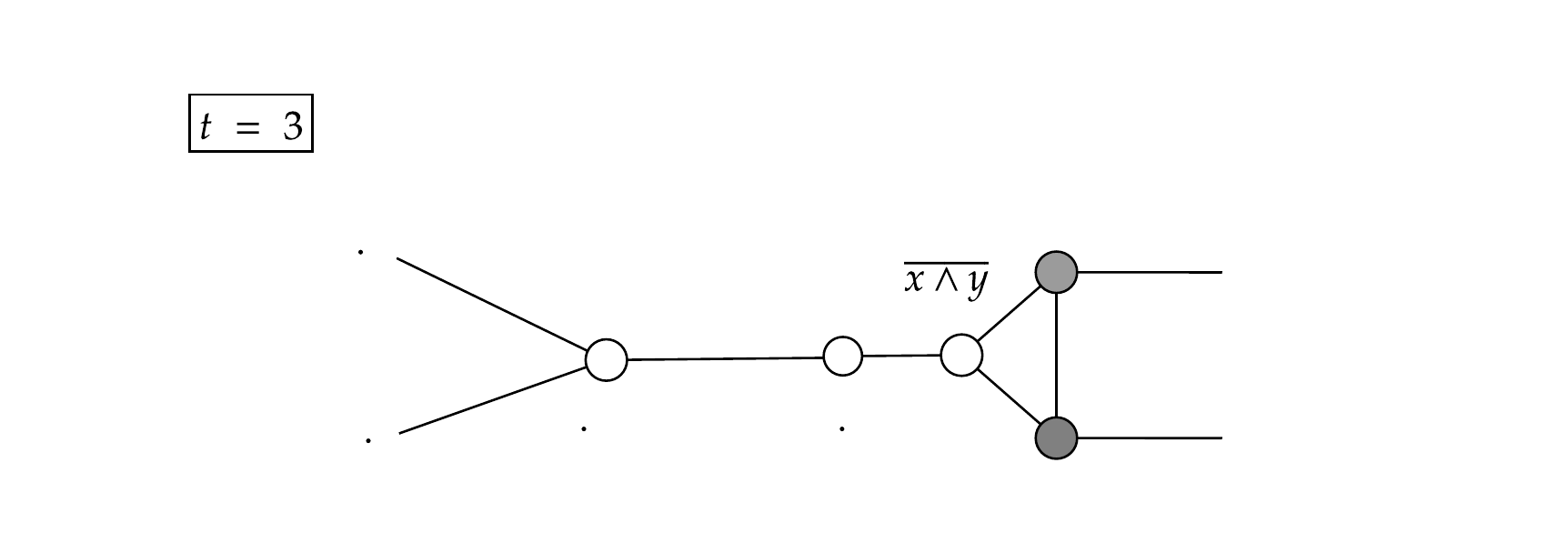}
	\caption{NAND gadget for Rule $2$. $x$ and $y$ nodes represent the inputs and the nodes in the left represent the output. Gray nodes are part of a clique (triangle) fixed in state $1$. Computation takes place in central node and output is read after $4$ time-steps.}
	\label{fig:NAND2}
\end{figure}

\subsection{Interval rules}

In this section we study the case in which the class of totalistic rules is defined by interval rules. In this particular sub-class of automata networks, active values are achieve by reaching an amount of active neighbours in a fixed interval $[\alpha,\beta]$. Of course we are assuming that $\beta$ is strictly less than the degree of each node in the interaction graph. We do this to avoid the case in which all the rules are simply threshold. Results of this section are completely analogous to the results related to isolated rules. In fact, we start by showing a NAND gadget, then we show a clock gadget and finally we show that we can use it to coordinate the evaluation of an arbitrary Boolean circuit.

First, we present a NOT gadget:
\begin{figure}[!tbp]
	\includegraphics[scale=0.35]{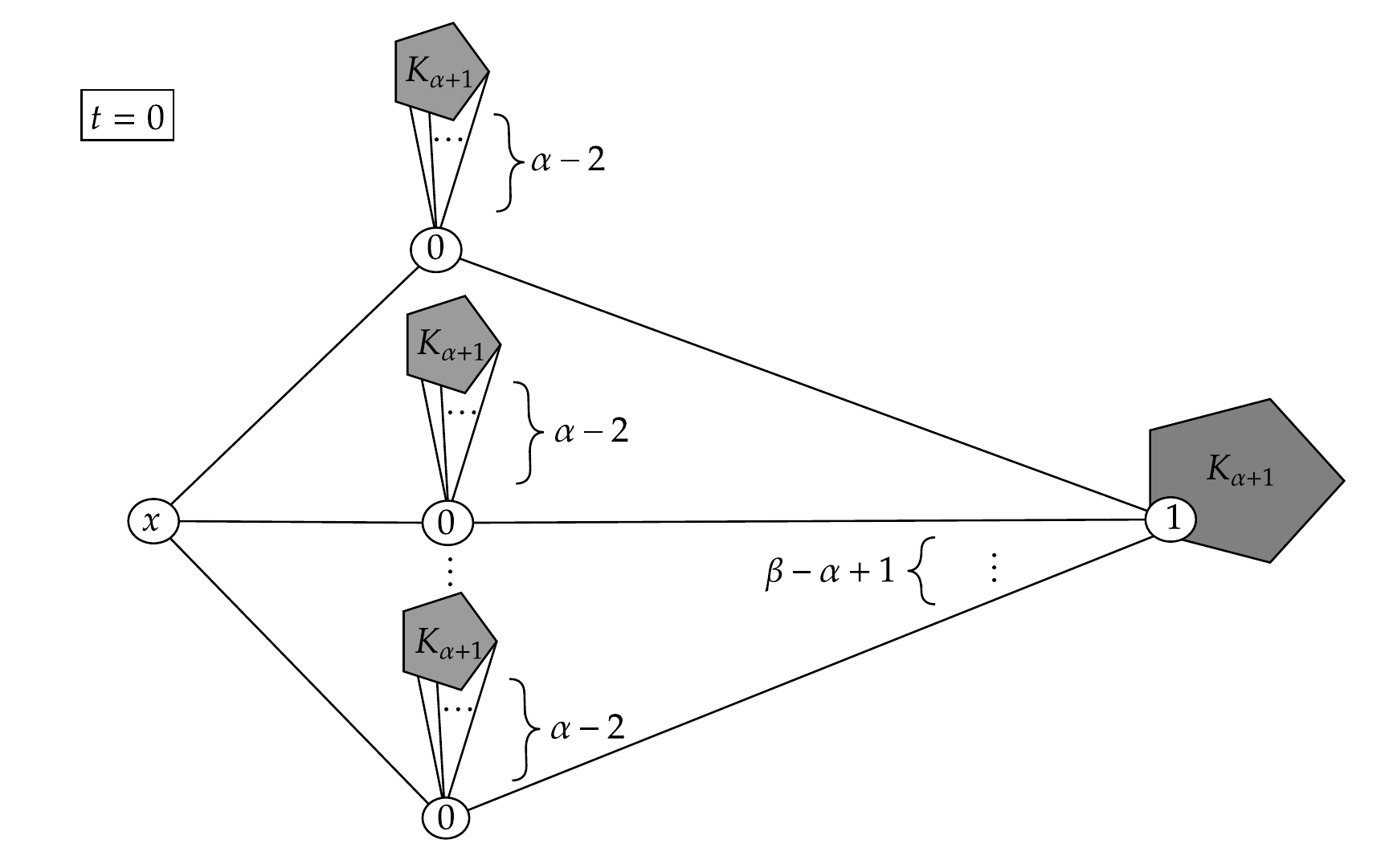}
	\includegraphics[scale=0.35]{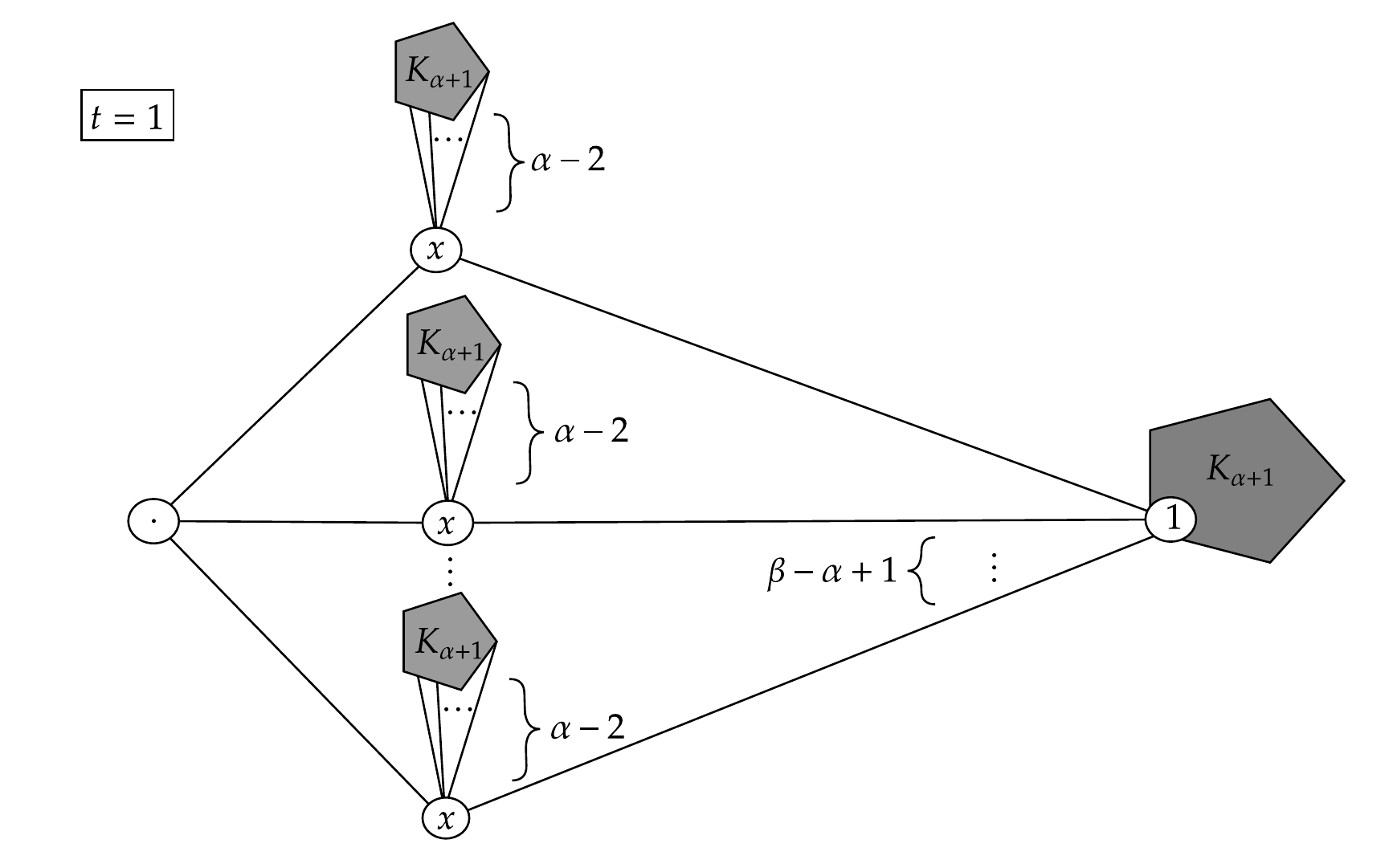}
	\includegraphics[scale=0.35]{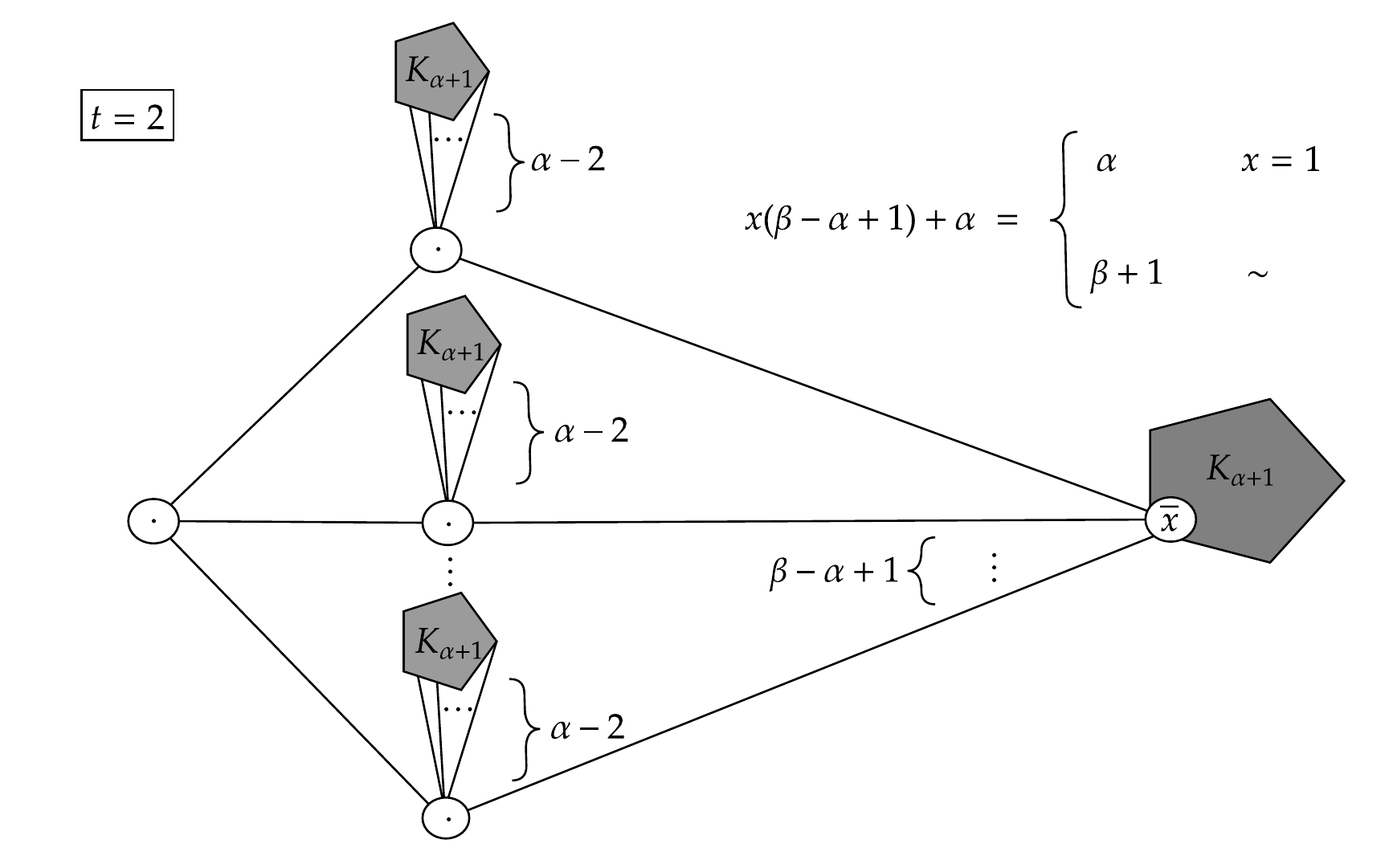}
	\caption{NOT gadget for Interval rules. $x$ and $y$ nodes represent the inputs and the nodes in the left represent the output.}
	\label{fig:NOTinter}
\end{figure}

\begin{lem}
	For each $\alpha, \beta$, $2\leq\alpha\leq \beta$, there is an automata network $\mathcal{A} =(G, \mathcal{F})$ in which every rule is interval with threshold $\alpha$ and $\beta$ and such its global rule $F$ satisfies $F^(x)_o = \overline{x_i}$ for some configuration $x \in Q^{|V(G)|}$.
\end{lem}
\begin{proof}
	See Figure \ref{fig:NOTinter}.
\end{proof}
Then, we combine the latter gadget with other structures in order to generate a NAND gadget.
\begin{lem}
	For each $\alpha, \beta$, $2\leq\alpha\leq \beta$, there is an automata network $\mathcal{A} =(G, \mathcal{F})$ in which every rule is interval with threshold $\alpha$ and $\beta$ and such its global rule $F$ satisfies $F^3(y)_{o_j} =  \textbf{NAND}(y|_{i_1},y|_{i_2})$
\end{lem}
\begin{figure}[!tbp]
	\includegraphics[scale=0.45]{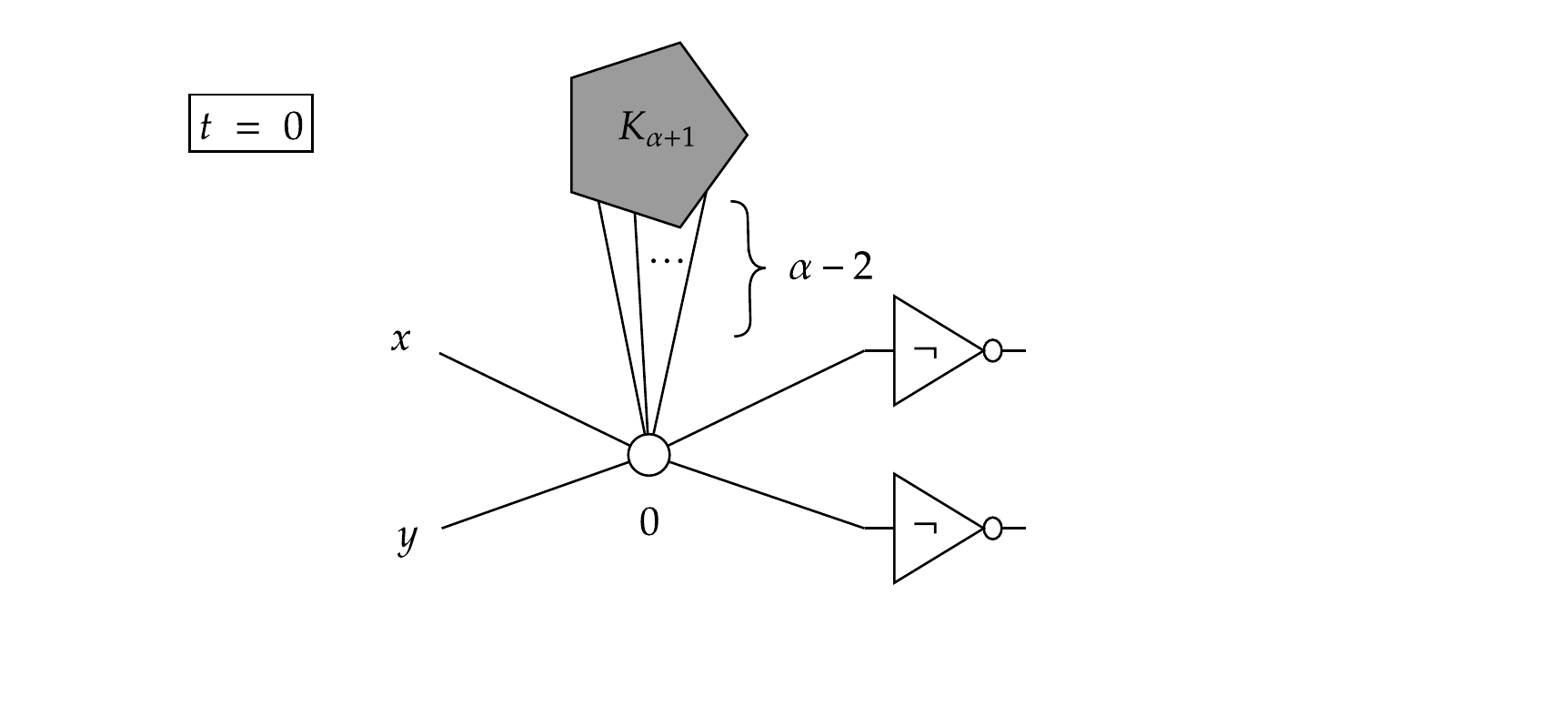}
	\includegraphics[scale=0.45]{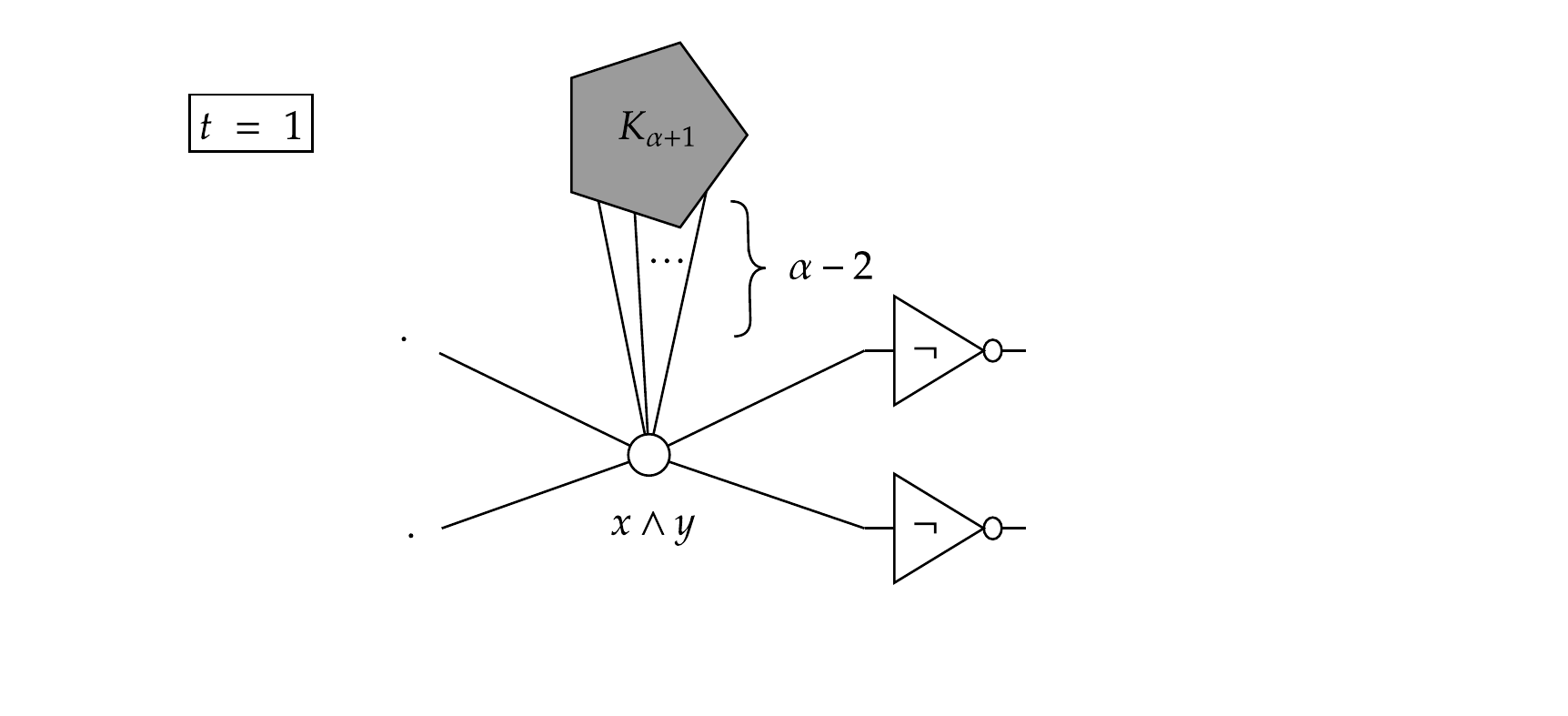}
	\includegraphics[scale=0.45]{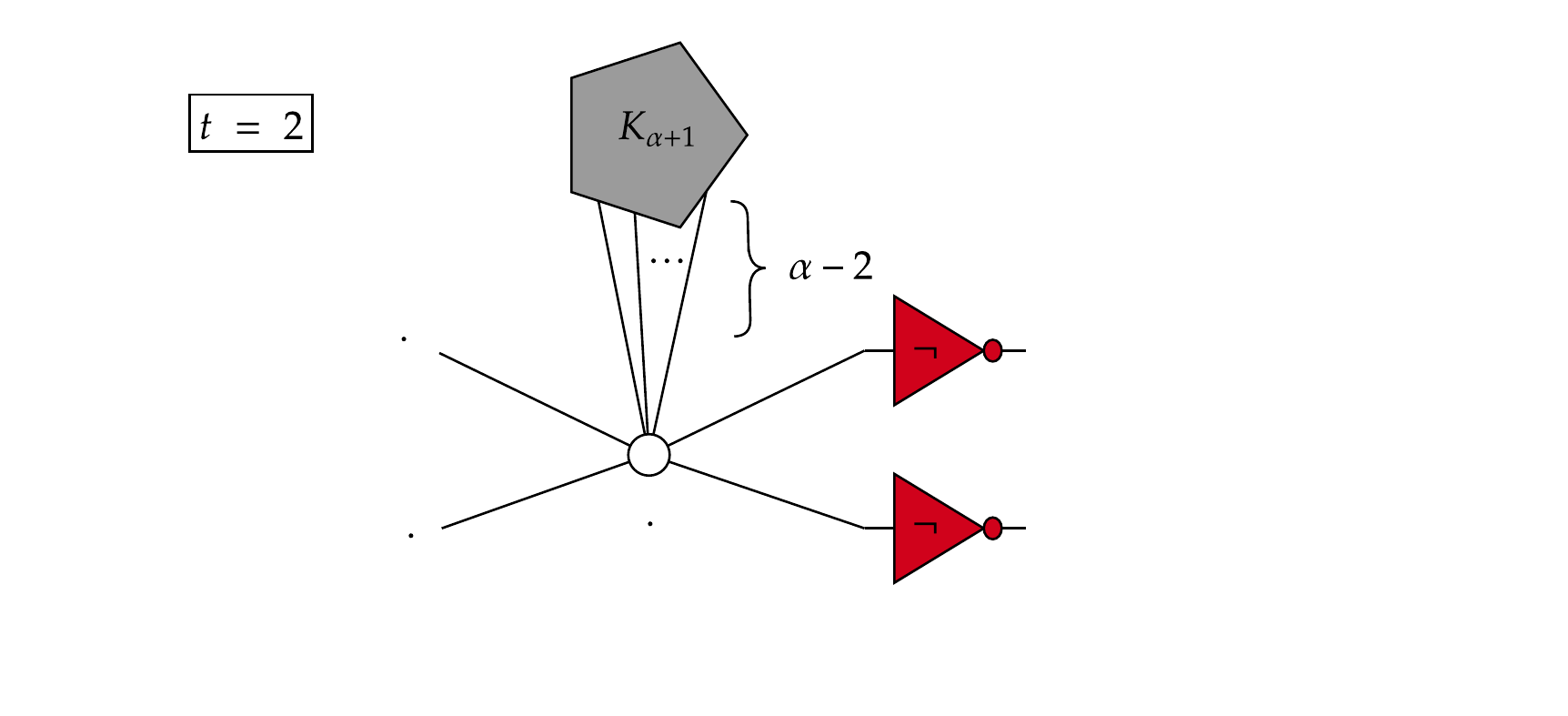}
	\includegraphics[scale=0.45]{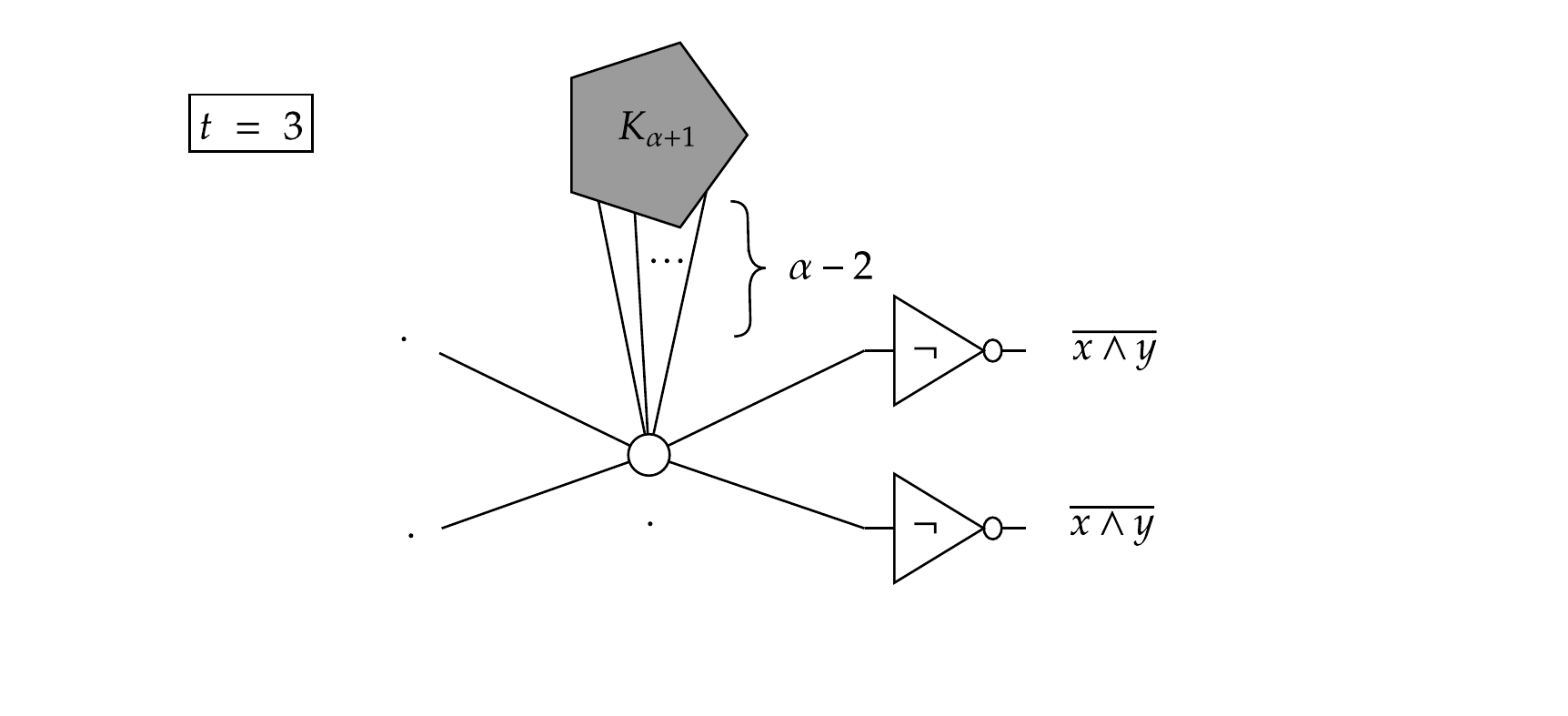}
	\caption{NAND gadget for Interval rules. $x$ and $y$ nodes represent the inputs and the nodes in the left represent the output.}
	\label{fig:NANDinter}
\end{figure}
We show now that we can generate a clock gadget from a wire gadget.
\begin{lem}
	For each $2\geq\alpha\geq\beta$ and $d \geq 1$ there is an automata network $\mathcal{A}_{\alpha,d}= (G_{\alpha,d},\mathcal{F}_{\alpha,d})$ such that every $f \in \mathcal{F}_{\alpha,d}$ is an interval rule with threshold $\alpha$ and $\beta$ and such that its global rule $F_C$ satisfies that there exists $o \in V(G)$: $F_C^s(x)_{o}  = 1$ for $0\leq s \leq d-1$ and $F_C^{d}(x)_{o} = 0$ for some  $x \in \{0,1\}^n$.
\end{lem}
\begin{proof}
	We use the  gadget from Figure \ref{fig:wireinter} to build a clock gadget. This is  analogous to Figure \ref{fig:clockiso}.
\end{proof}
\begin{figure}[!tbp]
	\centering
	\includegraphics[scale=0.5]{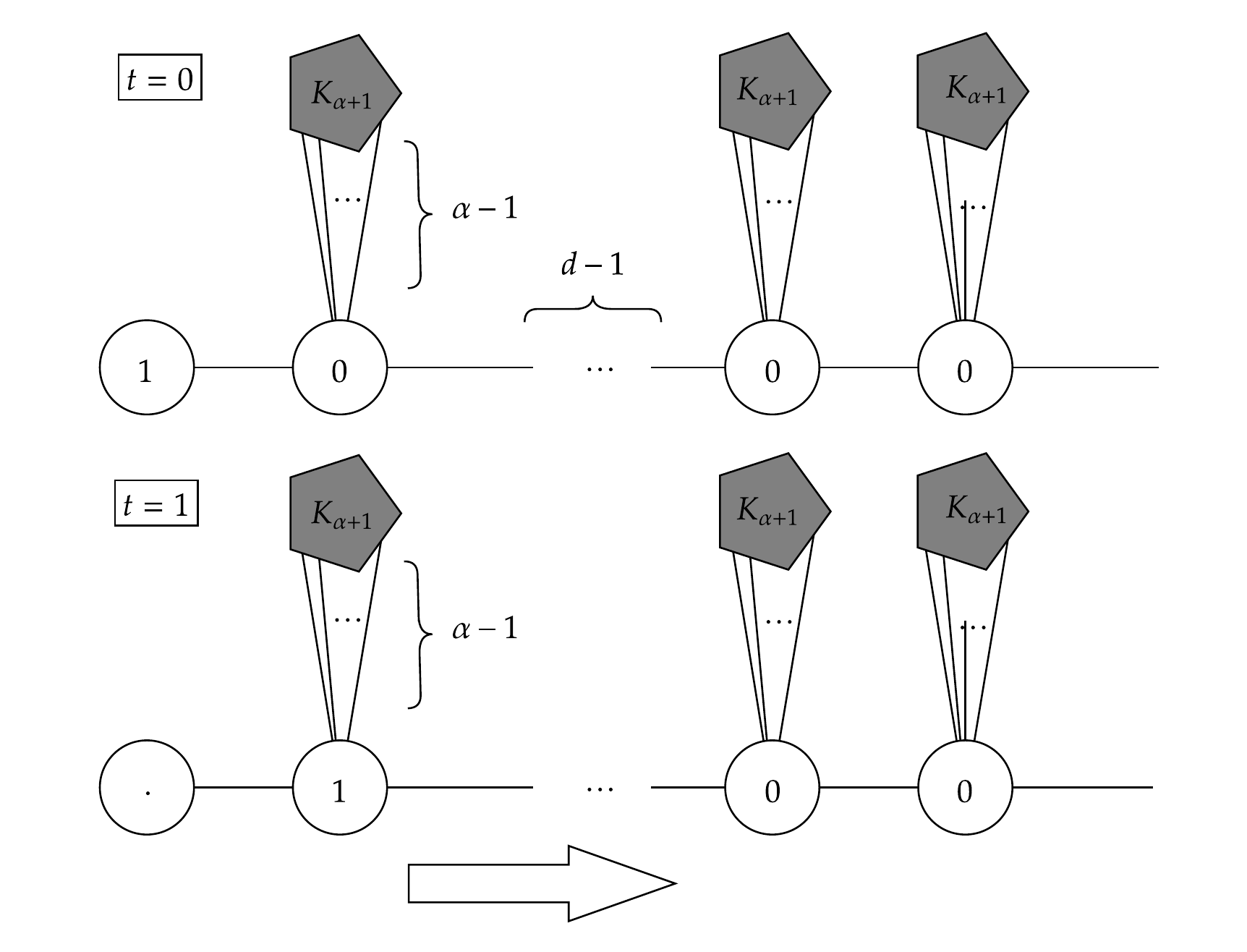}
	\caption{$d$-wire gadget for Interval rules.}
	\label{fig:wireinter}
\end{figure}
\begin{figure}[!tbp]
	\centering
	\includegraphics[scale=0.35]{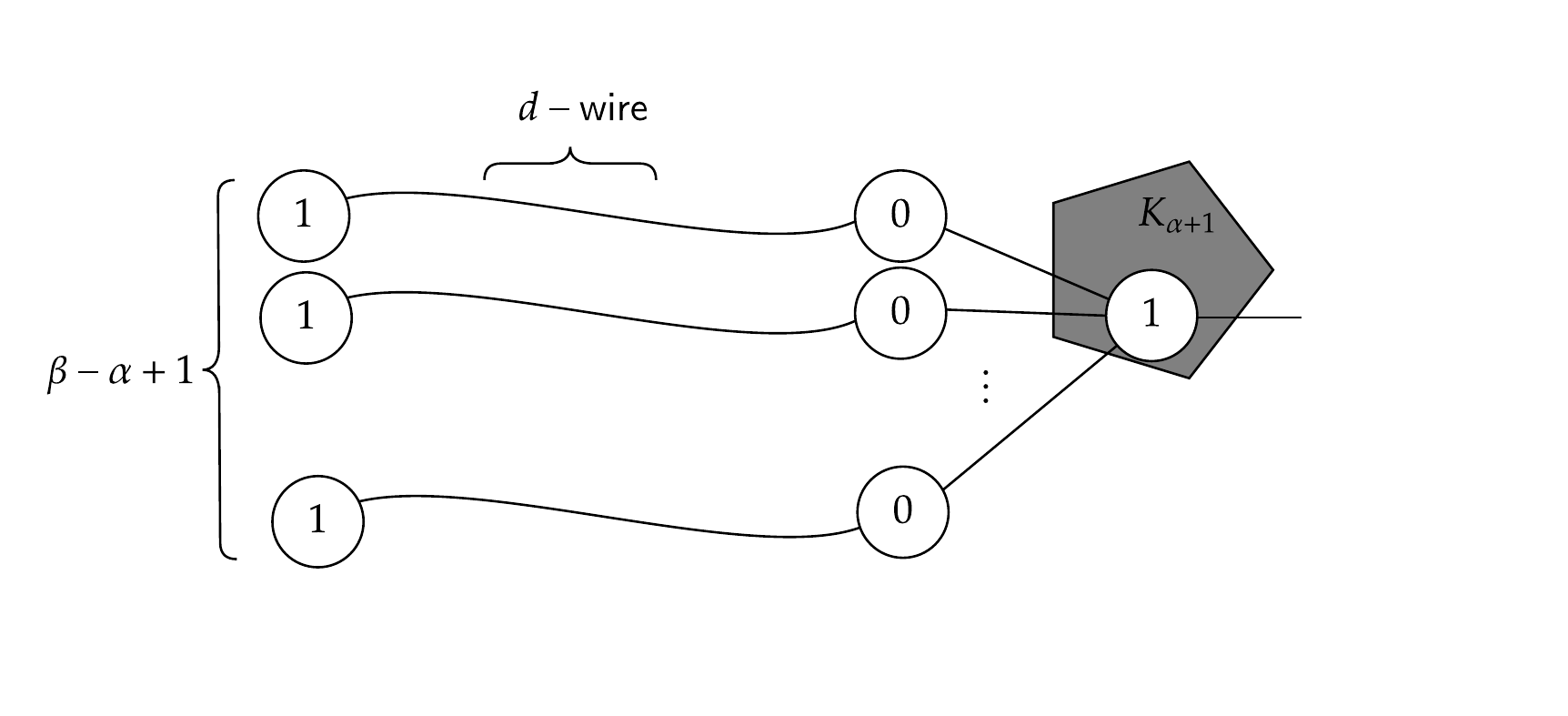}
	\caption{$d$-clock gadget for Interval rules.}
	\label{fig:clock}
\end{figure}
We now introduce a clocked-NAND gadget.
\begin{lem}
For each $2\geq\alpha\geq\beta$ and $d \geq 1$ there exists an automata network $\mathcal{A}_{\alpha,d}= (G_{\alpha,d},\mathcal{F}_{\alpha,d})$ such that every $f \in \mathcal{F}_{\alpha,d}$ is an interval rule with threshold $\alpha$ and $\beta$ and such that its global rule $F_{CN}$ satisfies that there exist $i_1,i_2,o_1,o_2 \in V(G)$ such that $F_{CN}^s(w)_{o_j}  = 1, j=1,2$ for $0\leq s \leq d-1$ and $F_{CN}^{d+3}(x)_{o_j} = \textbf{NAND}(F^d(w)|_{i_1},F^d(w)|_{i_2}), j=1,2$ for some $w \in \{0,1\}^n$
\end{lem}
\begin{figure}[!tbp]
	\centering
	\includegraphics[scale=0.35]{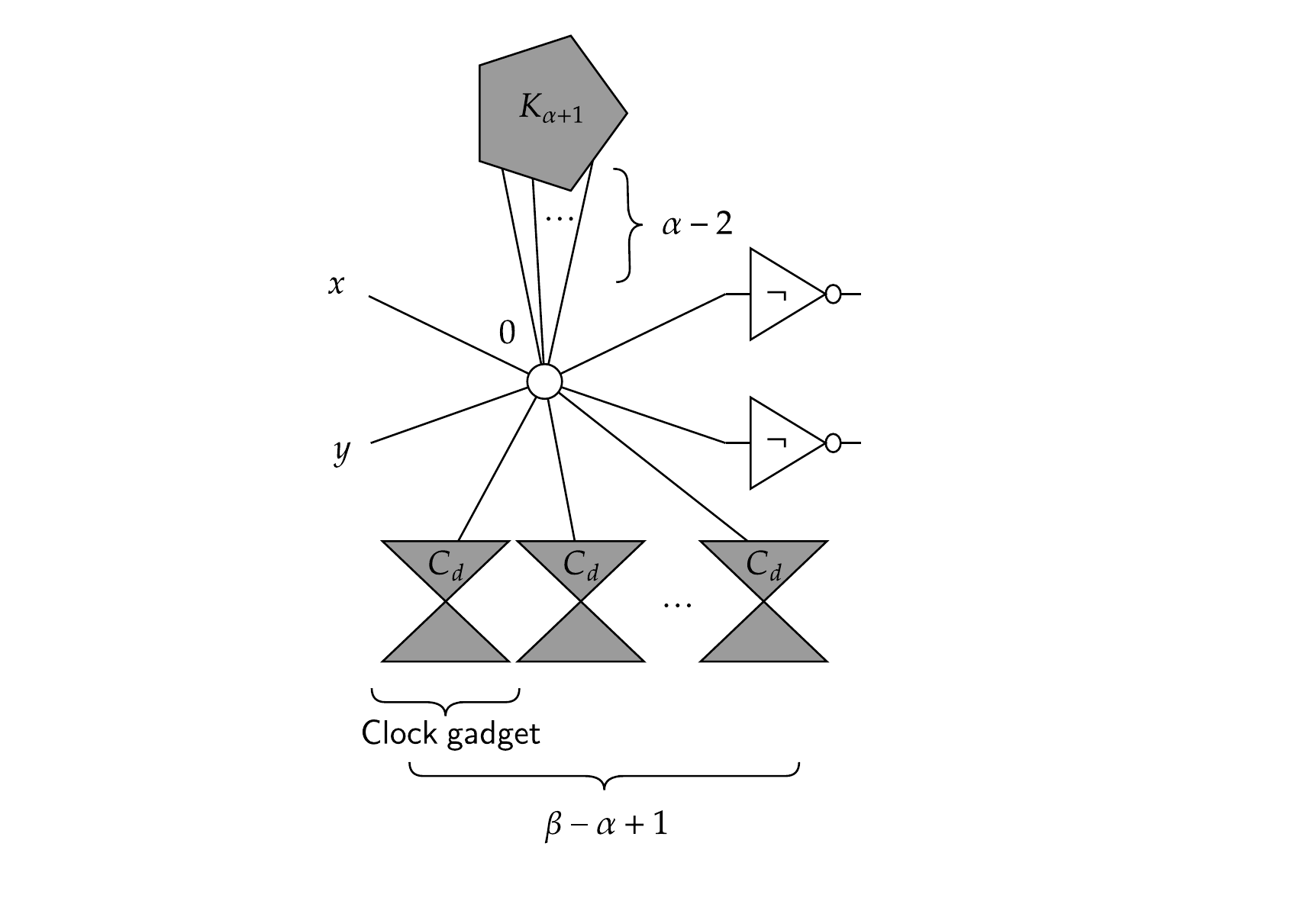}
	\caption{A clocked NAND gadget with delay $d$ for interval rules.}
	\label{fig:clockedNANDinter}
\end{figure}
\begin{proof}
	Gadget is shown in Figure \ref{fig:clockedNANDinter}. This latter gadget is  analogous to clocked NAND gadget for isolated rules (see Figure \ref{fig:clockedNANDiso}.)
\end{proof}
Now we can introduce the main result of the section:
\begin{theo}
	Let $r,s \in \mathbb{N}$ and $f: \{0,1\}^r \to \{0,1\}^s$ a Boolean function. For each $2\leq \alpha \leq \beta$ there exist a set of interval functions $\mathcal{F}$ with interval values [$\alpha$,$\beta$] and a bounded degree class of graphs $\mathcal{G}$ such that $\mathcal{F}$ simulates $f$ in $\mathcal{G}$.
	\label{teo:interval}
\end{theo}
\begin{proof}
	Proof is  analogous to the proof o Theorem \ref{teo:interval}. See Figure \ref{fig:circuitiso} for a scheme of the network simulating an arbitrary Boolean circuit. 
\end{proof}
\begin{remark}
The case in which $\alpha = 1$ is  analogous to Rule $1$.  In fact, the same gadgets can be use to simulate arbitrary Boolean functions with the exception of delayed NOR gadget which is given in Figure \ref{fig:1intervalNOR}
\begin{figure}[!tbp]
\centering
	\includegraphics[scale=0.5]{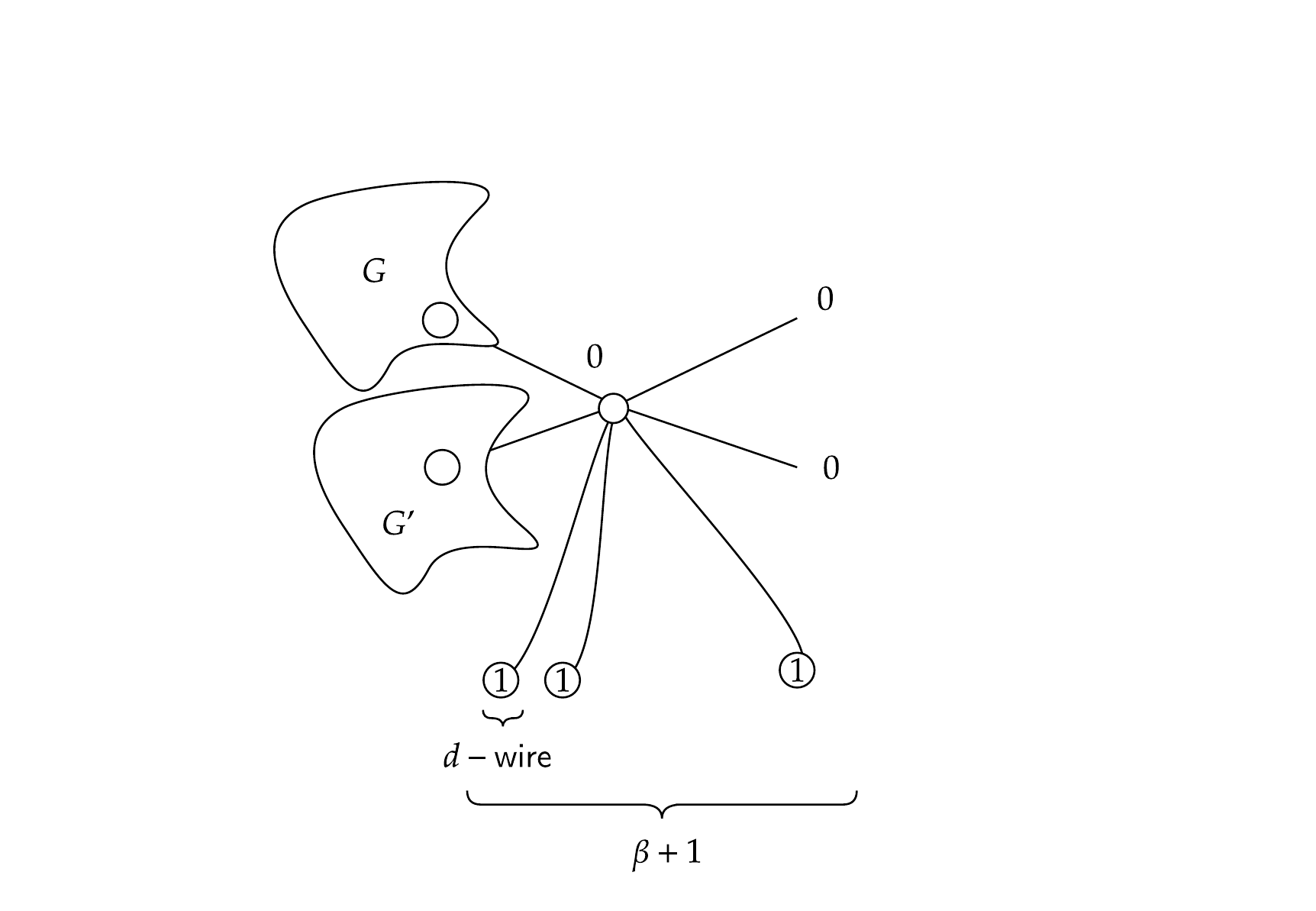}
	\caption{Delayed NOR gadget for $\alpha = 1$.}
	\label{fig:1intervalNOR}
\end{figure}
\end{remark}

\begin{remark}[Planarity of the gadgets.]
 It is important to point out that in most of our gadgets we have to consider a huge amount of fixed states (1), which implies that those subset are in general complete graphs (i.e., the vertex are fully connected) which are for most than 4 vertex non-planar. Actually, by considering the Kuratowski characterisation \cite[Theorem 4.4.6]{diestelgraph} in our constructions we have that for any $\alpha \geq 4$, the NAND gadgets and clocked NAND gadgets are not planar, because from Figure \ref{fig:NANDiso} and Figure \ref{fig:clockiso} we realise that each gadget contains at least $K_{\alpha+1}$ as a subgraph, which is a forbidden subgraph in the Kuratowski characterisation.  
 
\end{remark}
%
%
%

\section{Numerical Experiments}
In this section, first we show some results on simulations of totalistic automata networks over random graphs. More precisely, we generate a collection of $1000$ random graphs using the well known Erdös-Renyi model and we study all the totalistic rules with a maximum of $4$ active neighbours (i.e. $\mathcal{I}_v \subseteq
\{0,\hdots,4\}$ for each node $v$ in the network). 

In general, in this section we study all the possible Boolean gates that each rule in the latter class can calculate in some randomly generated graph, by trying any possible combination of three nodes as a set of two inputs and one output. We are also interested in identify a set of small graphs, that we call, gadgets, which exhibit a richer spectrum of Boolean gates in the latter simulation. We remark that in this section, in order to simplify the simulations, we observe the system capabilities to simulate boolean gates starting from perturbations of the quiescent state $\vec{0}.$ This latter choice could introduce some bias in some of our results because rules which need more than $2$ active neighbours in order to change to state $1$ (e.g. rules $3$, $34$, etc) will probably not show any interesting behaviour (we are introducing changes in only two nodes and the rest of the system stays in state $0$). In order to avoid this misleading effect, we consider, in the second part of this section, a sample of different fixed points and we repeat the previous simulations. However, as we consider that changing the interaction graph and the fixed point might add some extra complexity to the analysis of our results, instead we fix the underlying interaction graph as a small two dimensional grid and we study the simulation capabilities of the system starting from different fixed points. This is also interesting considering that for some of the totalistic rules that we are considering here, as rule $1$, the study of the capabilities of the system in order to simulate boolean circuits is still open in the case in which the interaction graph is the two dimensional grid.

Specifically, this section is organised as follows: 
\begin{enumerate}
	\item  First we show a general collection of results to describe the landscape of simulation capabilities of each totalistic automata network with at most $4$ active neighbours: we exhibit the relative frequency in which each graph can simulat all the possible $2^{2^2} = 16$ Boolean gates and we choose the graph that has the greater spectrum from the collection of randomly generated graph (that is to say that exhibits more logic gates from the set of $16$ possible Boolean gates with $2$ inputs and $1$ output). In this regard, we study the impact of the connectivity of each random generated graph by changing the probability of two given nodes to be connected. We call this probability $p \in [0,1]$.
	\item From the latter simulations, we choose one rule from the set of totalistic rules (rule $1$) that we have considered in the latter subsection and we exhibit the most representative gadget (i.e. the smallest graph with greatest spectrum)). Then, we study its dynamics, with emphasis in understanding how it simulates a given a Boolean gate.
	\item We extend the study of this specific gadget into the study of a particular one dimensional celullar automaton inspired in its topology.
	\item We fix a small two dimensional grid (4x4) and we generated a sample of fixed points for each rule. Then, we repeat previous simulations but now different fixed points play the role of random graphs of the latter sections.
\end{enumerate}

\subsection{General landscape in random graphs}
We start by remarking that as we use the Erdös-Renyi model, two parameters must be chosen in order to define a random graph: the number of nodes $n$ and the probability $p$ of two arbitrary nodes to be connected. We also recall that we are considering only logic gates with two inputs and one output. As we can encode each of these Boolean functions as a sequence of $4$ bits, we can represent each function by a number $i \in \{0,\hdots,15\}.$  For example, if $i=8$ then we have the AND rule since $8 = 0001$ and we are considering the following coding: $00 \to 0$, $01 \to 0$ $10 \to 0$ and $11 \to 1$. Analogously, OR gate is $i=14$,  XOR gate is $i=6$, NAND gate is $i=7$ and NOR gate is $i=1.$ Finally, we start all the simulations from the fixed point $\overline{x} = \vec{0}$ and we choose perturbations in all the possible assignations for two inputs and one output over a randomly generated graph. 

We are now in condition of summarize the set-up of parameters that we used for the simulations.

\subsection{Simulations set-up}
We start by describing the parameters of the following simulations:

\begin{enumerate}
	\item Probability of adjacency of two given nodes ($p$): 0.1, 0.5, 0.8
	\item Number of generated graphs ($N$): 100
	\item Number of nodes per graph ($n$): 10
	\item Simulation time $t$: 100	
\end{enumerate}
\subsubsection{Results}
\paragraph{The impact of connectivity.}

One of the most straightforwards observations from the results we show in Figures \ref{plot:p01}, \ref{plot:p05} and \ref{plot:p08} is that there is an effect of the connectivity of the different random graphs in the diversity of Boolean gates that certain rules can simulate. More precisely, if we see Table \ref{tab:p01} we observe that rules that need at least $2$ neighbours in order to activate one node (i.e. $2,23,24,3,34,$ etc ) can only simulate the trivial function (i.e. all inputs goes to $0$). Of course, this is intrinsically related to the definition of the dynamics which depends on the number of active neighbours to produce any dynamic behavior different from $\vec{0}$. So, in those cases, in order to study the spectrum we have, if possible, to consider other fixed points. In this context, in Table \ref{tab:p01} and \ref{tab:p05} we can observe that most of the gates are simulated by rules that have the possibility to change to state $1$ when they have at least one active neighbour. \\

\begin{table}[H]
	\resizebox{\textwidth}{!}{%
		\begin{tabular}{|l|l|l|l|l|l|l|l|l|l|l|l|l|l|l|l|l|}
			\hline
			& 0   & 1 & 2 & 3 & 4  & 5 & 6  & 7 & 8 & 9 & 10 & 11 & 12 & 13 & 14 & 15 \\ \hline
			1    & 100 & 0 & 6 & 0 & 10 & 0 & 18 & 0 & 7 & 0 & 42 & 0  & 44 & 0  & 12 & 0  \\ \hline
			2    & 100 & 0 & 0 & 0 & 0  & 0 & 0  & 0 & 0 & 0 & 0  & 0  & 0  & 0  & 0  & 0  \\ \hline
			3    & 100 & 0 & 0 & 0 & 0  & 0 & 0  & 0 & 0 & 0 & 0  & 0  & 0  & 0  & 0  & 0  \\ \hline
			4    & 100 & 0 & 0 & 0 & 0  & 0 & 0  & 0 & 0 & 0 & 0  & 0  & 0  & 0  & 0  & 0  \\ \hline
			12   & 100 & 0 & 6 & 0 & 8  & 0 & 1  & 0 & 0 & 0 & 79 & 0  & 75 & 0  & 38 & 0  \\ \hline
			13   & 100 & 0 & 1 & 0 & 1  & 0 & 33 & 0 & 1 & 0 & 51 & 0  & 53 & 0  & 1  & 0  \\ \hline
			14   & 100 & 0 & 6 & 0 & 9  & 0 & 20 & 0 & 7 & 0 & 43 & 0  & 45 & 0  & 12 & 0  \\ \hline
			23   & 100 & 0 & 0 & 0 & 0  & 0 & 0  & 0 & 0 & 0 & 0  & 0  & 0  & 0  & 0  & 0  \\ \hline
			24   & 100 & 0 & 0 & 0 & 0  & 0 & 0  & 0 & 0 & 0 & 0  & 0  & 0  & 0  & 0  & 0  \\ \hline
			34   & 100 & 0 & 0 & 0 & 0  & 0 & 0  & 0 & 0 & 0 & 0  & 0  & 0  & 0  & 0  & 0  \\ \hline
			123  & 100 & 0 & 1 & 0 & 2  & 0 & 0  & 0 & 0 & 0 & 83 & 0  & 82 & 0  & 50 & 0  \\ \hline
			124  & 100 & 0 & 8 & 0 & 8  & 0 & 2  & 0 & 1 & 0 & 80 & 0  & 76 & 0  & 39 & 0  \\ \hline
			134  & 100 & 0 & 3 & 0 & 2  & 0 & 34 & 0 & 3 & 0 & 52 & 0  & 54 & 0  & 4  & 0  \\ \hline
			234  & 100 & 0 & 0 & 0 & 0  & 0 & 0  & 0 & 0 & 0 & 0  & 0  & 0  & 0  & 0  & 0  \\ \hline
			1234 & 100 & 0 & 1 & 0 & 1  & 0 & 0  & 0 & 0 & 0 & 84 & 0  & 83 & 0  & 51 & 0  \\ \hline
		\end{tabular}%
	}
	\caption{Spectrum of Boolean gates by totalistic rule for $p=0.1$. Each column contains the number of graphs that are capable of simulate the corresponding Boolean gate.}
	\label{tab:p01}
\end{table}
\begin{figure}[!tbp]
	\centering
	\includegraphics[scale=0.20]{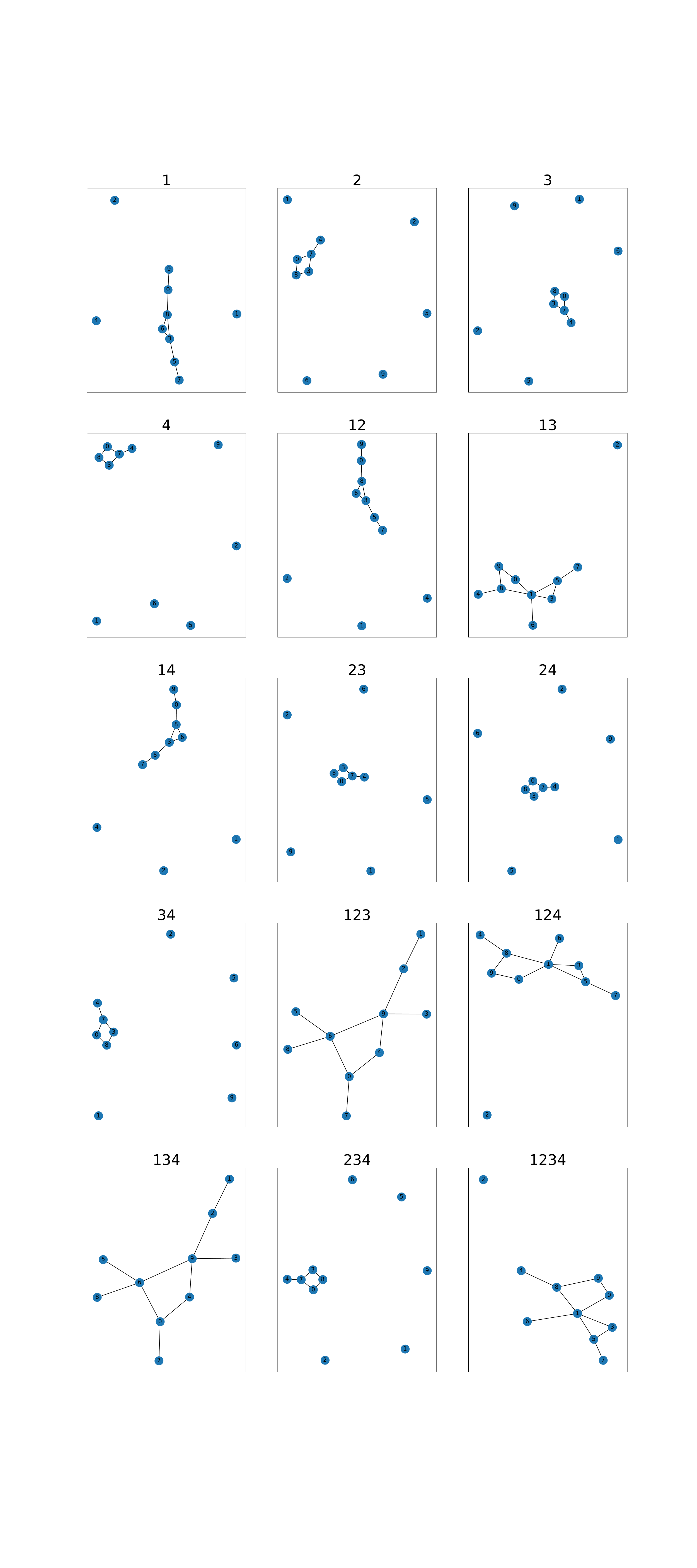}
	\caption{Graphs with the greatest spectrum by rule for $p=0.1.$}
	\label{plot:p01}
\end{figure}
On the other hand, if we observe Table \ref{tab:p05} we can see that rules with higher threshold (that is to say that needs at least 2 active neighbours in order to change to state $1$) start to show some simulation capabilities. In particular, rules that change with $2$ active neighbours exhibit the possibility of calculate AND gates supporting the remark that higher connectivity implies higher simulation capabilities for some rule. \\
Finally, if we observe Figures \ref{plot:p01}, \ref{plot:p05} and \ref{plot:p08} we see that if $p \geq 0.5$ we have that every gadget is connected. Contrarily, in the case $p=0.1$ we observe that rules with higher threshold (those who need more active neighbours to active its nodes) tend to exhibit gadgets with bigger connected components. We remark the case of rule $1$ which seems to exhibit a minimal modular structure which is able to calculate several different logic gates. 
\begin{table}[H]
	\resizebox{\textwidth}{!}{%
		\begin{tabular}{|l|l|l|l|l|l|l|l|l|l|l|l|l|l|l|l|l|}
			\hline
			& 0   & 1 & 2  & 3 & 4  & 5 & 6  & 7 & 8   & 9 & 10 & 11 & 12 & 13 & 14  & 15 \\ \hline
			1    & 100 & 0 & 74 & 0 & 74 & 0 & 63 & 0 & 82  & 0 & 60 & 0  & 56 & 0  & 44  & 0  \\ \hline
			2    & 100 & 0 & 0  & 0 & 0  & 0 & 0  & 0 & 58  & 0 & 0  & 0  & 0  & 0  & 0   & 0  \\ \hline
			3    & 100 & 0 & 0  & 0 & 0  & 0 & 0  & 0 & 0   & 0 & 0  & 0  & 0  & 0  & 0   & 0  \\ \hline
			4    & 100 & 0 & 0  & 0 & 0  & 0 & 0  & 0 & 0   & 0 & 0  & 0  & 0  & 0  & 0   & 0  \\ \hline
			12   & 100 & 0 & 86 & 0 & 88 & 0 & 86 & 0 & 79  & 0 & 83 & 0  & 80 & 0  & 91  & 0  \\ \hline
			13   & 100 & 0 & 94 & 0 & 94 & 0 & 91 & 0 & 96  & 0 & 89 & 0  & 90 & 0  & 83  & 0  \\ \hline
			14   & 100 & 0 & 83 & 0 & 85 & 0 & 78 & 0 & 90  & 0 & 81 & 0  & 81 & 0  & 71  & 0  \\ \hline
			23   & 100 & 0 & 0  & 0 & 0  & 0 & 0  & 0 & 83  & 0 & 0  & 0  & 0  & 0  & 0   & 0  \\ \hline
			24   & 100 & 0 & 0  & 0 & 0  & 0 & 0  & 0 & 70  & 0 & 0  & 0  & 0  & 0  & 0   & 0  \\ \hline
			34   & 100 & 0 & 0  & 0 & 0  & 0 & 0  & 0 & 0   & 0 & 0  & 0  & 0  & 0  & 0   & 0  \\ \hline
			123  & 99  & 0 & 78 & 0 & 77 & 0 & 82 & 0 & 71  & 0 & 75 & 0  & 73 & 0  & 100 & 0  \\ \hline
			124  & 100 & 0 & 93 & 0 & 96 & 0 & 94 & 0 & 95  & 0 & 93 & 0  & 95 & 0  & 97  & 0  \\ \hline
			134  & 100 & 0 & 92 & 0 & 92 & 0 & 95 & 0 & 97  & 0 & 92 & 0  & 94 & 0  & 97  & 0  \\ \hline
			234  & 100 & 0 & 0  & 0 & 0  & 0 & 0  & 0 & 100 & 0 & 0  & 0  & 0  & 0  & 0   & 0  \\ \hline
			1234 & 94  & 0 & 53 & 0 & 50 & 0 & 56 & 0 & 40  & 0 & 46 & 0  & 46 & 0  & 100 & 0  \\ \hline
		\end{tabular}%
	}
	\caption{Spectrum of Boolean gates by totalistic rule for $p=0.5$. Each column contains the number of graphs that can simulate the corresponding Boolean gate.\}}
	\label{tab:p05}
\end{table}
\begin{figure}[!tbp]
	\centering
	\includegraphics[scale=0.20]{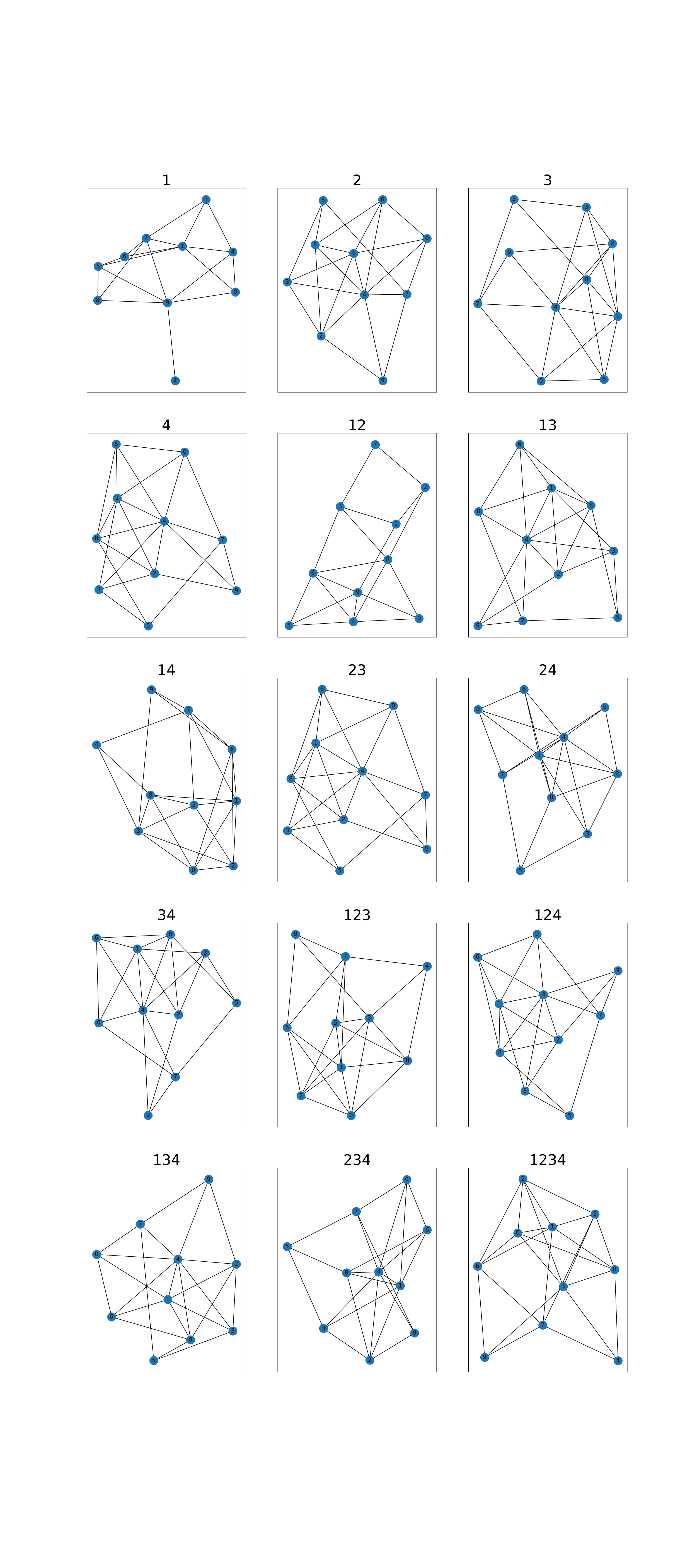}
	\caption{Graphs with the greatest spectrum by rule for $p=0.5$}
	\label{plot:p05}
\end{figure}

\paragraph{More connectivity does not imply more diversity.}
If we study now rules that exhibit strong simulation capabilities even for low connectivity such as rule $1$, we observe that they do not show a significant change in their spectrum for higher values of $p$ (even if we can observe an increase in the frequencies of each Boolean gate in their spectrum, which implies that more graphs of the random sample are being capable of simulating one specific gate). Roughly, this observation may suggest that probably we need to study graphs with more nodes  to observe a significant impact of connectivity in the spectrum of different totalistic rules. This is also coherent with the theoretical results regarding isolated totalistic networks in which we have theoretically constructed gadgets requiring  big cliques subgraphs fixed in the state $1$ in order to perform calculations. \\ 

\paragraph{Complete Boolean gates set are not modular.}
Finally, we observe that gates NOR and NAND ($i=7$ and $i=1$ respectively in our notation) do not seem to be calculable by totalistic rules with at most $4$ active neighbours (see Table \ref{tab:p01}, Table \ref{tab:p05} and Table \ref{tab:p08}). It might suggest that we need a more complex structure in order to simulate complete sets of Boolean gates. We conjecture that maybe we can simulate them by ``glueing'' different gadgets in a coherent way. Nevertheless, in order to achieve this task, we need a better understanding of the particular dynamics of certain gadgets (see next subsection). 

\begin{table}[H]
	\resizebox{\textwidth}{!}{%
		\begin{tabular}{|l|l|l|l|l|l|l|l|l|l|l|l|l|l|l|l|l|}
			\hline
			& 0   & 1 & 2  & 3 & 4  & 5 & 6  & 7 & 8  & 9 & 10 & 11 & 12 & 13 & 14 & 15 \\ \hline
			1    & 100 & 0 & 1  & 0 & 1  & 0 & 0  & 0 & 39 & 0 & 1  & 0  & 1  & 0  & 0  & 0  \\ \hline
			2    & 100 & 0 & 0  & 0 & 0  & 0 & 0  & 0 & 27 & 0 & 0  & 0  & 0  & 0  & 0  & 0  \\ \hline
			3    & 100 & 0 & 0  & 0 & 0  & 0 & 0  & 0 & 0  & 0 & 0  & 0  & 0  & 0  & 0  & 0  \\ \hline
			4    & 100 & 0 & 0  & 0 & 0  & 0 & 0  & 0 & 0  & 0 & 0  & 0  & 0  & 0  & 0  & 0  \\ \hline
			12   & 100 & 0 & 11 & 0 & 10 & 0 & 8  & 0 & 0  & 0 & 0  & 0  & 0  & 0  & 2  & 0  \\ \hline
			13   & 100 & 0 & 40 & 0 & 41 & 0 & 28 & 0 & 81 & 0 & 34 & 0  & 32 & 0  & 18 & 0  \\ \hline
			14   & 100 & 0 & 61 & 0 & 61 & 0 & 48 & 0 & 91 & 0 & 54 & 0  & 53 & 0  & 38 & 0  \\ \hline
			23   & 100 & 0 & 0  & 0 & 0  & 0 & 0  & 0 & 66 & 0 & 0  & 0  & 0  & 0  & 0  & 0  \\ \hline
			24   & 100 & 0 & 0  & 0 & 0  & 0 & 0  & 0 & 83 & 0 & 0  & 0  & 0  & 0  & 0  & 0  \\ \hline
			34   & 100 & 0 & 0  & 0 & 0  & 0 & 0  & 0 & 0  & 0 & 0  & 0  & 0  & 0  & 0  & 0  \\ \hline
			123  & 99  & 0 & 29 & 0 & 27 & 0 & 18 & 0 & 11 & 0 & 14 & 0  & 13 & 0  & 14 & 0  \\ \hline
			124  & 100 & 0 & 54 & 0 & 54 & 0 & 39 & 0 & 26 & 0 & 29 & 0  & 24 & 0  & 22 & 0  \\ \hline
			134  & 99  & 0 & 66 & 0 & 63 & 0 & 52 & 0 & 90 & 0 & 58 & 0  & 53 & 0  & 41 & 0  \\ \hline
			234  & 100 & 0 & 0  & 0 & 0  & 0 & 0  & 0 & 91 & 0 & 0  & 0  & 0  & 0  & 0  & 0  \\ \hline
			1234 & 96  & 0 & 63 & 0 & 65 & 0 & 48 & 0 & 38 & 0 & 33 & 0  & 32 & 0  & 37 & 0  \\ \hline
		\end{tabular}%
	}
	\caption{Spectrum of Boolean gates by totalistic rule for $p=0.8$. Each column contains the number of graphs that can simulate the corresponding Boolean gate.\}}
	\label{tab:p08}
\end{table}
\begin{figure}[!tbp]
	\centering
	\includegraphics[scale=0.20]{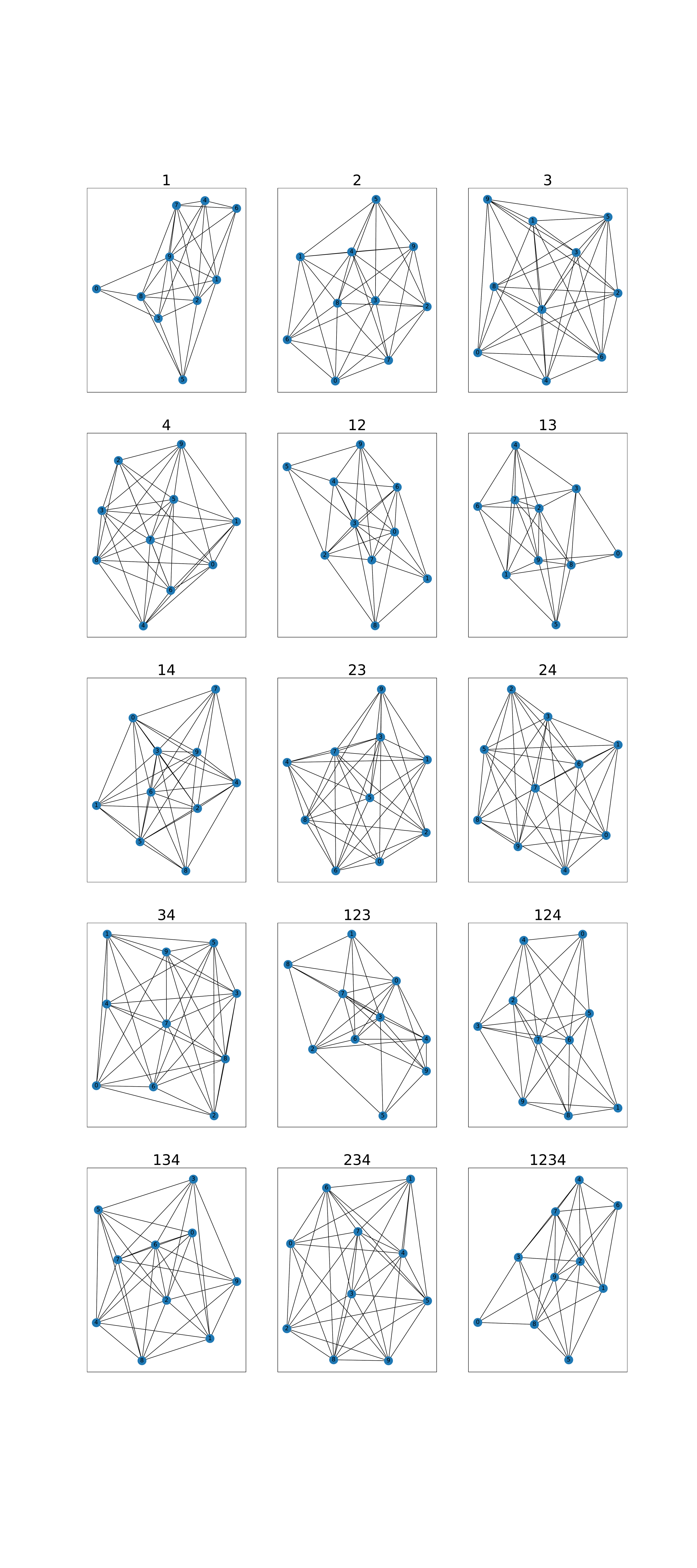}
	\caption{Graphs with the greatest spectrum by rule for $p=0.8$}
	\label{plot:p08}
\end{figure}
\subsection{Generating boolean gates in the two dimensional grid.}
In this subsection, in order to consider fixed points different from $\vec{0}$ (if there exists)we repeat the same computational experiments that we did before with different random graphs but this time, we fix the graph and we change the fixed point we are perturbing in order to generate boolean gates. More precisely we consider a two dimensional grid and for each rule we consider a sample of fixed points that we previously generated by simply simulating the system and waiting for it to attain a fixed point. This simulation time was previously established as $t=100$. Once the fixed point sample is obtained for each rule, we try any combination of input and output in order to observe the boolean gates that system is capable of generating as a result of perturbing given inputs. Then, as same as we did in the latter section, we show the fixed points that have shown the greater spectrum (that have exhibit the greater amount of boolean gates) and we discuss a possible link between its structure and its simulation capabilities.
\subsection{Simulations set-up}
\begin{enumerate}
    \item Grid size: $4 \times 4$
    \item Simulation time: $100$
    \item Rule set: any totalistic rule up to a maximum activation value (number of neighbours in state $1$) of $\Delta = 4.$
\end{enumerate}
\subsection{Results}
\begin{table}[H]
\centering
\begin{tabular}{|l|l|}
\hline
Rules & \#Fixed Points \\ \hline
1     & 41             \\ \hline
2     & 57             \\ \hline
3     & 9              \\ \hline
4     & 2              \\ \hline
12    & 9              \\ \hline
13    & 1              \\ \hline
14    & 58             \\ \hline
23    & 57             \\ \hline
24    & 74             \\ \hline
34    & 34             \\ \hline
123   & 25             \\ \hline
124   & 58             \\ \hline
134   & 74             \\ \hline
234   & 34             \\ \hline
1234  & 2              \\ \hline
\end{tabular}
\caption{Number of fixed points found for each totalistic rule after $t=100$ time steps of simulation.}
\label{tab:nfixedpoints}
\end{table}
\begin{table}[]
\resizebox{\textwidth}{!}{%
\begin{tabular}{|l|l|l|l|l|l|l|l|l|l|l|l|l|l|l|l|l|}
\hline
     & 0   & 1  & 2  & 3  & 4  & 5  & 6  & 7  & 8   & 9  & 10 & 11 & 12 & 13 & 14  & 15 \\ \hline
1    & 100 & 98 & 76 & 98 & 59 & 98 & 78 & 0  & 78  & 0  & 78 & 0  & 78 & 0  & 83  & 98 \\ \hline
2    & 100 & 70 & 21 & 88 & 11 & 86 & 21 & 95 & 56  & 0  & 42 & 0  & 42 & 0  & 96  & 98 \\ \hline
3    & 100 & 78 & 78 & 78 & 56 & 78 & 0  & 0  & 89  & 0  & 89 & 0  & 89 & 0  & 0   & 89 \\ \hline
4    & 50  & 0  & 0  & 0  & 0  & 0  & 0  & 0  & 50  & 0  & 50 & 0  & 50 & 0  & 0   & 50 \\ \hline
12   & 100 & 0  & 89 & 0  & 67 & 0  & 0  & 0  & 0   & 89 & 89 & 0  & 89 & 0  & 0   & 89 \\ \hline
13   & 100 & 0  & 0  & 0  & 0  & 0  & 0  & 0  & 0   & 0  & 0  & 0  & 0  & 0  & 0   & 0  \\ \hline
14   & 98  & 64 & 76 & 79 & 62 & 78 & 95 & 38 & 100 & 76 & 98 & 90 & 98 & 88 & 97  & 98 \\ \hline
23   & 100 & 63 & 96 & 86 & 68 & 86 & 63 & 42 & 84  & 54 & 84 & 32 & 84 & 28 & 40  & 98 \\ \hline
24   & 99  & 0  & 0  & 62 & 0  & 58 & 0  & 88 & 95  & 0  & 97 & 46 & 97 & 57 & 99  & 99 \\ \hline
34   & 97  & 0  & 0  & 0  & 0  & 0  & 0  & 0  & 94  & 0  & 94 & 0  & 94 & 0  & 68  & 97 \\ \hline
123  & 100 & 96 & 96 & 96 & 88 & 96 & 96 & 0  & 100 & 0  & 96 & 0  & 96 & 0  & 0   & 96 \\ \hline
124  & 98  & 38 & 90 & 79 & 88 & 78 & 76 & 64 & 97  & 95 & 98 & 76 & 98 & 62 & 100 & 98 \\ \hline
134  & 99  & 88 & 46 & 62 & 57 & 58 & 0  & 0  & 99  & 0  & 97 & 0  & 97 & 0  & 95  & 99 \\ \hline
234  & 97  & 0  & 0  & 0  & 0  & 0  & 0  & 0  & 68  & 0  & 94 & 0  & 94 & 0  & 94  & 97 \\ \hline
1234 & 50  & 0  & 0  & 0  & 0  & 0  & 0  & 0  & 0   & 0  & 50 & 0  & 50 & 0  & 50  & 50 \\ \hline
\end{tabular}
}
\caption{Spectrum of different totalistic rules considering a sample of fixed points in a $4x4$ grid. Each entry of the matrix represent the percent frequency of each boolean gate (enumerated from $0$ to $15$) related to total amount of fixed points}
\label{tab:spectrumfixedpoint}
\end{table}

\begin{figure}[!tbp]
	\centering
	\includegraphics[scale=0.5]{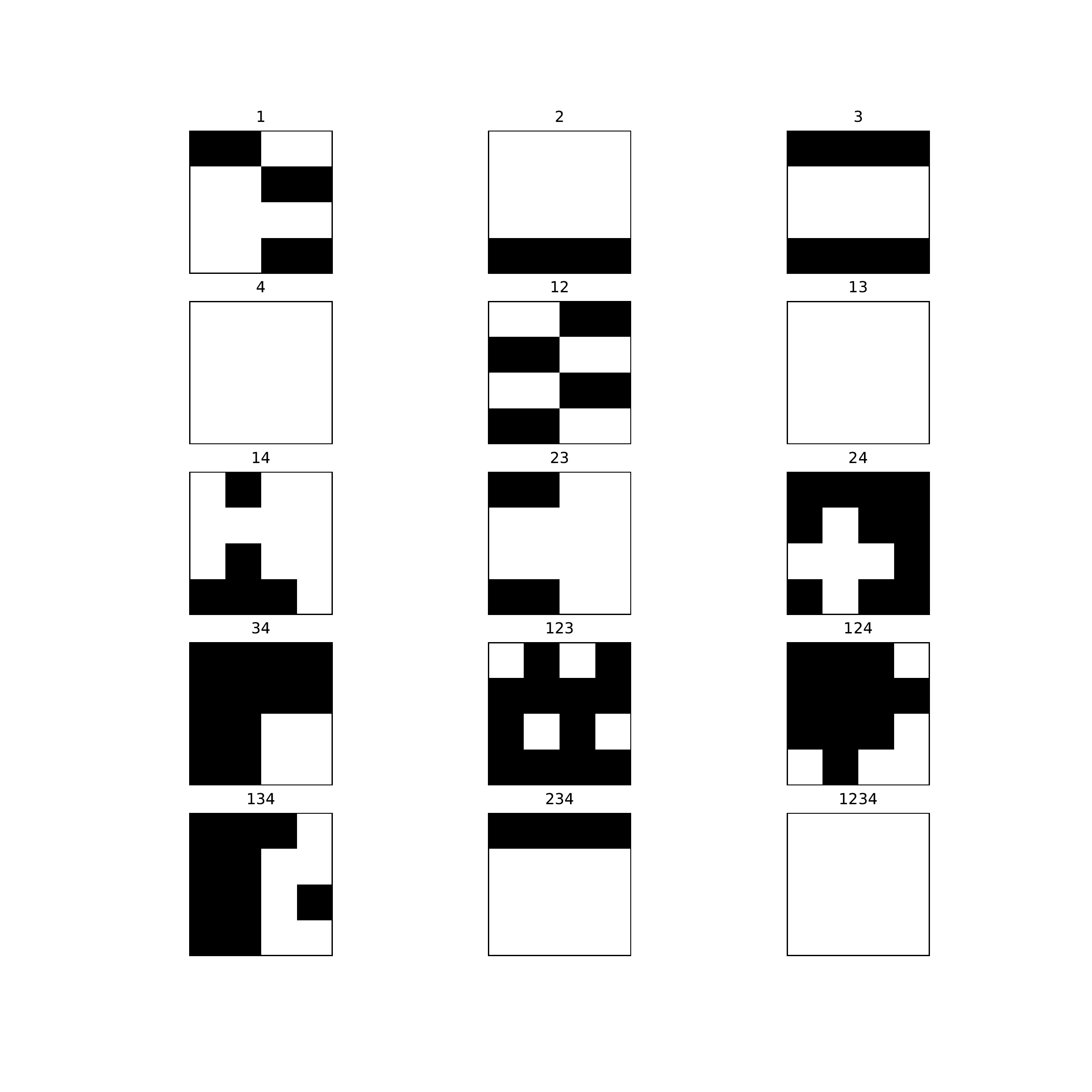}
	\caption{Fixed points exhibiting the greatest spectrum of boolean gates for each totalistic rule.}
	\label{plot:fixedpoint2D}
\end{figure}

\paragraph{Uniformity and stability.}
Roughly, as it is shown in  Figure \ref{plot:fixedpoint2D} we can identify two types of fixed points according to two criteria: uniformity and stability. In the first group we observe those rules which have fixed point $\vec{0}$ as the one who has greater spectrum of simulated Boolean gates. As it is shown in Figure  \ref{plot:fixedpoint2D} this is the case of rule $4,1234$ and $13$. Note that one can easily deduce that rules $4$ and $1234$ only have $\vec{1}$ and $\vec{0}$ as fixed points. In addition rule $13$ is the XOR rule. We have already shown in Theorem \ref{teo:linear} that matrix-defined rules can only generate other matrix-defined rules on every possible subset of inputs defined by fixing some of the input variables. As a consequence of the latter result, in this case,  we can only expect the XOR function (rule $13$) to generate the XOR gate (gate $6$), its complementary gate by conjugation (gate $9$), constant gates ($0$ and $15$), projection gates ($10$ and $12$) and their conjugated gates ($3$ and $5$ respectively). Nevertheless, we observe in Table \ref{tab:nfixedpoints} a that $\vec{0}$ is the only fixed point of $13$ and one can easy show this in general for a $4 \times 4$ grid. In order to eliminate any misleading effect of the small dimension chosen for the experiment, we exhibit a fixed point for this rule in a $6\times 6$ grid in Figure \ref{plot:13fix}. For this particular fixed point, we have repeated latter experiment and, and we have obtain as result that it is capable to implement any of the previously described gates that XOR function can simulate.
\begin{figure}[!tbp]
	\centering
	\includegraphics[scale=0.25]{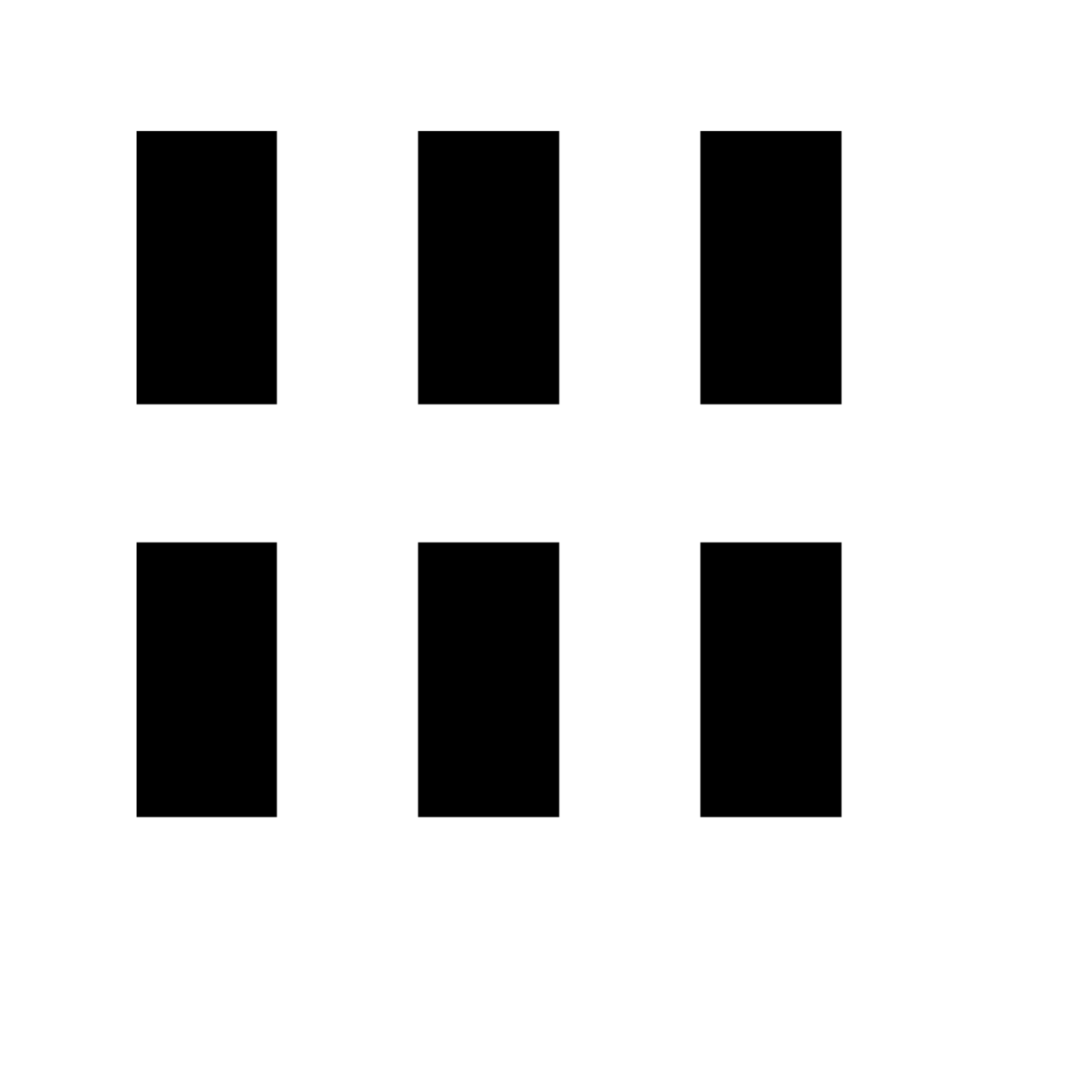}
	\caption{Fixed point for Rule $13$ in a $6 \times 6$ grid.}
	\label{plot:13fix}
\end{figure}
On the other hand, some of the non-uniform rules, i.e. those who have a fixed point different from $\vec{1}$ or $\vec{0}$ contributing the most to its spectrum, have actually complete spectrum (they can implement any possible Boolean gate). Then, ignoring the obvious misleading effects produced by taking a small grid as a interaction graph, we can roughly deduce that uniformity by itself plays a role on how complex is the spectrum of some given rule. Nevertheless, as the results on spectrum of non-uniform rules suggest, it is not the only element that seems to impact in the richness of the spectrum of some rules. In fact, taking the example of rule $14$ and $4$ (the first produces a complete spectrum while the second one only produce a reduced amount of Boolean gates), we conjecture that diffusion properties induced by having $1$ in the active set might play an interesting role explaining the difference between the spectrum of both rules. We observe this effect also in the other way: rule $23$ has complete spectrum but rule $123$ doesn't.

In addition, according to Table \ref{tab:spectrumfixedpoint} we can see that there are only three rules exhibiting only $2$ or $1$ fixed point.

On the other hand, we observe that there are fixed points in which, if we change only two cells, roughly, we do not produce any different dynamical behaviour compared to the initial condition. More precisely, dynamics tend to the original fixed point.  It is interesting to observe that most of the rules having greatest spectrum exhibit a non-uniform unstable fixed point such us rules $1,23,123,2$. However, here we can also roughly identify two types of fixed points: those who are uniformly unstable, that is to say, they produce different dynamical behaviour independently from the position of he cell we perturb. This is the case of rule $1$ and $12$ according to igure  \ref{plot:fixedpoint2D}. Contrarily, some rules (such as rule) $34$, $24$ and $124$ have a fixed point in which there are some stable areas and unstable areas. Again, the effect of having this type of distribution of $1$s and $0$s in the fixed points of some rules might play some role in the richness of their spectrum. Nevertheless it is not possible to directly deduce that from the results as there some rules such as rule $1$ having an uniformly unstable fixed point and having at the same time a considerable amount of richness in its spectrum while $124$ have some stable zones and produces a complete spectrum.

\subsection{Gadget dynamics: a rule $1$ automata network example.}
In this section we take a deeper look in the structures obtained for rule $1$ in latter simulations. We study the dynamics of a gadget obtained in the latter simulation which is capable of calculating different logic gates by changing the assignation of inputs and output. This gadget exhibits the maximum value of the spectrum for the rule $1$ in the simulations for $p=0.1$ (i.e. it is able to simulate the maximum amount of different logic gates observed in the simulations). In addition, this gadget has only $6$ nodes as it is shown in Figure \ref{plot:gadget} \\
\begin{figure}[!tbp]
	\centering
	\includegraphics[scale=0.5]{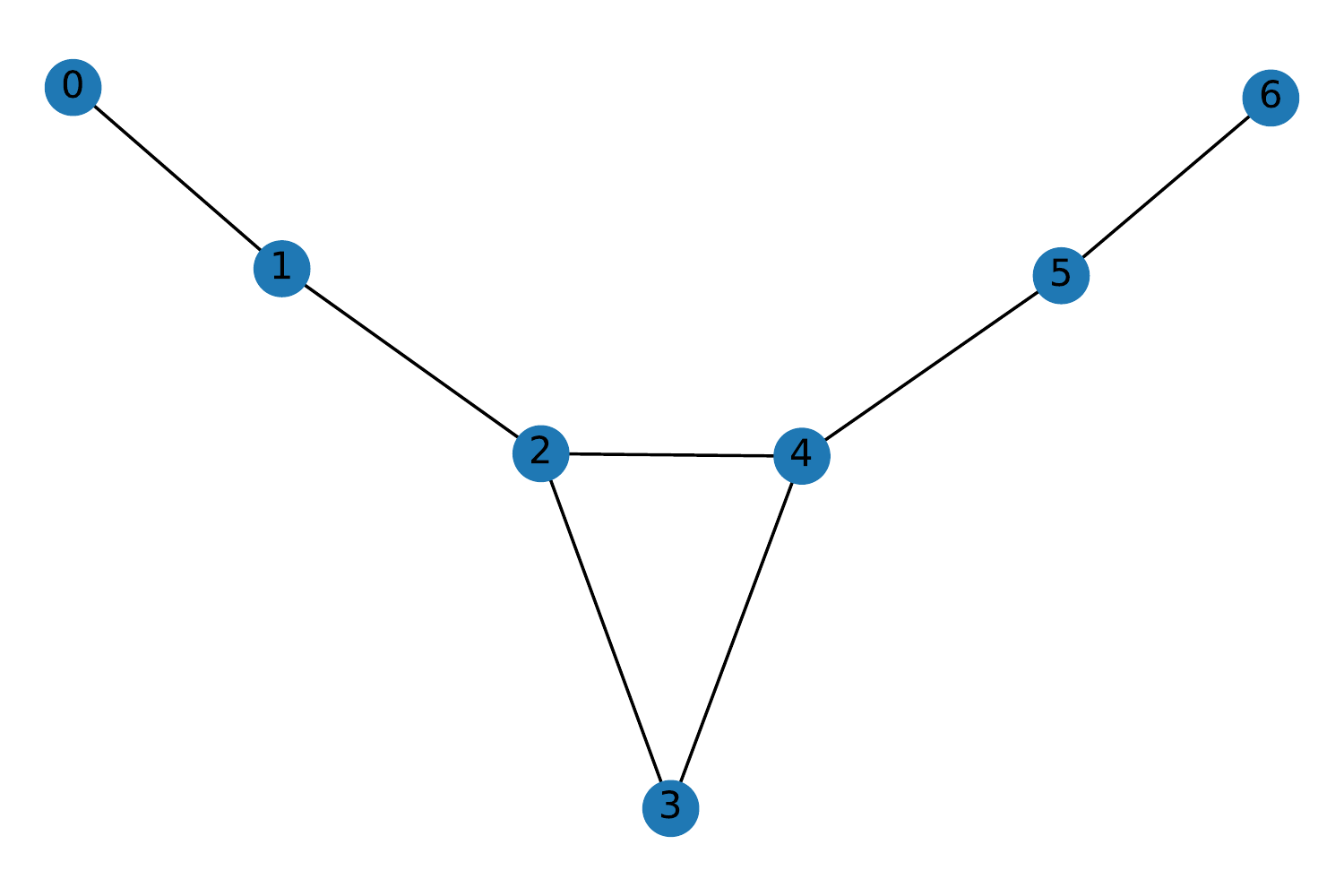}
	\caption{Gadget obtained for rule $1$ in the previous simulations with $p=0.1$ (see Figure \ref{plot:p01})}
	\label{plot:gadget}
\end{figure}

This section is organised in the following way: in the first subsection we study the dependency of the spectrum in the simulation time, showing that there is a critical time in which the gadget starts to simulate some gates such as AND and OR gates. In the second subsection, we explore the dynamics behind the simulation of the AND gate.

\subsubsection{Frequency of simulated gates v/s simulation time}

We now study the dependency of the spectrum of the latter gadget automata network that we have found for rule $1$ on the simulation time. Remember that the frequency here is a measure of how many different assignation of two input and one output produces the same logic gate. In this case, we focus in the studying the AND and OR gates. In Figure \ref{plot:gadget} we can see that the frequency for these two logic gates v/s simulation time. As we can observe in Figure \ref{plot:gadget}, there is no monotonic behaviour on the frequency. More over it appears to have a periodic behaviour. We conjecture that it is related to the maximal period of some attractor of the gadget automata network. 
\begin{figure}[!tbp]
	\centering
	\includegraphics[scale=0.25]{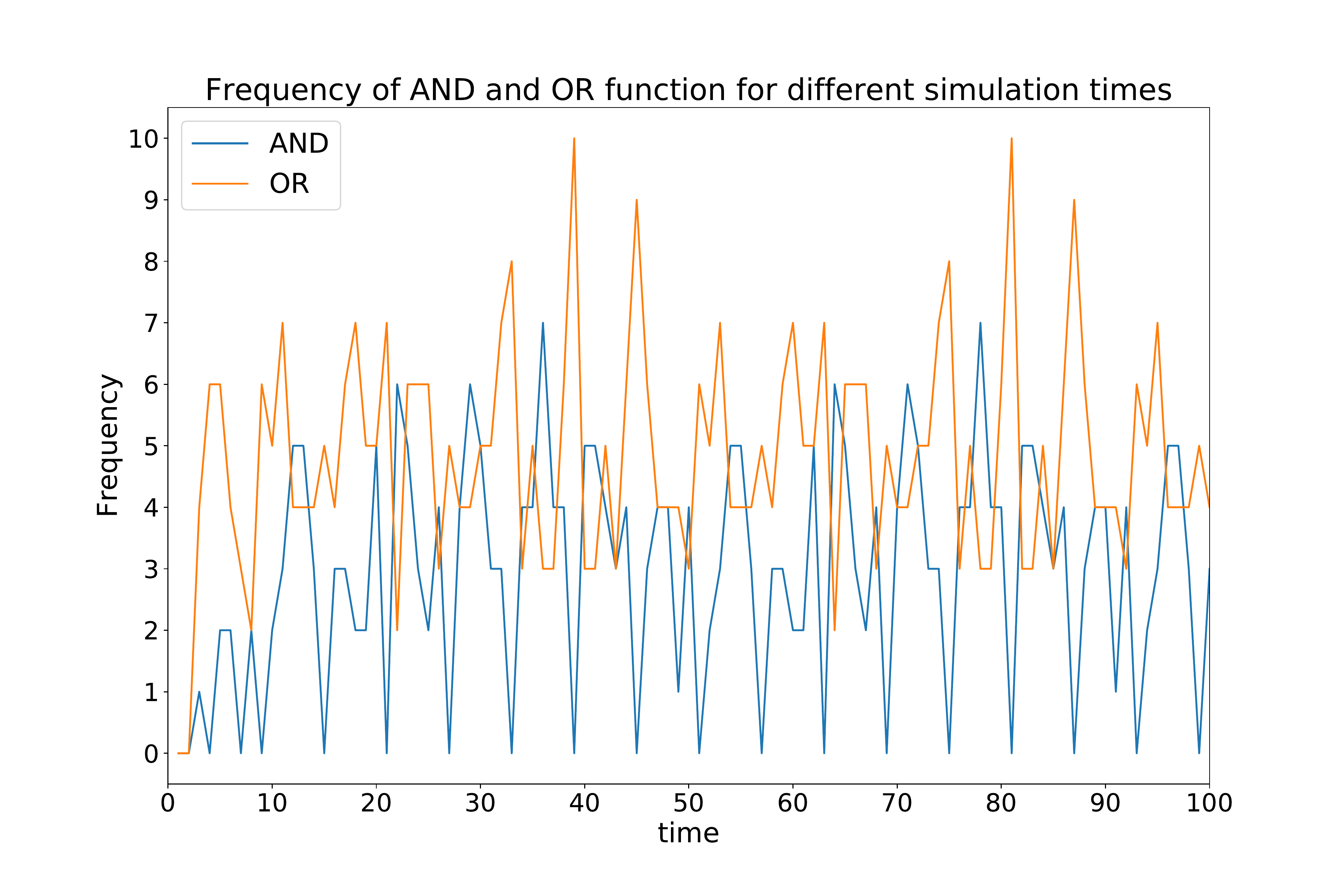}
	\caption{Frequency of AND and OR gates v/s simulation time for the gadget in Figure \ref{plot:gadget}}
	\label{plot:exampleANDt1}
\end{figure}

On the other hand, we can observe that a critical time is necessary for the gadget in order to simulate AND and OR gates. More precisely, the system is capable of simulating both logic gates starting from $t = 3$ as it is shown in Figure \ref{plot:exampleANDt2}. Finally, note that the time steps in which the system is capable of simulating AND gates are not necessarily the same for simulating OR gates. For example, as it is shown in Figure \ref{plot:exampleANDt2}, the gadget is not capable of simulating OR gates in time $t=4$ but it can simulate AND gates with $6$ different assignations of input and output at this same time step.
\begin{figure}[!tbp]
	\centering
	\includegraphics[scale=0.25]{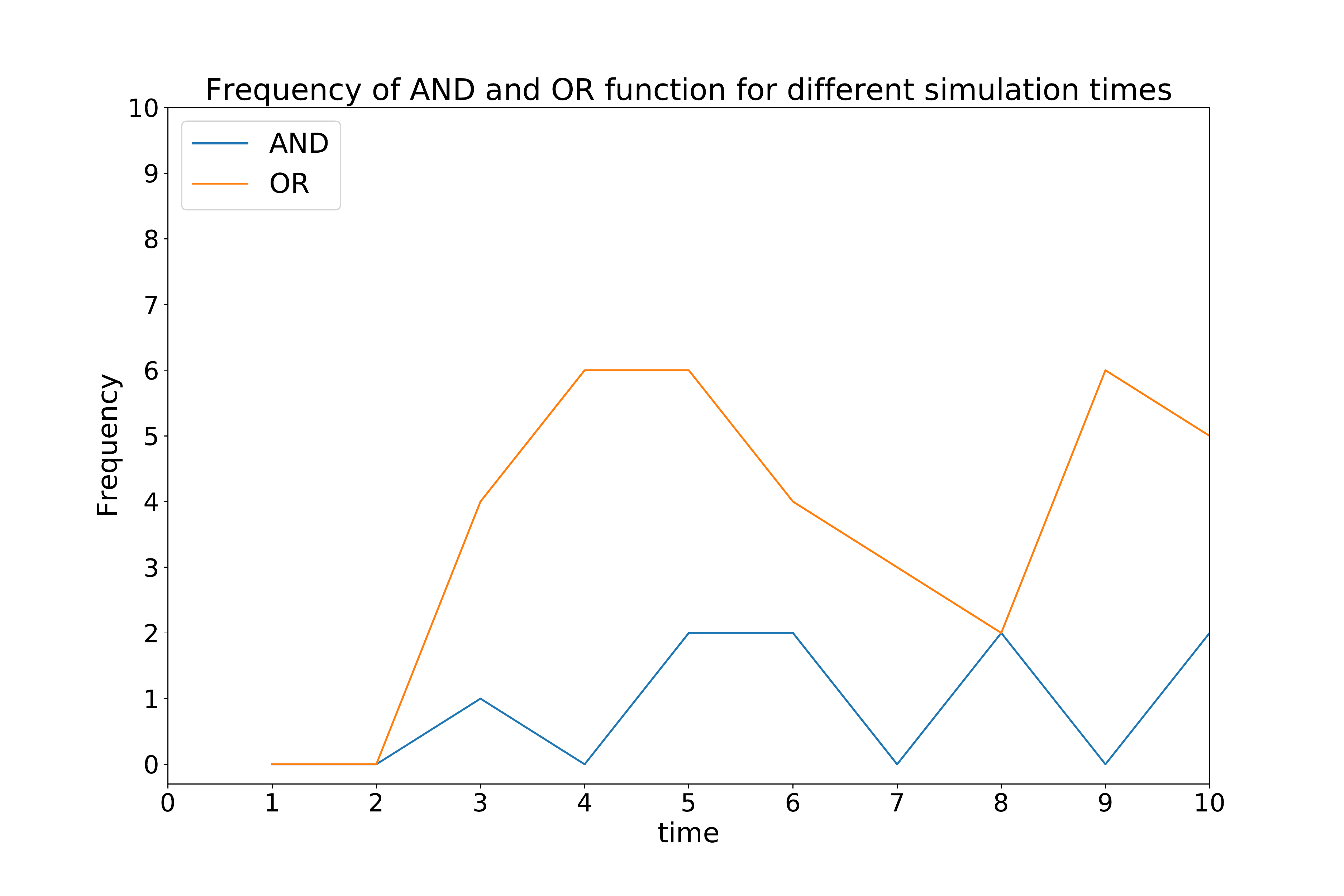}
	\caption{Frequency of AND and OR gates v/s simulation time for the gadget in Figure \ref{plot:gadget}). The system is capable of simulating both gates starting from $t=3$.}
	\label{plot:exampleANDt2}
\end{figure}
\subsubsection{AND gate dynamics}
In the section we study in the detail the dynamics behind the simulation of an AND gate by the gadget automata network which graph is shown in Figure \ref{plot:gadget}. We choose the minimum simulation time required for the system to compute an AND gate. As it is shown in Figure \ref{plot:exampleANDt2} this time step is $t=3$. In Figure \ref{plot:exampleAND1} we show the dynamics of the network starting from the input $11$. In this case, nodes marked as $2$ and $3$ are the inputs and $4$ is the output. Active nodes are coloured in yellow and inactive nodes are coloured blue. Note that AND gate is computed after $t=3$ time steps. \\
\begin{figure}[!tbp]
	\centering
	\includegraphics[scale=0.25]{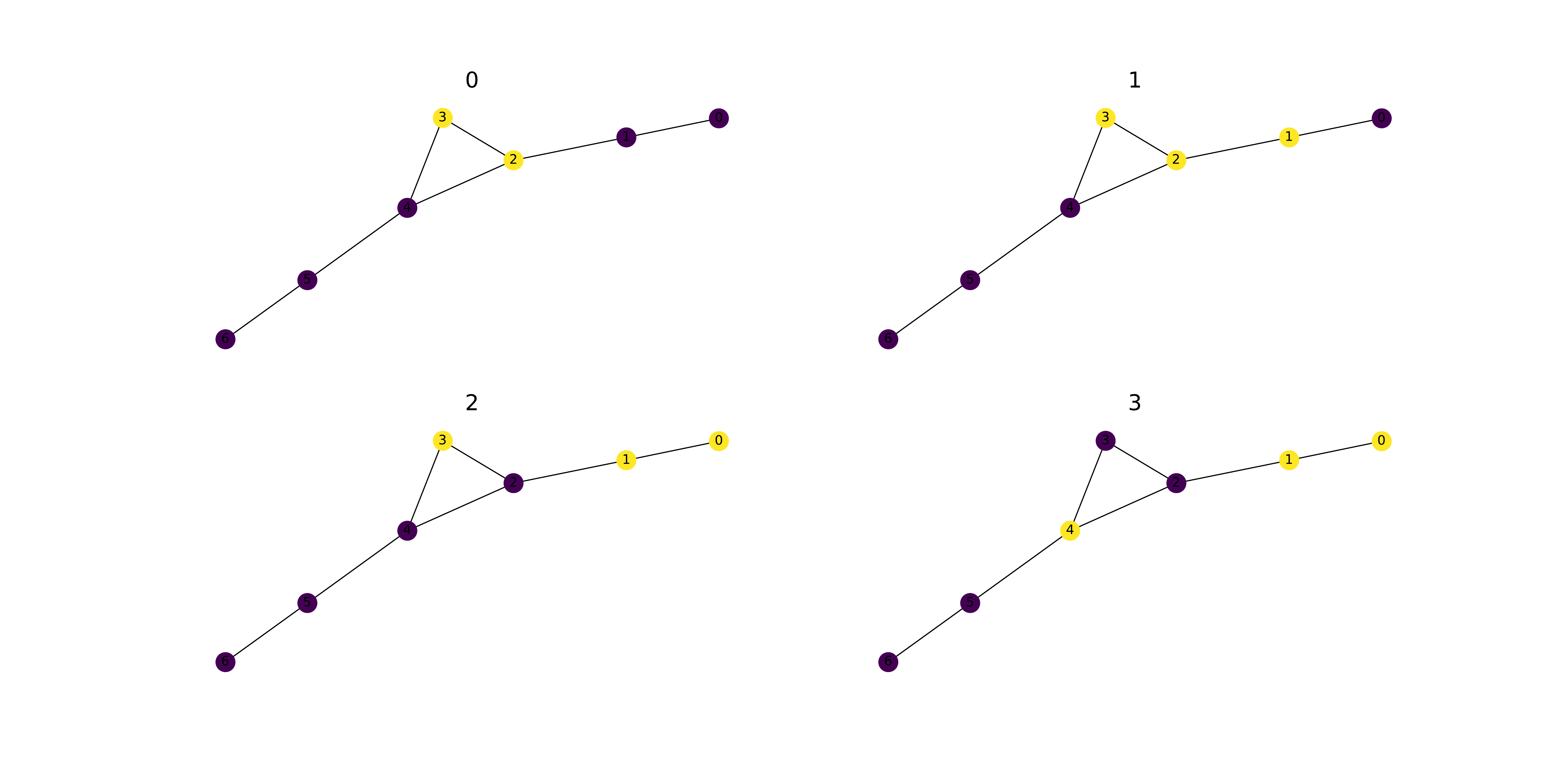}
	\caption{Dynamics behind simulation of an AND gate for an input $11$. Nodes $2$ and $3$ act as input nodes and node $4$ operates as the output. }
	\label{plot:exampleAND1}
\end{figure}
On the other hand, Figure \ref{plot:exampleAND2} shows the evolution of the input $00$ using the same input-out assignment. We see here how the paths that are connected to the central triangle of the graph play an essential role in blocking a single $1$ signal. \\
\begin{figure}[!tbp]
	\centering
	\includegraphics[scale=0.25]{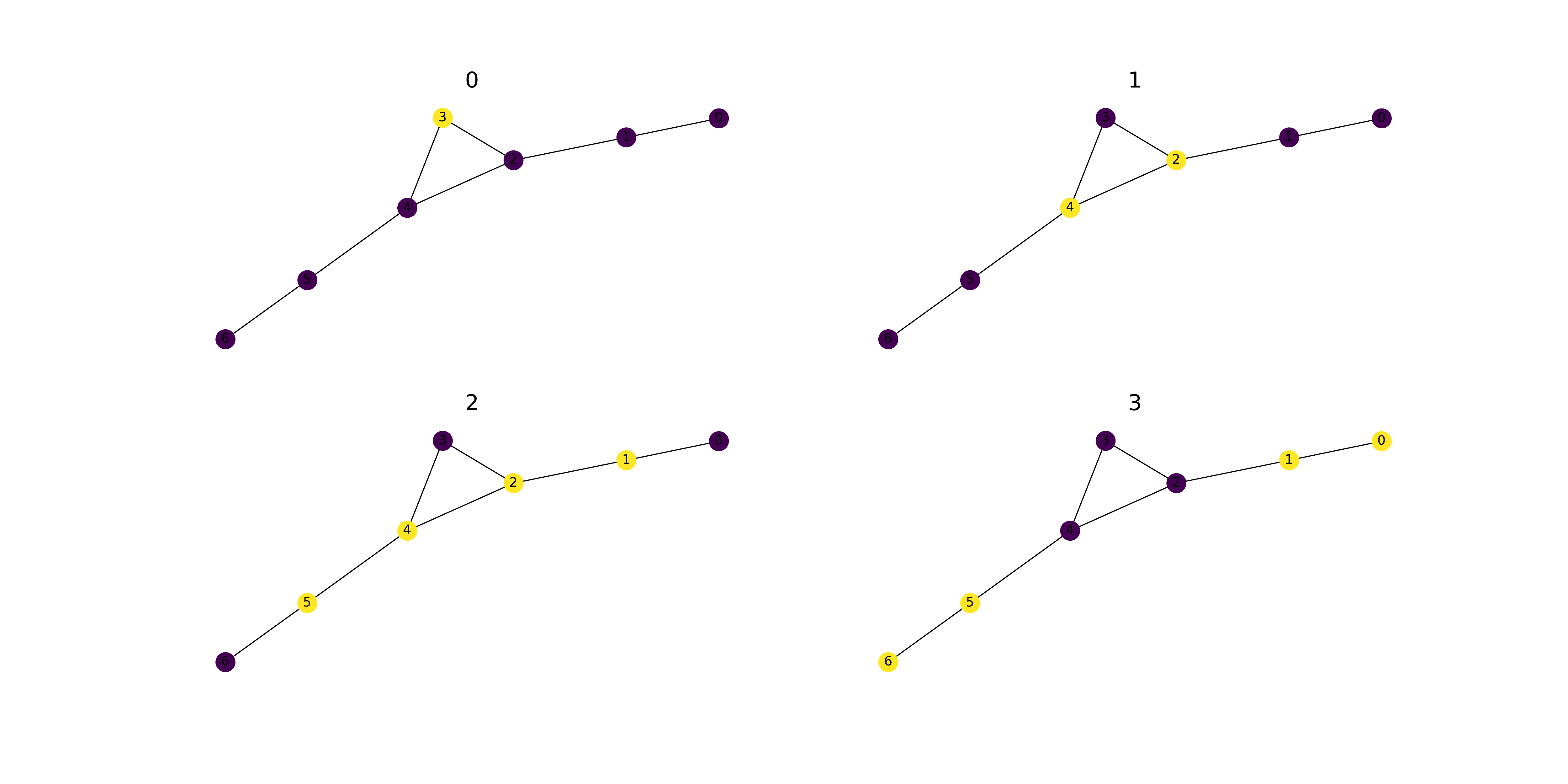}
	\caption{Dynamics behind simulation of  AND gate for an input $01$. Nodes $2$ and $3$ act as input nodes and node $4$ operates as the output. }
	\label{plot:exampleAND2}
\end{figure}
Contrarily to the case of the gadgets that we have theoretically proposed in the latter sections, we observe that, in both Figure \ref{plot:exampleAND1} and \ref{plot:exampleAND2}, the simulation of the AND gate takes more than a straightforward sequence of calculations performed in some linear way. We conjecture that the system uses different attractors in order to simulate the output of one logic gate for different input assignments.
\subsection{Rule $1$ Bull Cellular Automata}
In this section, we introduce a non-uniform one dimensional cellular automata based in the gadget shown in Figure \ref{plot:gadget}. Considering the topology of latter gadget we introduce a cellular automaton with uniform radius $1$ (as in the definition of elementary cellular automata) but we modify the neighbourhood of certain cells in order to allow them to consider not only the value of their adjacent neighbours but to consider the values of the cells at distance $2$ to the right (right radius equal $2$) and also, to mantain symmetry, to the left. More precisely, we define this cellular automaton as a function $F:\{0,1,2,3,4,5\}^n: \to \{0,1,2,3,4,5\}^n.$ defined by:
$$F(x)_i\begin{cases}
1 & \text{ if } \overline{x}_{i+1} + \overline{x}_{i-1} = 1 \wedge x_i \in \{0,1\} \\
0 & \text{ if } \overline{x}_{i+1} + \overline{x}_{i-1} \not = 1 \wedge \overline{x}_i \in \{0,1\} \\
3 & \text{ if } \overline{x}_{i+2} + \overline{x}_{i+1} + \overline{x}_{i-1}  = 1 \wedge \overline{x}_i \in \{2,3\} \\
2 & \text{ if } \overline{x}_{i+2} + \overline{x}_{i+1} + \overline{x}_{i-1} \not = 1 \wedge \overline{x}_i \in \{2,3\} \\
5 & \text{ if } \overline{x}_{i-2} + \overline{x}_{i+1} + \overline{x}_{i-1} = 1 \wedge \overline{x}_i \in \{4,5\} \\
4 & \text{ if } \overline{x}_{i-2} + \overline{x}_{i+1} + \overline{x}_{i-1} \not = 1 \wedge \overline{x}_i \in \{4,5\} 
\end{cases}$$
where $\overline{x}_i = \begin{cases} 
1 & \text{ if } x \in \{1,3,5\} \\
0 & \sim 
\end{cases}$   for $i = 1, \hdots, n$ and taking $n+k$ and $0-k$ modulo $n$ for any positive integer $k$ and where $n$ is some positive integer. In simple words, we are non uniformly repeating several copies of gadget in Figure \ref{plot:gadget}. Based on our lasts results, our main aim here is to study simulation capabilities of this system. We remark that, for example, it is not clear wheter we can simulate an evaluation of an arbitrary boolen circuit in two dimensional rule $1$ cellular automata. We remark also that in the case in which there are no cell with special neighbourhood (so the cellular automaton is an elementary cellular automaton) the function is the well known rule $90$ which is a class $4$ rule in Wolfram classification \cite{wolfram1984universality}.

In order to illustrate the dynamics of the system, we show a simulation starting from a random initial condition in Figure \ref{fig:dynbull}. In the latter simulation, we have generated randomly a few bulls in three different areas of length $8$ cells (at the begining, in the middle and at the end of the ring) and the rest of the states of the rest of the cells were randomly generated with uniform probability.
\begin{figure}[!tbp]
	\centering
	\includegraphics[scale=0.6]{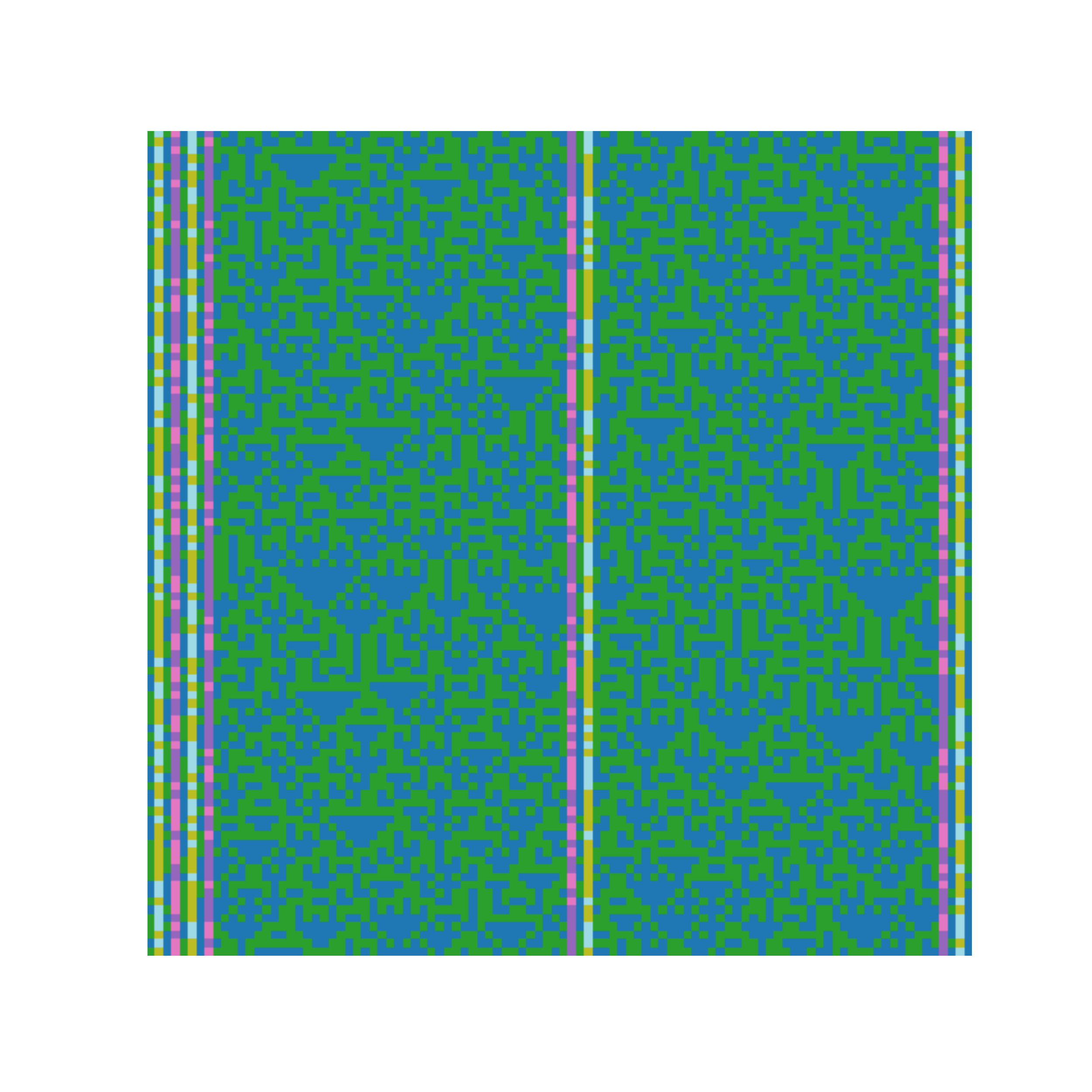}
	\caption{Dynamics of Rule $1$ Bull cellular automaton starting from a random initial condition.  Green cells are regular cells in state $1$; blue cells are regular cells in state $0$; purple cells are right-modified neighbourhood cells in state $0$; pink cells are modified right-neighbourhood cells in state $1$; yellow cells are left-modified neighbourhood cells in state $0$ and light blue cells are modified left-neighbourhood cells in state $1$ }
	\label{fig:dynbull}
\end{figure}
In this subsection, we have repeated the same experiments we did for totalistic two dimensional cellular automata in previous sections, i.e., we start by searching for a set of fixed points by simulating for small values of $n$ ($n=12$) the dynamics of the system starting from any initial condition. Note that in this case, the state values also code the type of neighbourhood of the cell so, we are also taking into account any possible assignation for different types of neighbourhoods in any cell. More precisely, given a initial condition we simulate the system for $t=100$ steps and we verify if the system has reach a fixed point. Then, we search for boolean gates by perturbing this sample of fixed points in the same way we did in previous sections. We show the fixed point which attains the great spectrum (the one which simulates the greatest amount of different boolean gates) and we exhibit the its dynamics in order to understand how it simulate certain boolean gates.

\subsubsection{Simulation set-up}
We start by describing the parameters of the following simulations:

\begin{enumerate}
	\item Number of cells ($n$) : 12
	\item Simulation time $t$: 100	
	\item Number of fixed points considered: 608
\end{enumerate}
\subsubsection{Results}
\begin{table}[]
\resizebox{\textwidth}{!}{%
\begin{tabular}{|l|l|l|l|l|l|l|l|l|l|l|l|l|l|l|l|l|}
\hline
      & 0   & 1  & 2  & 3  & 4  & 5  & 6  & 7  & 8  & 9  & 10  & 11 & 12 & 13 & 14 & 15 \\ \hline
1Bull & 100 & 67 & 96 & 70 & 98 & 71 & 98 & 69 & 98 & 67 & 100 & 69 & 99 & 69 & 96 & 71 \\ \hline
\end{tabular}%
}
\caption{Spectrum for rule $1$ bull cellular automaton. Numbers are percent frequency of each boolean gate related to the total amount of gates simulated by a sample of 608 fixed points obtained by simulating the system starting from different initial conditions.}
\label{tab:spectrumbull}
\end{table}
\begin{figure}[H]
	\centering
	\includegraphics[scale=0.55]{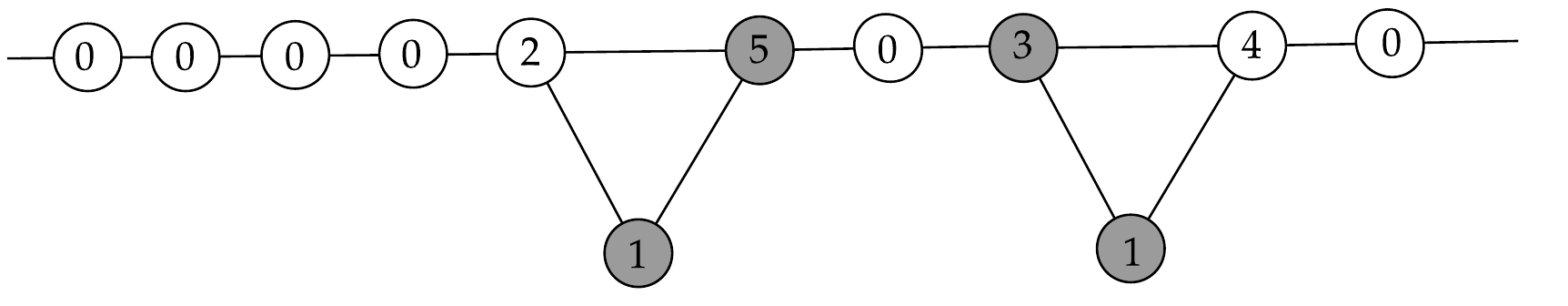}
	\caption{Fixed point for rule 1 bull cellular automaton with the greatest spectrum value. Gray nodes are considered active and white nodes inactive. Coded state is shown inside each node.}
	\label{fig:fpbull}
\end{figure}
\paragraph{Non-uniformity produces all possible boolean network}
From Table \ref{tab:spectrumbull} we deduce that system can simulate all two input-one output boolean gates. That is very interesting as it confirms our initial insight from previous result on gadget shown in Figure \ref{plot:gadget}. In order to precise what role plays these modified neighbourhoods in terms of how many of them we need to simulate a large amount of boolean gates, we show in Figure \ref{fig:fpbull} the most representative fixed point of the sample needs roughly $33\%$ of cells with a extended neighborhood (two bull graphs) to produce the greatest spectrum. Moreover, we have analysed its particular spectrum of generated boolean gates and we observe it can simulate all of them.

Finally, we are interested in exhibit the dynamical behaviour of the sistem starting from different perturbations of this fixed point in some given input-output assignation.  This experiment will simulate a NAND gate. As it is shown in Figure \ref{fig:NANDBull}, we fixed the first and the fourth cells as a inputs (they are marked with red dots in the figure) and the fifth cell as output. 
 We studied the whole dynamics for all the simulation time $t=100$ however, we have found that, actually, at different time steps system is capable of simulate a NAND gate. In fact, one can deduce from in Figure \ref{fig:NANDBull} (which shows $t=27$ time steps) , system is exhibiting a periodic behaviour for perturbations $(1,0)$ and $(0,1)$ and it remains in a fixed point for the rest of them. We observe that depending of the period of both attractors we can periodically see the simulation of the NAND gate for different simulation times. 
\begin{figure}
	\centering
	\includegraphics[scale=0.5]{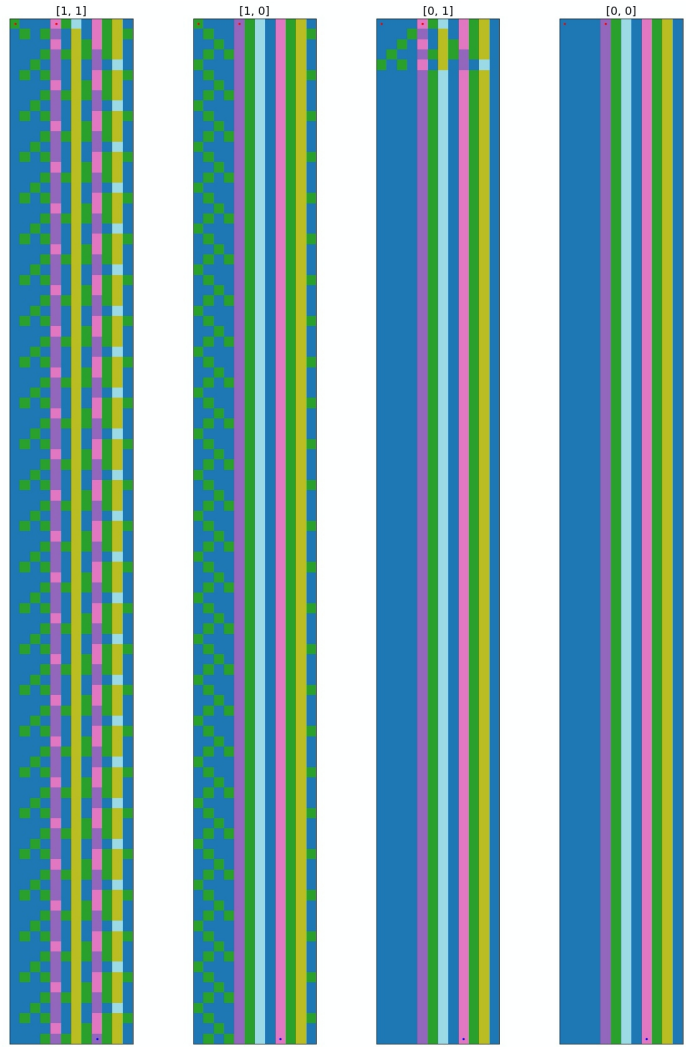}
	\caption{Dynamics of Rule $1$ Bull cellular automaton starting from fixed point shown in Figure \ref{fig:fpbull}. Title of plots indicates a perturbation of inputs. Inputs are marked by red dots and output is marked by a blue dot. Green cells are regular cells in state $1$; blue cells are regular cells in state $0$;  purple cells are right-modified neighbourhood cells in state $0$; pink cells are modified right-neighbourhood cells in state $1$; yellow cells are left-modified neighbourhood cells in state $0$ and light blue cells are modified left-neighbourhood cells in state $1$  }
	\label{fig:NANDBull}
\end{figure}

\begin{figure}[!tbp]
    \centering
    \includegraphics[scale=0.6]{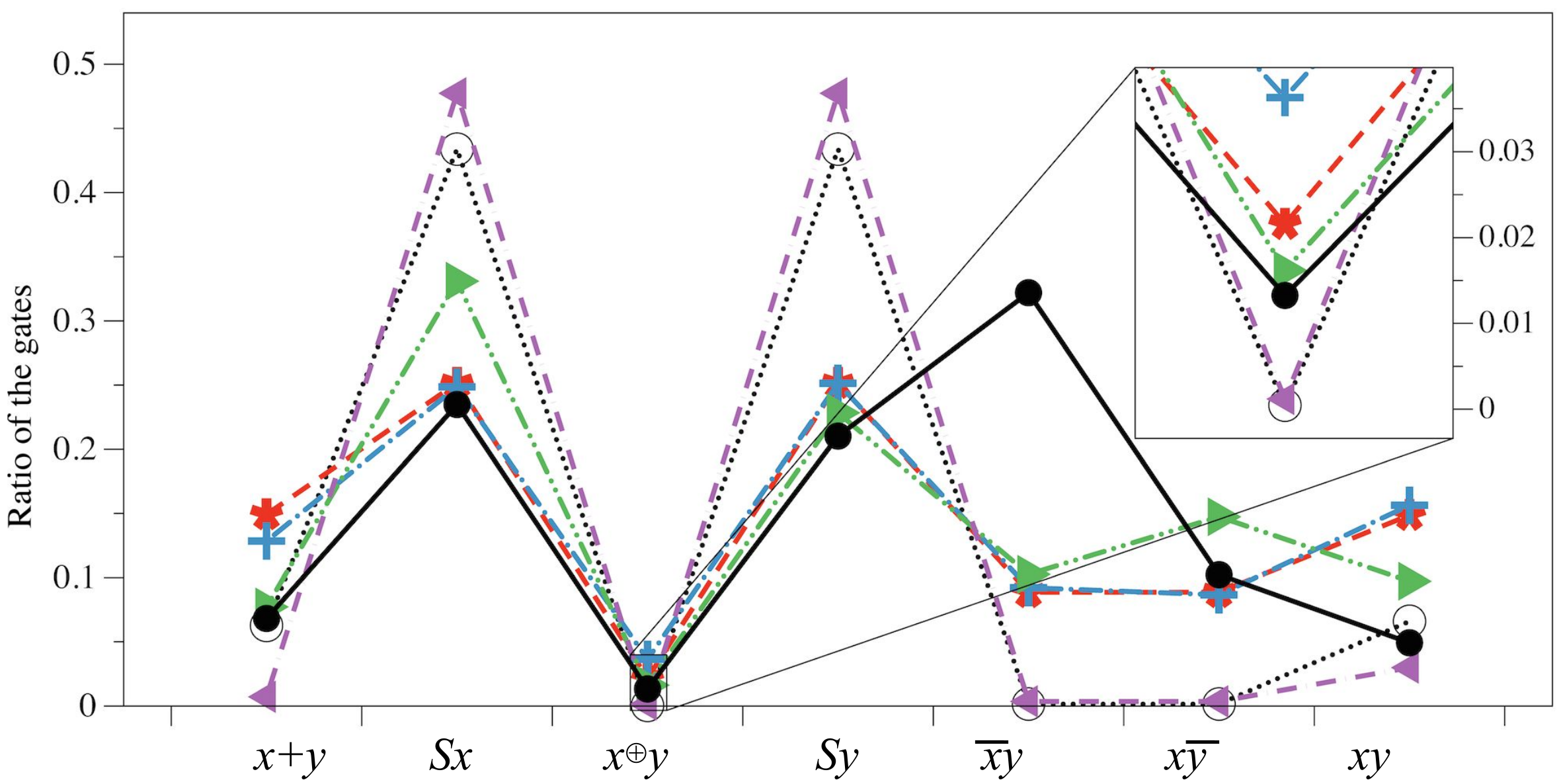}
    \caption{Comparative ratios of Boolean gates discovered in 
    mycelium network of real fungal colony~\cite{adamatzky2020booleanFungal}, black disc and solid line;
    slime mould \emph{Physarum polycephalum}~\cite{harding2018discovering}, black circle and dotted line;
    succulent plant~\cite{adamatzky2018computers}, red snowflake and dashed line; 
    single molecule of protein verotoxin~\cite{adamatzky2017computing}, light blue `+' and dash-dot line;
    actin bundles network~\cite{adamatzky2019computing}, green triangle pointing right and dash-dot-dot line;
    actin monomer~\cite{adamatzky2017logical}, magenta triangle pointing left and dashed line. Area of {\sc xor} gate is magnified in the insert. Lines are to guide eye only.}
    \label{fig:gatesDistribution}
\end{figure}

In numerical experiments we calculated spectra of Boolean gates implementable by totalistic automata network with various probabilities of connectivity (Tabs.~\ref{tab:p01}, \ref{tab:p05}, \ref{tab:p08}). Let us compare these with our previous experimental laboratory and numerical modelling results shown in Fig.~\ref{fig:gatesDistribution}.  The ratios of the gate discovered are obtained in  experimental laboratory reservoir computing with slime mould \emph{Physarum polycephalum}~\cite{harding2018discovering}, succulent plant~\cite{adamatzky2018computers} and numerical modelling experiments on computing with 
protein verotoxin~\cite{adamatzky2017computing},actin bundles network~\cite{adamatzky2019computing}, actin monomer~\cite{adamatzky2017logical}, and fungal colony~\cite{adamatzky2020boolean}. The distributions are quite similar. The gates selecting one of the inputs are in majority, followed by {\sc or} gate, {\sc not-and} an {\sc and-not} gates. The gate {\sc and} is usually underrepresented in experimental and modelling experiments. The gate {\sc xor} is a rare find. The similarity of the distributions give us a hint on universality of all types of biological networks and the totalistic automata network echoing them.

 The fact that only automata networks governed by totalistic rules with at least one isolated interval can generate any Boolean function echoes our previous numerical experiments on Boolean gates in protein molecules~\cite{adamatzky2017computing,adamatzky2017logical}: a variety of logical gates realisable on excitable automata networks, where functions rely on excitation interval is much higher than that of the excitable automata networks with threshold excitation functions. 
 
\section{Discussion}

In this article we have studied theoretically and computationally the complexity of totalistic networks, according to their ability to simulate different Boolean functions (or, equivalently, Boolean circuits). It was shown that using networks of isolated and interval totalistic rules any Boolean function is generated. Using threshold networks it is possible to construct any monotonic function and with those defined by a matrix (linear: with local disjunction or the XOR) the generation of Boolean functions is poor: only constants or those of their own type (OR or XOR ). Furthermore, by means of computational experiments, the spectrum of Boolean functions generated by totalistic networks on random graphs and in a two-dimensional grid was established.

It is important to point out that for the generation of Boolean functions not only the totalistic rule and the considered graph are important, but also the fixed points. In this sense, there is a relationship between the totalistic function and the balance of active (1) and inactive (0) states of the fixed point considered. For example, if the local function reaches the state 1 with at least one active state, then the fixed point may have very few nodes in active state and still its disturbance can produce a wide spectrum of Boolean functions. In particular, for this type of functions we established that the fixed point $\vec{0}$ is enough to determine a complete spectrum of Boolean functions. However, if the local totalistic rule requires two or more active states to be activated, then the fixed point $\vec{0}$ is not a sufficient support to generate an interesting set of Boolean functions and, if they exist, fixed points with active states should be considered. In this perspective, a future work could be the study of the relationship between the perturbations of fixed points, the dynamics generated (how many nodes of the network change state before reaching a steady state again) and the Boolean functions obtained. In some way, this problem seems to be related to the amplitude of the “avalanche” produced by the disturbance (cells that change their value) and the latter could be studied from the point of view of the sand pile model proposed by Per Bak and the self-organized criticality paradigm \cite{bak1988self}. 

On the other hand, in this work we have associated Boolean functions with circuits implemented in the network, but apart from proving that it is possible to generate the universal base functions (AND, OR, NAND, XOR) that allow us to construct arbitrary Boolean functions. However, this is a result, we would say, extreme: we are the ones who, based on the local functions considered, build these gadgets, so the network. In the biological problem, the network is given. Therefore, we also did numerical experiments to observe the type of functions that appear according to the type of totalistic function and the topology of the network. This led us to looking at the problem from a different angle: searching for “wild” Boolean functions, that is, setting a totalistic function and discovering subgraphs with active and inactive states that emulate interesting Boolean functions. Specifically, we study one of them, the bull gadget, inserting it into a one-dimensional totalistic automaton. Clearly, this path: the search for “wild gadgets” and its applications, is a promising line of future work.

 \section*{Acknowledgement}

AA has received funding from the European Union's Horizon 2020 research and innovation programme FET OPEN ``Challenging current thinking'' under grant agreement No 858132. EG residency in UWE has been supported by funding from the Leverhulme Trust under the Visiting Research Professorship grant VP2-2018-001 and  from the project the project 1200006, FONDECYT-Chile. MRW has recieved funding from ANID via PFCHA/DOCTORADO NACIONAL/2018 – 21180910 and was also supported by PIA AFB 170001.

\bibliographystyle{plain}
\bibliography{biblio}

\begin{thebibliography}{10}

\bibitem{adamatzky2016advances}
Andrew Adamatzky, editor.
\newblock {\em Advances in Unconventional Computing}.
\newblock Springer, 2016.

\bibitem{adamatzky2017computing}
Andrew Adamatzky.
\newblock Computing in verotoxin.
\newblock {\em ChemPhysChem}, 18(13):1822--1830, 2017.

\bibitem{adamatzky2017logical}
Andrew Adamatzky.
\newblock Logical gates in actin monomer.
\newblock {\em Scientific reports}, 7(1):1--14, 2017.

\bibitem{adamatzky2019plant}
Andrew Adamatzky.
\newblock Plant leaf computing.
\newblock {\em Biosystems}, 182:59--64, 2019.

\bibitem{adamatzky2018computers}
Andrew Adamatzky, Simon Harding, Victor Erokhin, Richard Mayne, Nina Gizzie,
  Frantisek Balu{\v{s}}ka, Stefano Mancuso, and Georgios~Ch Sirakoulis.
\newblock Computers from plants we never made: Speculations.
\newblock In {\em Inspired by nature}, pages 357--387. Springer, 2018.

\bibitem{adamatzky2019computing}
Andrew Adamatzky, Florian Huber, and J{\"o}rg Schnau{\ss}.
\newblock Computing on actin bundles network.
\newblock {\em Scientific reports}, 9(1):1--10, 2019.

\bibitem{adamatzky2020boolean}
Andrew Adamatzky, Martin Tegelaar, Han~AB Wosten, Anna~L Powell, Alexander~E
  Beasley, and Richard Mayne.
\newblock On boolean gates in fungal colony.
\newblock {\em Biosystems}, page 104138, 2020.

\bibitem{adamatzky2020booleanFungal}
Andrew Adamatzky, Martin Tegelaar, Han~AB Wosten, Anna~L Powell, Alexander~E
  Beasley, and Richard Mayne.
\newblock On boolean gates in fungal colony.
\newblock {\em Biosystems}, page 104138, 2020.

\bibitem{appali2012comparison}
Revathi Appali, Ursula Van~Rienen, and Thomas Heimburg.
\newblock A comparison of the hodgkin--huxley model and the soliton theory for
  the action potential in nerves.
\newblock In {\em Advances in Planar Lipid Bilayers and Liposomes}, volume~16,
  pages 275--299. Elsevier, 2012.

\bibitem{atrubin1965one}
AJ~Atrubin.
\newblock A one-dimensional real-time iterative multiplier.
\newblock {\em IEEE Transactions on Electronic Computers}, 3:394--399, 1965.

\bibitem{bak1988self}
Per Bak, Chao Tang, and Kurt Wiesenfeld.
\newblock Self-organized criticality.
\newblock {\em Physical review A}, 38(1):364, 1988.

\bibitem{contreras2013non}
Fidel Contreras, Fernando Ongay, Omar Pav{\'o}n, and M{\'a}ximo Aguero.
\newblock Non-topological solitons as traveling pulses along the nerve.
\newblock {\em International Journal of Modern Nonlinear Theory and
  Application}, 2013, 2013.

\bibitem{diestelgraph}
Reinhard Diestel.
\newblock Graph theory. 2005. electronic edition.

\bibitem{fichtl2016protons}
Bernhard Fichtl, Shamit Shrivastava, and Matthias~F Schneider.
\newblock Protons at the speed of sound: Predicting specific biological
  signaling from physics.
\newblock {\em Scientific reports}, 6:22874, 2016.

\bibitem{fischer1965generation}
Patrick~C Fischer.
\newblock Generation of primes by a one-dimensional real-time iterative array.
\newblock {\em Journal of the ACM (JACM)}, 12(3):388--394, 1965.

\bibitem{fromm2007electrical}
J{\"o}rg Fromm and Silke Lautner.
\newblock Electrical signals and their physiological significance in plants.
\newblock {\em Plant, cell \& environment}, 30(3):249--257, 2007.

\bibitem{gajardo2006crossing}
Anah{\'\i} Gajardo and E~Goles.
\newblock Crossing information in two-dimensional sandpiles.
\newblock {\em Theoretical Computer Science}, 369(1-3):463--469, 2006.

\bibitem{goles1996sand}
Eric Goles and Maurice Margenstern.
\newblock Sand pile as a universal computer.
\newblock {\em International Journal of Modern Physics C}, 7(02):113--122,
  1996.

\bibitem{goles1997universality}
Eric Goles and Maurice Margenstern.
\newblock Universality of the chip-firing game.
\newblock {\em Theoretical Computer Science}, 172(1-2):121--134, 1997.

\bibitem{goles2014computational}
Eric Goles and Pedro Montealegre.
\newblock Computational complexity of threshold automata networks under
  different updating schemes.
\newblock {\em Theoretical Computer Science}, 559:3--19, 2014.

\bibitem{goles2020complexity}
Eric Goles and Pedro Montealegre.
\newblock The complexity of the asynchronous prediction of the majority
  automata.
\newblock {\em Information and Computation}, page 104537, 2020.

\bibitem{goles2018complexity}
Eric Goles, Pedro Montealegre, K{\'e}vin Perrot, and Guillaume Theyssier.
\newblock On the complexity of two-dimensional signed majority cellular
  automata.
\newblock {\em Journal of Computer and System Sciences}, 91:1--32, 2018.

\bibitem{greenlaw1995limits}
Raymond Greenlaw, H~James Hoover, Walter~L Ruzzo, et~al.
\newblock {\em Limits to parallel computation: P-completeness theory}.
\newblock Oxford University Press on Demand, 1995.

\bibitem{harding2018discovering}
Simon Harding, Jan Koutn{\'\i}k, J{\'u}rgen Schmidhuber, and Andrew Adamatzky.
\newblock Discovering boolean gates in slime mould.
\newblock In {\em Inspired by Nature}, pages 323--337. Springer, 2018.

\bibitem{heimburg2005soliton}
Thomas Heimburg and Andrew~D Jackson.
\newblock On soliton propagation in biomembranes and nerves.
\newblock {\em Proceedings of the National Academy of Sciences},
  102(28):9790--9795, 2005.

\bibitem{hodgkin1952propagation}
Alan~Lloyd Hodgkin and Andrew~Fielding Huxley.
\newblock Propagation of electrical signals along giant nerve fibres.
\newblock {\em Proceedings of the Royal Society of London. Series B-Biological
  Sciences}, 140(899):177--183, 1952.

\bibitem{marr2009outer}
Carsten Marr and Marc-Thorsten H{\"u}tt.
\newblock Outer-totalistic cellular automata on graphs.
\newblock {\em Physics Letters A}, 373(5):546--549, 2009.

\bibitem{park1986soliton}
James~K Park, Kenneth Steiglitz, and William~P Thurston.
\newblock Soliton-like behavior in automata.
\newblock {\em Physica D: Nonlinear Phenomena}, 19(3):423--432, 1986.

\bibitem{shrivastava2014evidence}
Shamit Shrivastava and Matthias~F Schneider.
\newblock Evidence for two-dimensional solitary sound waves in a lipid
  controlled interface and its implications for biological signalling.
\newblock {\em Journal of The Royal Society Interface}, 11(97):20140098, 2014.

\bibitem{squier1994programmable}
Richard~K Squier and Ken Steiglitz.
\newblock Programmable parallel arithmetic in cellular automata using a
  particle model.
\newblock {\em Complex systems}, 8(5):311--324, 1994.

\bibitem{waksman1966optimum}
Abraham Waksman.
\newblock An optimum solution to the firing squad synchronization problem.
\newblock {\em Information and control}, 9(1):66--78, 1966.

\bibitem{wolfram1984universality}
Stephen Wolfram.
\newblock Universality and complexity in cellular automata.
\newblock {\em Physica D: Nonlinear Phenomena}, 10(1-2):1--35, 1984.

\end{thebibliography}

\end{document}